\documentclass[a4paper]{amsart}
\usepackage{amssymb}
\usepackage{mathrsfs}
\usepackage[enableskew]{youngtab}

\usepackage[all,cmtip]{xy}

\newtheorem{theorem}{Theorem}[section]
\newtheorem{lem}[theorem]{Lemma}
\newtheorem{cor}[theorem]{Corollary}
\newtheorem{prop}[theorem]{Proposition}
\theoremstyle{definition}
\newtheorem{defi}[theorem]{Definition}
\newtheorem{example}[theorem]{Example}

\theoremstyle{remark}
\newtheorem{rem}[theorem]{Remark}
\newtheorem{que}[theorem]{Question}

\newcommand{\BC}{\mathbb{C}}            
\newcommand{\BZ}{\mathbb{Z}}            
\newcommand{\BM}{\mathbb{M}}           
\newcommand{\BD}{\mathbb{D}}            
\newcommand{\BI}{\mathbf{i}}                 
\newcommand{\CM}{\mathcal{M}}                
\newcommand{\CMw}{\mathcal{M}_{\mathrm{w}}}                
\newcommand{\CMt}{\mathcal{M}_{\mathrm{t}}}          
\newcommand{\CR}{\mathcal{R}}               
\newcommand{\CRf}{\mathcal{R}_{\mathrm{fd}}}   
\newcommand{\CQ}{\mathcal{Q}}      
\newcommand{\CP}{\mathcal{P}}         

\newcommand{\End}{\mathrm{End}}          
\newcommand{\Id}{\mathrm{Id}}            

\newcommand{\Glie}{\mathfrak{g}}     
\newcommand{\Hlie}{\mathfrak{h}}      

\newcommand{\BQ}{\mathbf{Q}}       
\newcommand{\SW}{\mathcal{W}}         
\newcommand{\SP}{\mathscr{P}}         
\newcommand{\SB}{\mathscr{B}}         

\newcommand{\CFm}{\mathcal{F}_{\mathrm{mer}}}   
\newcommand{\Rep}{\mathbf{Rep}}    
\newcommand{\rep}{\mathbf{rep}}
\newcommand{\BGG}{\mathcal{O}}         
\newcommand{\tBGG}{\widetilde{\mathcal{O}}}    
\newcommand{\HJ}{\BGG_{\mathrm{HJ}}}
\newcommand{\BGGf}{\mathcal{O}_{\mathrm{fd}}}   
\newcommand{\CF}{\mathcal{F}}         
\newcommand{\wt}{\mathrm{wt}}       
\newcommand{\ewt}{\mathrm{wt}_{\mathrm{e}}}         
\newcommand{\qc}{\chi_{\mathrm{q}}}            

\newcommand{\CV}{\mathcal{V}}       
\newcommand{\CVf}{\mathcal{V}_{\mathrm{ft}}}   

\newcommand{\BV}{\mathbf{V}}         
\newcommand{\CW}{\mathscr{W}}     
\newcommand{\SF}{\mathscr{F}}         
\newcommand{\SG}{\mathscr{G}}     
\newcommand{\SL}{\mathscr{L}}     
\newcommand{\SD}{\mathscr{D}} 

\newcommand{\CE}{\mathcal{E}}    
\newcommand{\SE}{\mathfrak{e}}     
\newcommand{\BR}{\mathbf{R}}     
\newcommand{\eD}{\mathcal{D}}   
\newcommand{\hL}{\hat{L}}   

\newcommand{\dt}{\widetilde{\otimes}}     
\newcommand{\wtimes}{\bar{\otimes}}     
\newcommand{\mol}{\mu_{\mathbf{l}}}  
\newcommand{\mor}{\mu_{\mathbf{r}}}  
\newcommand{\ev}{\mathrm{ev}}         

\newcommand{\Ba}{\mathbf{a}}  
\newcommand{\Bb}{\mathbf{b}}
\newcommand{\Bd}{\mathbf{d}}       
\newcommand{\Be}{\mathbf{e}}
\newcommand{\Bf}{\mathbf{f}}
\newcommand{\Bg}{\mathbf{g}}
\newcommand{\Bw}{\mathbf{w}}       
\newcommand{\Bm}{\mathbf{m}}     
\newcommand{\Bs}{\mathbf{s}}

\allowdisplaybreaks                
\begin{document}
\title[Elliptic quantum groups]{Elliptic quantum groups and Baxter relations}
\author{Huafeng Zhang}
\address{Laboratoire Paul Painlev\'e \& 
Universit\'e Lille 1, 
59655 Villeneuve d'Ascq, France  }
\email{Huafeng.Zhang@math.univ-lille1.fr}
\begin{abstract}
We introduce a category $\mathcal O$ of modules over the
elliptic quantum group of $\mathfrak{sl}_N$ with
well-behaved q-character theory. We construct asymptotic modules as analytic continuation of a family of finite-dimensional modules, the Kirillov--Reshetikhin modules. In the Grothendieck ring of
this category we prove two types of identities: generalized Baxter relations in the spirit of Frenkel--Hernandez between finite-dimensional modules and asymptotic modules;  three-term Baxter TQ relations of infinite-dimensional modules.
\end{abstract}
\maketitle
\setcounter{tocdepth}{1}
\tableofcontents
\section*{Introduction}
Fix $\mathfrak{sl}_N$ a special linear Lie algebra, $\BC/(\BZ + \BZ \tau)$ an elliptic curve, and $\hbar$ a complex number. Associated to this triple is the {\it elliptic quantum group} $\CE_{\tau,\hbar}(\mathfrak{sl}_N)$ introduced by G. Felder \cite{F}. It is a Hopf algebroid (neither commutative nor co-commutative) in the sense of Etingof--Varchenko \cite{EV}, so that the tensor product of two $\CE_{\tau,\hbar}(\mathfrak{sl}_N)$-modules is naturally endowed with a module structure. In this paper we study (finite- and infinite-dimensional) representations of the elliptic quantum group.

\medskip

Suppose $\hbar$ is a formal variable. $\CE_{\tau,\hbar}(\mathfrak{sl}_2)$ is an $\hbar$-deformation \cite{EF} of the universal enveloping algebra of a Lie algebra $\mathfrak{sl}_2 \otimes R_{\tau}$, where $R_{\tau}$ is an algebra of meromorphic functions of $z \in \BC$ built from the Jacobi theta function of the elliptic curve. For $\Glie$ an arbitrary finite-dimensional simple Lie algebra, $\CE_{\tau,\hbar}(\Glie)$ is defined \cite{twist} to be a quasi-Hopf algebra twist of the {\it affine quantum group} $U_{\hbar}(L\Glie)$, an $\hbar$-deformation of the loop Lie algebra $\Glie \otimes \BC[z,z^{-1}]$. It admits a universal dynamical R-matrix in a completed tensor square, which provides solutions $R(z;\lambda) \in \mathrm{End}(V\otimes V)$, for $V$ a suitable $\CE_{\tau,\hbar}(\Glie)$-module, to the {\it quantum dynamical Yang--Baxter Equation}:
\begin{multline*}
R_{12}(z-w;\lambda+\hbar h^{(3)}) R_{13}(z;\lambda) R_{23}(w;\lambda + \hbar h^{(1)}) \\
= R_{23}(w;\lambda)R_{13}(z;\lambda+\hbar h^{(2)}) R_{12}(z-w;\lambda) \in \mathrm{End}(V^{\otimes 3}).
\end{multline*}
Here $z, w$ are complex spectral parameters, $\lambda$ is the dynamical parameter lying in a Cartan subalgebra of $\Glie$, the sub-indexes of $R$ indicate the tensor factors of $V^{\otimes 3}$ to be acted on, and the $h^{(i)}$ are grading operators arising from the weight grading on $V$ by the Cartan subalgebra. See the comments following Eq.\eqref{equ: dYBE elliptic}.

Such R-matrices $R(z;\lambda)$ appeared previously in face-type integrable models \cite{FV2,HSY}; for instance, the R-matrix of the Andrews--Baxter--Forrester model comes from two-dimensional irreducible modules of $\CE_{\tau,\hbar}(\mathfrak{sl}_2)$, as does the 6-vertex model from the affine quantum group $U_{\hbar}(L\mathfrak{sl}_2)$. The definition of $\CE_{\tau,\hbar}(\mathfrak{sl}_N)$ in \cite{F}, by RLL exchange relations, is in the spirit of Faddeev--Reshetikhin--Takhatajan, originated from Quantum Inverse Scattering Method. We mention that elliptic R-matrices describe the monodromy of the quantized Knizhnik--Zamolodchikov equation associated with representations of affine quantum groups, e.g. \cite{IFR,GS,K0,TV0}. 

Recently Aganagic--Okounkov \cite{AO} proposed the elliptic stable envelope in equivariant elliptic cohomology, as a geometric framework to obtain elliptic R-matrices. This was made explicit \cite{FRV} for cotangent bundles of Grassmannians, resulting in tensor products of two-dimensional irreducible representations of $\CE_{\tau,\hbar}(\mathfrak{sl}_2)$. The higher rank case of $\mathfrak{sl}_N$  was studied later by H. Konno \cite{K3}.

 Meanwhile, Nekrasov--Pestun--Shatashvili \cite{NPS} from the 6d quiver gauge theory predicted the elliptic quantum group associated to an arbitrary Kac--Moody algebra, the precise definition of which (as an associative algebra) was proposed by Gautam--Toledano Laredo \cite{GTL2}. See also \cite{YZ} in the context of quiver geometry.

We are interested in the representation theory of $\CE_{\tau,\hbar}(\Glie)$ with $\hbar \in \BC$ generic. The formal twist constructions \cite{EF,twist} from $U_{\hbar}(L\Glie)$ might reduce the problem to the representation theory of affine quantum groups, which is a subject developed intensively in the last three decades from algebraic, geometric and combinatorial aspects. However {\it loc.cit.} involve formal power series of $\hbar$ and infinite products in the comultiplication of $\CE_{\tau,\hbar}(\Glie)$. Some of these divergence issues was addressed \cite{EM} by Etingof--Moura, who defined a fully faithful tenor functor between representation categories of BGG type for $U_{\hbar}(L\mathfrak{sl}_N)$ and $\CE_{\tau,\hbar}(\mathfrak{sl}_N)$. Towards this functor not much is known: its image, the induced homomorphism of Grothendieck rings, etc.

In this paper we study representations of $\CE_{\tau,\hbar}(\mathfrak{sl}_N)$ via the RLL presentation \cite{F} so as to bypass affine quantum groups, yet along the way we borrow ideas from the affine case. Compared to other works \cite{C,EM,FV1,GTL2,K2,K1,TV,YZ}, our approach emphasizes more on the Grothendieck ring structure of representation category.  It is a higher rank extension of a recent joint work with G. Felder \cite{FZ}.  

\medskip

The presence of the dynamical parameter $\lambda$ is one of the technical difficulties of elliptic quantum groups. 
To resolve this, we need a commuting family of elliptic Cartan currents $\phi_j(z) \in \CE_{\tau,\hbar}(\mathfrak{sl}_N)$ for $j \in J := \{1,2,\cdots,N-1\}$. They act as difference operators on an $\CE_{\tau,\hbar}(\mathfrak{sl}_N)$-module $V$, and their matrix entries are meromorphic functions of $(z,\lambda) \in \BC \times \Hlie$ where $\Hlie$ denotes the Cartan subalgebra of $\mathfrak{sl}_N$. As in \cite{FZ}, we impose the following triangularity condition: \footnote{In terms of the $K_i(z)$ from Eq.\eqref{def: elliptic diagonal}, we have $\phi_j(z) = K_j(z+\ell_j\hbar)K_{j+1}(z+\ell_j\hbar)^{-1}$ where $\ell_j = (N-j-1)/2$. These are elliptic deformations of diagonal matrices in $\mathfrak{sl}_N$. }
\begin{itemize}
\item[(i)] there exists a basis of $V$, with respect to which the matrices $\phi_j(z)$ are upper triangular and their diagonal entries are independent of $\lambda$.
\end{itemize}
Our category $\BGG$ is the full subcategory of category BGG \cite{EM} of $\CE_{\tau,\hbar}(\mathfrak{sl}_N)$-modules subject to Condition (i); see Definition \ref{def: O}. It is abelian and monoidal. It contains most of the modules in \cite{C,EM,K2,K1,TV}, although the proof is rather indirect. (We believe category $\BGG$ to be the image of the functor \cite{EM}.)

We extend the {\it q-character} of H. Knight \cite{Kn} and Frenkel--Reshetikhin \cite{FR} to the elliptic case. The q-character of a module $V$ encodes the spectra of the $\phi_j(z)$, which are meromorphic functions of $z$ thanks to Condition (i). It distinguishes the isomorphism class $[V]$ in the Grothendieck ring $K_0(\BGG)$, and embeds $K_0(\BGG)$ in a commutative ring. Our main results are summarized as follows.
\begin{itemize}
\item[(ii)] Proposition \ref{prop: asymptotic representations} on limit construction of infinite-dimensional {\it asymptotic} modules $\CW_{r,x}$, for $r \in J$ and $x \in \BC$, from a distinguished family of finite-dimensional modules, the {\it Kirillov--Reshetikhin} modules.
\item[(iii)] Theorem \ref{thm: generalized TQ} on generalized Baxter relations \`a la Frenkel--Hernandez \cite{FH}: the isomorphism class of any finite-dimensional module is a polynomial of the $\frac{[\CW_{r,x}]}{[\CW_{r,y}]}$ for $r \in J$ and $x,y \in \BC$.
\item[(iv)] Corollary \ref{cor: TQ} relating an asymptotic module $\CW$ to a module $D$ and tensor products $S',S''$ of asymptotic modules such that $[D] [\CW] = [S'] + [S'']$.
\end{itemize}

\medskip

The above results are known in category $\HJ$ of Hernandez--Jimbo \cite{HJ} for representations over a Borel subalgebra of an affine quantum group $U_{\hbar}(L\Glie)$. Category $\HJ$ contains the modules $L_{r,a}^{\pm}$ for $a \in \BC$ and $r$ a Dynkin node of $\Glie$. The $L_{r,a}^{\pm}$ are ``prefundamental" in that their tensor products realize all irreducible objects of $\HJ$ as sub-quotients, and they are not modules over $U_{\hbar}(L\Glie)$, which makes Borel subalgebras indispensable. The Grothendieck ring of $\HJ$ is commutative.

(ii) is the asymptotic limit construction \cite{HJ} of the $L_{r,a}^-$. (iii) is the relation \cite{FH} between finite-dimensional modules and the $L_{r,a}^+$. (iv) is either $QQ^*$-system \cite{Jimbo2,HL} or $Q\widetilde{Q}$-system \cite{FH2}, as there are two choices of the modules $D$ for $\CW = L_{r,a}^+$.

Hernandez--Leclerc \cite{HL} interpreted the $QQ^*$ system \cite{HL} as cluster mutations of Fomin--Zelevinsky. They provided conjectural monoidal categorifications of infinite rank cluster algebras by certain subcategories of $\HJ$. 

In a quantum integrable system associated to $U_{\hbar}(L\Glie)$, the transfer-matrix construction  defines an action of the Grothendieck ring $K_0(\HJ)$ on the quantum space; to an isomorphism class $[V]$ is attached a transfer matrix $t_V(z)$.  

(iii) is one key step \cite{FH} in solving the conjecture of Frenkel--Reshetikhin \cite{FR} on the spectra of the quantum integrable system, which connects the eigenvalues of the $t_V(z)$ to the q-character of $V$ by the so-called Baxter polynomials \cite{Baxter72}. These polynomials are  eigenvalues of the $t_{L_{r,a}^+}(z)$ up to an overall factor \cite{FH}. In this sense the $L_{r,a}^+$ have simpler structures than finite-dimensional modules, and the $t_{L_{r,a}^+}(z)$ are defined as {\it Baxter Q operators}, as an extension of earlier works of V. Bazhanov et al. \cite{BLZ2,BLZ3,BT} for $\Glie$ a special linear Lie (super)algebra. (iv) has as consequence the Bethe Ansatz Equations for the roots of Baxter polynomials \cite{Jimbo2,FH2}. 

 Recently category $\HJ$ was studied for quantum toroidal algebras \cite{Jimbo1}.

\medskip

For elliptic quantum groups there are no obvious Borel subalgebras. Our idea is to replace the $L_{r,a}^{\pm}$ over Borel subalgebras by the asymptotic modules $\CW_{d,a}^{(r)}$ (with a new parameter $d\in \BC$) over the {\it entire} quantum group, which we now explain.

Let $\theta(z) := \theta(z|\tau)$ be the Jacobi theta function. 
For $r \in J$ a Dynkin node, $a \in \BC$ a spectral parameter, and $k$ a positive integer, by \cite{C,TV} there exists a unique finite-dimensional irreducible module $W_{k,a}^{(r)}$ which contains a non-zero vector $\omega$ (highest weight with respect to a triangular decomposition) such that:
$$ \phi_j(z) \omega = \omega \ \mathrm{if}\ j \neq r,\quad \phi_r(z) \omega = \frac{\theta(z+a\hbar + k\hbar)}{\theta(z+a\hbar)} \omega. $$
This is a Kirillov--Reshetikhin (KR) module, a standard terminology for affine quantum groups and Yangians once the $\theta$ symbol is removed.

The core of this paper (Section \ref{sec: asym}) is analytic continuation with respect to $k$. We modify the asymptotic limits $L_{r,a}^-$ of Hernandez--Jimbo \cite{HJ}, as in \cite{FZ,Z3}. 

Firstly the existence of the inductive system  $(W_{k,a}^{(r)})_{k > 0}$ in \cite{HJ} relied on a cyclicity property of M. Kashiwara, Varagnolo--Vasserot and V. Chari, which is unavailable in the elliptic case. We reduce the problem to $\CE_{\tau,\hbar}(\mathfrak{sl}_2)$ by counting ``dominant weights" in q-characters (Theorem \ref{thm: Demazure KR}), as in the proofs of T-system of KR modules over affine quantum groups by H. Nakajima \cite{Nakajima} and D. Hernandez \cite{H1}. 

Secondly we express the matrix coefficients of any element of $\CE_{\tau,\hbar}(\mathfrak{sl}_N)$ acting on the $W_{k,a}^{(r)}$, viewed as functions of $k \in \BZ_{>0}$, in products of the $\theta(k\hbar+ c)$ where $c \in \BC$ is independent of $k$; see Lemma \ref{lem: asymptotic property}. In \cite{HJ} these are polynomials in $k$ by induction. Our proof relies on the RLL comultiplication and is explicit.

$\theta(k\hbar+c)$ being an entire function of $k$, we take $k$ in the matrix coefficients to be a fixed complex number $d$. This results in the asymptotic module $\CW_{d,a}^{(r)}$ on the inductive limit $\lim\limits_{\rightarrow} W_{k,a}^{(r)}$. The module $\CW_{r,x}$ in (ii) is $\CW_{x,0}^{(r)}$. All irreducible modules of category $\BGG$ are sub-quotients of tensor products of asymptotic modules.

\medskip

For $\Glie$ of general type (ii)--(iv) and their proofs can be adapted to affine quantum groups, whose asymptotic modules appeared in \cite[Appendix]{Z3}, as well as Yangians \cite{GTL,GTL1}. Borel subalgebras or double Yangians are not needed.

(ii)--(iv) were established for affine quantum general linear Lie superalgebras \cite{Z2}; their proofs require more \cite{Z1} than q-characters as counting dominant weights is inefficient. It is interesting to consider elliptic quantum supergroups \cite{GS}.

For elliptic quantum groups associated with other simple Lie algebras, one possible first step would be to derive the RLL presentation; see \cite{Guay,JLM} for Yangians.

R-matrix of Baxter--Belavin is governed by the vertex-type elliptic quantum group \cite{twist}. The equivalence \cite{ES} of representation categories between this elliptic algebra and $\CE_{\tau,\hbar}(\mathfrak{sl}_N)$, a Vertex-IRF correspondence, might give a representation theory meaning to the original Baxter Q operator of the 8-vertex model \cite{Baxter72}.

\medskip

The paper is structured as follows. In Section \ref{sec: elliptic} we review the theory
of the elliptic quantum group associated to $\mathfrak {sl}_N$, and define category $\mathcal O$ of representations. We show that the $q$-character map is an injective
ring homomorphism from the Grothendieck ring $K_0(\mathcal O)$ to a commutative ring $\CMt$ of meromorphic functions. 
Then we present the q-character formula of finite-dimensional evaluation modules.  

Section \ref{sec: evaluation} is devoted to the proof of the q-character formula. 

We derive in Section \ref{sec: KR} basic facts on tensor products of KR modules (T-system, fusion) from the q-character formula. They are needed in Section \ref{sec: asym} to construct the inductive system of KR modules and the asymptotic modules. We obtain a highest weight classification of irreducible modules in category $\BGG$. As a consequence, all standard irreducible evaluation modules of \cite{TV} are in category $\BGG$.

In Section \ref{sec: TQ} we establish the three-term Baxter TQ relations in $K_0(\BGG)$, which are infinite-dimensional analogs of the T-system. These relations are interpreted as functional relations of transfer matrices in Section \ref{sec: Q}. 

\section{Elliptic quantum groups and their representations}  \label{sec: elliptic}
Let $N \in \BZ_{>0}$. We introduce a  category $\BGG$ (abelian and monoidal) of representations of the elliptic quantum group attached to the Lie algebra $\mathfrak{sl}_N$, and prove that its Grothendieck ring is commutative, based on q-characters. 

Fix a complex number $\tau \in \BC$ with $\mathrm{Im}(\tau) > 0$. Define the Jacobi theta function 
$$ \theta(z) = \theta(z|\tau) := - \sum_{j=-\infty}^{\infty} \exp\left(\BI \pi (j+\frac{1}{2})^2\tau + 2\BI\pi (j+\frac{1}{2})(z+\frac{1}{2})\right), \quad \BI = \sqrt{-1}.  $$
It is an entire function of $z \in \BC$ with zeros lying on the lattice $\Gamma := \BZ + \BZ \tau$ and 
$$\theta(z+1) = -\theta(z),\quad \theta(z+\tau) = -e^{-\BI\pi \tau-2\BI \pi z} \theta(z),\quad \theta(-z) = - \theta(z).  $$
Fix a complex number $\hbar \in \BC \setminus (\mathbb{Q} + \mathbb{Q}\tau)$, which is the deformation parameter. 

 Let $\Hlie$ be standard Cartan subalgebra of $\mathfrak{sl}_N$; it is a complex vector space generated by the $\epsilon_i$ for $1\leq i \leq N$ subject to the relation $\sum_{i=1}^N \epsilon_i = 0$.  
Let $\BC^N := \oplus_{i=1}^N \BC v_i$ and $E_{ij} \in \End_{\BC}(\BC^N)$ be the elementary matrices: $v_k \mapsto \delta_{jk} v_i$ for $1\leq i,j,k \leq N$. Define the $ \End_{\BC}(\BC^N\otimes \BC^N)$-valued meromorphic functions of $(z;\lambda) \in \BC \times \Hlie$ by:
\begin{align*}
\BR(z;\lambda) = \sum_i E_{ii}^{\otimes 2} + \sum_{i\neq j} \left( \frac{\theta(z)\theta(\lambda_{ij}-\hbar)}{\theta(z+\hbar)\theta(\lambda_{ij})} E_{ii}\otimes E_{jj} + \frac{\theta(z+\lambda_{ij})\theta(\hbar)}{\theta(z+\hbar)\theta(\lambda_{ij})} E_{ij} \otimes E_{ji}\right).
\end{align*}
In the summations $1\leq i,j \leq N$, and $\lambda_{ij}  \in \Hlie^*$ sends $\sum_{i=1}^N c_i \epsilon_i \in \Hlie$ to $c_i - c_j \in \BC$.  By \cite{F},  $\BR(z;\lambda)$  satisfies the {\it quantum dynamical Yang--Baxter equation}:
\begin{multline}   \label{equ: dYBE elliptic}  
\quad \quad \BR^{12}(z-w;\lambda +\hbar h^{(3)}) \BR^{13}(z;\lambda) \BR^{23}(w;\lambda+\hbar h^{(1)}) \\
= \BR^{23}(w;\lambda)\BR^{13}(z;\lambda+\hbar h^{(2)})\BR^{12}(z-w;\lambda) \in \End_{\BC}(\BC^N)^{\otimes 3}.
\end{multline}
If $\BR(z;\lambda) = \sum_{p}c_{\lambda}^px_p \otimes y_p$ with $x_p,y_p \in \End_{\BC}(\BC^N)$, then
$$ \BR^{13}(z;\lambda+\hbar h^{(2)}): u \otimes v_j \otimes w \mapsto \sum_p c^p_{\lambda+\hbar \epsilon_j}  x_p(u)\otimes v_j \otimes y_p(w). $$
for $u,w \in \BC^N$ and $1\leq j \leq N$. The other symbols have a similar meaning.

Let $\BM := \BM_{\Hlie}$ be the field of meromorphic functions of $\lambda \in \Hlie$. It contains the subfield $\BC$ of constant functions. A $\BC$-linear map $\Phi$ of two $\BM$-vector spaces will sometimes be denoted by $\Phi(\lambda)$ to emphasize the dependence on $\lambda$.  

\subsection{Algebraic notions.} \label{ss-h algebras}
Since the elliptic quantum groups will act on $\BM$-vector spaces via difference operators, which are in general not $\BM$-linear, we need to recall some basis constructions about difference operators. Our exposition follows largely \cite{EV}, with minor modifications as in \cite{FZ}.

 Define the category $\CV$ as follows. An object is $X = \oplus_{\alpha\in \Hlie}X[\alpha]$ where each $X[\alpha]$ is an $\BM$-vector space and, if non-zero, is called a weight space of weight (or $\Hlie$-weight) $\alpha$. Let $\wt(X) \subseteq \Hlie$ be the set of weights of $X$. Write $\wt(v) = \alpha$ if $v \in X[\alpha]$.  
 
 A morphism $f: X \longrightarrow Y$ in $\CV$ is an $\BM$-linear map which respects the weight gradings. Let $\CVf$ be the full subcategory of $\CV$ consisting of $X$ whose weight spaces are finite-dimensional $\BM$-vector spaces. (``ft" means finite type in \cite{FV1}.)

Viewed as subcategories of the category of $\BM$-vector spaces, $\CV$ and $\CVf$ are abelian.  
 
Let $X, Y$ be objects of $\CV$. Their {\it dynamical tensor product} $X \wtimes Y$ is constructed as follows. For $\alpha,\beta \in \Hlie$, let $X[\alpha] \wtimes Y[\beta]$ be the quotient of the usual tensor product of $\BC$-vector spaces $X[\alpha]\otimes_{\BC} Y[\beta]$ by the relation 
$$g(\lambda)v \otimes_{\BC} w = v \otimes_{\BC} g(\lambda+\hbar\beta) w \quad \mathrm{for}\ v \in X[\alpha],\ w \in Y[\beta], \ g(\lambda) \in \BM.$$
Let $\wtimes$ denote the image of $\otimes_{\BC}$ under the quotient. $X[\alpha]\wtimes Y[\beta]$ becomes an $\BM$-vector space by setting $g(\lambda) (v \wtimes w) = v \wtimes g(\lambda) w$. For $\gamma \in \Hlie$, the weight space $(X\wtimes Y)[\gamma]$ is then the direct sum of the $X[\alpha] \wtimes Y[\beta]$ with $\alpha+\beta = \gamma$.

For $\alpha,\beta \in \Hlie$, a $\BC$-linear map $\Phi: X \longrightarrow Y$ is called a {\it difference map} of bi-degree $(\alpha,\beta)$ if it sends every weight space $X[\gamma]$ to $Y[\gamma + \beta - \alpha]$, and if  \cite[\S 4.2]{EV}:
  $$   \Phi(g(\lambda) v) = g(\lambda+\beta\hbar) \Phi(v) \quad \mathrm{for}\ g(\lambda) \in \BM \ \mathrm{and}\ v \in X.  $$
Such a map can be recovered from its matrix as in the case of $\BM$-linear maps. Choose $\BM$-bases $\mathcal{B}, \mathcal{B}'$ for $X$ and $Y$ respectively. Define the $\mathcal{B}' \times \mathcal{B}$ matrix $[\Phi]$ by taking its $(b',b)$-entry $[\Phi]_{b' b}(\lambda) \in \BM$, for $b\in \mathcal{B}$ and $b' \in \mathcal{B'}$, to be the coefficient of $b'$ in $\Phi(b)$.
Then for any vector $v = \sum_{b \in \mathcal{B}} g_b(\lambda) b$ of $X$ where $g_b(\lambda) \in \BM$, we have \footnote{Note that difference maps of bi-degree $(\alpha,\alpha)$ make sense for arbitrary $\BM$-vector spaces.  \label{f1: difference operators} }
$$ \Phi(v) = \sum_{b' \in \mathcal{B'}} b' \sum_{b \in \mathcal{B}}  [\Phi]_{b'b}(\lambda) \times g_b(\lambda+\hbar \beta). $$
When $X = Y$, a difference map is also called a difference operator. To define its matrix, we always assume $\mathcal{B}' = \mathcal{B}$.

\medskip

By an algebra we mean a unital associative algebra over $\BC$. 

 As in \cite[Definition 4.1]{EV}, an $\Hlie$-{\it algebra} is an algebra $A$, endowed with $\Hlie$-bigrading $A = \oplus_{\alpha,\beta\in \Hlie} A_{\alpha,\beta}$ which respects the algebra structure and is called the {\it weight decomposition}, and two algebra embeddings $\mol,\mor: \BM \longrightarrow A_{0,0}$ called the {\it left and right moment maps}, such that for $a \in A_{\alpha,\beta}$ and $g(\lambda) \in \BM$, we have
$$ \mol(g(\lambda)) a = a \mol(g(\lambda-\hbar\alpha)),\quad \mor(g(\lambda)) a = a \mor(g(\lambda-\hbar\beta)). $$
Call $(\alpha,\beta)$ the bi-degree of elements in $A_{\alpha,\beta}$. A morphism of $\Hlie$-algebras is an algebra morphism preserving the moment maps and the weight decompositions. 

From two $\Hlie$-algebras $A,B$ we construct their tensor product $A\dt B$ as follows. For $\alpha,\beta,\gamma \in \Hlie$, let $A_{\alpha,\beta}\ \dt\ B_{\beta,\gamma}$ be $A_{\alpha,\beta} \otimes_{\BC} B_{\beta,\gamma}$ modulo the relation 
$$\mor^A(g(\lambda)) a \otimes_{\BC} b = a \otimes_{\BC} \mol^B(g(\lambda)) b \quad \mathrm{for}\ a \in A_{\alpha,\beta},\  b \in B_{\beta,\gamma},\  g(\lambda) \in \BM.$$
$(A\dt B)_{\alpha,\gamma}$ is the direct sum of the $A_{\alpha,\beta}\ \dt\ B_{\beta,\gamma}$ over $\beta \in \Hlie$. Multiplication in $A \dt  B$ is induced by $(a\dt b)(a'\dt b') = aa' \dt bb'$. The moment maps are given by ($\dt$ denotes the image of $\otimes_{\BC}$ under the quotient $\otimes_{\BC} \longrightarrow \dt$)
$$\mol^{A\dt B}: g(\lambda) \mapsto \mol^A(g(\lambda)) \dt 1, \quad \mor^{A\dt B}: g(\lambda) \mapsto 1 \dt \mor^B(g(\lambda))\quad \mathrm{for}\ g(\lambda) \in \BM. $$

\medskip

To an $\Hlie$-graded vector space one can attach naturally an $\Hlie$-algebra. Let $X$ be an object of $\CV$. 
Let $\BD^X_{\alpha,\beta}$ denote the $\BC$-vector space of difference operators $X \longrightarrow X$ of bidegree $(\alpha,\beta)$. Then the direct sum $\BD^X := \oplus_{\alpha,\beta \in \Hlie} \BD^X_{\alpha,\beta}$ is a subalgebra of $\End_{\BC}(X)$. It is an $\Hlie$-algebra structure with  the moment maps  
$$\mor(g(\lambda)) v = g(\lambda)v,\quad \mol(g(\lambda)) v = g(\lambda+\hbar\alpha) v \quad \mathrm{for}\ v \in X[\alpha],\ g(\lambda) \in \BM.$$

Tensor products of difference operators are also difference operators. To be precise,
let $X,Y$ be two objects of $\CV$. Let $\Phi: X \longrightarrow X$ and $\Psi: Y \longrightarrow Y$ be difference operators of bi-degree $(\alpha,\beta)$ and $(\beta,\gamma)$ respectively.  The $\BC$-linear map
\begin{equation*} 
 X \otimes_{\BC} Y \longrightarrow X \wtimes Y, \quad v \otimes_{\BC} w \mapsto \Phi(v) \wtimes \Psi(w)
\end{equation*}
is easily seen to factorize through $X \otimes_{\BC} Y \longrightarrow X \wtimes Y$ and 
 induces the $\BC$-linear map $\Phi \wtimes \Psi: X \wtimes Y \longrightarrow X \wtimes Y$, which is shown to be a difference operator of bi-degree $(\alpha,\gamma)$. As in \cite[Lemma 4.3]{EV}, the following defines a morphism of $\Hlie$-algebras
 $$\BD^X\ \dt\ \BD^Y \longrightarrow \BD^{X\wtimes Y}, \quad \Phi \dt \Psi \mapsto \Phi \wtimes \Psi. $$

\subsection{Elliptic quantum groups.}  \label{ss-elliptic}
For $1\leq i,j,p,q \leq N$ let $R_{pq}^{ij}(z;\lambda)$ be the coefficient of $v_p \otimes v_q$ in $\BR(z;\lambda)(v_i \otimes v_j)$; it can be viewed as an element of $\BM$ after fixing $z \in \BC$. The {\it elliptic quantum group} $\CE := \CE_{\tau,\hbar}(\mathfrak{sl}_N)$ is an $\Hlie$-algebra generated by\footnote{We use $\mathfrak{sl}_N$, as in \cite{F,FV1}, to emphasize that $\Hlie$ is the Cartan subalgebra of $\mathfrak{sl}_N$. Other works \cite{C,K1} use $\mathfrak{gl}_N$ for the reason that the elliptic quantum determinant is not fixed to be 1.   } 
$$L_{ij}(z) \in \CE_{\epsilon_i,\epsilon_j} \quad \mathrm{for}\ 1 \leq i,j \leq N $$
subject to the dynamical RLL relation \cite[\S 4.4]{EV}: for $1\leq i,j,m,n \leq N$,
\begin{multline} \label{rel: RLL explicit}
\quad \quad \sum_{p,q=1}^N \mol\left(R^{pq}_{mn}(z-w;\lambda)\right) L_{pi}(z) L_{qj}(w) \\
= \sum_{p,q=1}^N  \mor\left(R_{pq}^{ij}(z-w;\lambda)\right) L_{nq}(w)L_{mp}(z).
\end{multline}
There is an $\Hlie$-algebra morphism \cite{EV,FV1}:
\begin{align} \label{def: coproduct}
\Delta: \CE \longrightarrow \CE\  \dt\  \CE,\quad L_{ij}(z) \mapsto \sum_{k=1}^N L_{ik}(z)\ \dt\ L_{kj}(z) \quad \mathrm{for}\ 1\leq i,j \leq N
\end{align}
which is co-associative $(1\dt \Delta) \Delta = (\Delta \dt 1) \Delta$ and is called the coproduct. For $u \in \BC$,
\begin{align} \label{def: spectral shift}
\Phi_u: \CE \longrightarrow \CE,\quad L_{ij}(z) \mapsto L_{ij}(z+u\hbar) \quad \mathrm{for}\ 1 \leq i,j \leq N
\end{align}
extends uniquely to an $\Hlie$-algebra automorphism (spectral parameter shift).

Strictly speaking, $\CE$ is not well-defined as an $\Hlie$-algebra because of the additional parameter $z$; this is resolved in \cite{K1} by viewing $z, \hbar$ as formal variables. In this paper we are mainly concerned with representations in which Eq.\eqref{rel: RLL explicit}--\eqref{def: spectral shift} make sense as identities of difference operators depending analytically on $z$. 

Let $\mathfrak{S}_N$ be the group of permutations of $\{1,2,\cdots,N\}$. For $1\leq k \leq N$, let $\mathfrak{S}^k$ be the subgroup of permutations  which fix the last $k$ letters.  
The $k$-th {\it fundamental weight} $\varpi_k$ and {\it elliptic quantum minor} $\eD_k(z)$ are defined by \cite[Eq.(2.5)]{TV}:
\begin{align}
 \varpi_k &:= \sum_{i=1}^k \epsilon_{i} \in \Hlie,\quad   \Theta_{k}(\lambda) := \prod_{N-k+1\leq i < j \leq N} \theta(\lambda_{ij}) \in \BM^{\times},  \label{def: Vandermond} \\
\eD_k(z) &:= \frac{\mor(\Theta_k(\lambda))}{\mol(\Theta_k(\lambda))} \sum_{\sigma \in \mathfrak{S}^k} \mathrm{sign}(\sigma) \prod_{i=N}^{N-k+1} L_{\sigma(i),i}(z+(N-i)\hbar) \in \CE.\label{def: determinant}
\end{align}
Here $\mathrm{sign}(\sigma) \in \{\pm 1\}$ denotes the signature of the permutation $\sigma$. We take the descending product over $N \geq i \geq N-k+1$ in Eq.\eqref{def: determinant}. Set $\varpi_0 := 0$.

We shall need the following elements $\hat{L}_{k}(z)$ of $\CE_{\epsilon_k,\epsilon_k}$ as in \cite[Eq.(4.1)]{TV}:
\begin{equation}  \label{defi: auxiliary L}
\hat{L}_N(z) := L_{NN}(z),\quad \hat{L}_k(z) = L_{kk}(z) \prod_{j=k+1}^{N}  \frac{\mor(\theta(\lambda_{kj}))}{\mol(\theta(\lambda_{kj}))}.
\end{equation}

\begin{theorem}\cite[Proposition 2.1]{TV} \cite[Eq.(E.18)]{K1}  \label{thm: determinant}
$\eD_N(z)$ is central in $\CE$ and grouplike: $\Delta(\eD_N(z)) = \eD_N(z)\ \dt\ \eD_N(z)$.
\end{theorem}

The simple roots $\alpha_i := \epsilon_i-\epsilon_{i+1}$ for $1 \leq i < N$ generate a free abelian subgroup $\BQ$ of $\Hlie$, called the root lattice. Let $\BQ_+, \BQ_-$ be submonoids of $\BQ$ generated by the $\alpha_i, -\alpha_i$ respectively.  Define the lexicographic partial ordering $\prec$ on $\Hlie$ as follows: 
\begin{itemize}
    \item[] $\alpha \prec \beta$ if $\beta - \alpha=n_l \alpha_l + \sum_{i=l+1}^{N-1} n_i \alpha_i \in \BQ$ with $n_l \in \BZ_{>0}$. 
\end{itemize}
This is weaker than the standard ordering: $\alpha\leq \beta$ if $\beta - \alpha \in \BQ_+$.

\begin{cor} \label{cor: Jucys-Murphy}
$\eD_k(z)$ commutes with the $L_{ij}(w)$ for $N-k <  i,j \leq N$ and $\Delta(\eD_k(z)) - \eD_k(z)\ \dt\ \eD_k(z)$ is a finite sum $\sum_{\alpha} x_{\alpha} \dt y_{\alpha}$ over $\{\alpha \in  \Hlie \ |\  -\varpi_{N-k} \prec \alpha\}$ where $x_{\alpha}$ and $y_{\alpha}$ are of bi-degree $(-\varpi_{N-k},\alpha)$ and $(\alpha,-\varpi_{N-k})$ respectively.
\end{cor}
The proof of the corollary is postponed to Section \ref{ss: proof of Cor}.

\subsection{Categories.} \label{ss-repr}
{\it From now on unless otherwise stated vector spaces, linear maps and bases are defined over $\BM$.} Let $X$ be an object of $\CVf$. A representation of $\CE$ on $X$ consists of difference operators $L_{ij}^X(z): X \longrightarrow X$ of bi-degree $(\epsilon_i,\epsilon_j)$ for $1\leq i,j \leq N$ depending on $z \in \BC$ with the following properties:\footnote{This is called a representation of {\it finite type} in \cite{FV1}. From condition (M1) it follows that the coefficients of the $L_{ij}^X(z)$ are meromorphic functions with respect to {\it any} basis of $X$. }
\begin{itemize}
\item[(M1)]  there exists a basis of $X$ with respect to which all the matrix entries of the difference operators $L_{ij}^X(z)$ are meromorphic functions of $(z,\lambda) \in \BC \times \Hlie$.
\item[(M2)] Eq.\eqref{rel: RLL explicit} holds in $\BD^X$ with $\mol,\mor$ being moment maps in $\BD^X$. 
\end{itemize} 
Call $X$ an $\CE$-module. Property (M2) can be interpreted as an $\Hlie$-algebra morphism $\CE \longrightarrow \BD^X$ sending $L_{ij}(z) \in \CE$ to the difference operator $L_{ij}^X(z)$ on $X$.  Applying $\rho$ to the elements of Eqs.\eqref{def: determinant}--\eqref{defi: auxiliary L}, one gets difference operators $\eD_k^X(z), \hat{L}_k^X(z)$ acting on $X$ bi-degree $(-\varpi_{N-k}, -\varpi_{N-k})$ and $(\epsilon_k, \epsilon_k)$ respectively.   When no confusion arises, we shall drop the superscript $X$ from $L^X, \eD^X, \hat{L}^X$ to simplify notations.

A morphism  $\Phi: X \longrightarrow Y$ of $\CE$-modules is a linear map which respects the $\Hlie$-gradings (so that $\Phi$ is a morphism in category $\CV$) and satisfies $\Phi L_{ij}^X(z) = L_{ij}^Y(z) \Phi$ for $1\leq i,j \leq N$. The category of $\CE$-modules is denoted by $\Rep$. It is a subcategory of $\CVf$, and is abelian, since the kernel and cokernel of a morphism of $\CE$-modules, as $\Hlie$-graded $\BM$-vector spaces, are naturally $\CE$-modules.  \footnote{In other works \cite{C,EM,FV1,GTL2,K2,K1,TV,YZ}: a module $V$ is an $\Hlie$-graded $\BC$-vector space; morphisms of modules depend on the dynamical parameter $\lambda$, so do their kernel and cokernel; the abelian category structure is non trivial. The scalar extension gives a module $V \otimes_{\BC} \BM$ in the present situation. Since our modules and morphisms are $\BM$-linear, the dependence of kernels and images on the dynamical parameter does not matter.  }

\begin{defi} \cite[\S 4]{EM}  \label{def: tilde O}
$\tBGG$ is the full subcategory of $\Rep$ whose objects $X$ are such that $\wt(X)$ is contained in a {\it finite} union of cones $\mu + \BQ_-$ with $\mu \in \Hlie$.
\end{defi}

 For $X, Y$ objects in category $\tBGG$, the $L_{ij}^{X\wtimes Y}(z) := \sum_{k=1}^N L_{ik}^X(z)\ \wtimes\ L_{kj}^Y(z)$ define a representation of $\CE$ on $X \wtimes Y$ which is easily seen to be in category $\tBGG$. So $\tBGG$ is a monoidal subcategory of $\CV$. Similarly, $\tBGG$ is an abelian subcategory of $\Rep$.

\begin{defi}\cite[\S 2]{FZ} \label{defi: merom eigen}
An object in $\CFm$ consists of a finite-dimensional vector space $V$ equipped with difference operators $D_{l}(z): V \longrightarrow V$ of bi-degree $(-\varpi_{N-l},-\varpi_{N-l})$ (see Footnote \ref{f1: difference operators}) for $1\leq l \leq N$ depending on $z \in \BC$ such that:  
\begin{itemize}
\item[(M3)] there exists an ordered basis of $V$ with respect to which the matrices of the difference operators $D_l(z)$ are upper triangular, the diagonal entries are non-zero meromorphic functions of $z \in \BC$, and the off-diagonal entries are meromorphic functions of $(z,\lambda) \in \BC \times \Hlie$. 
\end{itemize}
A morphism $\Phi: V \longrightarrow W$ in $\CFm$ is a linear map commuting with the $D_l(z)$. (Namely, $\Phi D_{l}^V(z) = D_l^W(z) \Phi: V \longrightarrow W$ for $1\leq l \leq N$. Here we add the superscripts $V, W$ in the $D_l(z)$ to indicate the space on which they act.)
\end{defi}

For $V$ an object of $\CFm$, the operators $D_l(z)$ being invertible because of the triangularity, one has a unique factorization of operators for $1\leq l \leq N$:
\begin{gather}  
D_l(z) = K_{N}(z) K_{N-1}(z+\hbar) K_{N-2}(z+2\hbar) \cdots K_{N-l+1}(z+(l-1)\hbar).  \label{def: elliptic diagonal} 
\end{gather} 
Notably $K_l(z): X \longrightarrow X$ is a difference operator of bi-degree $(\epsilon_l,\epsilon_l)$.  Property (M3) still holds if the $D_l(z)$ are replaced by the $K_l(z)$.

 The forgetful functor from $\CFm$ to the category of finite-dimensional vector spaces equips $\CFm$ with an abelian category structure. (For a proof, we refer to \cite[\S 2.1]{FZ} where another characterization of category $\CFm$ in terms of Jordan--H\"older series is given.)  Let us describe its Grothendieck group $K_0(\CFm)$.

The multiplicative group $\BM_{\BC}^{\times}$ of non-zero meromorphic functions of $z \in \BC$ contains a subgroup $\BC^{\times}$ of non-zero constant functions. Let $\CM$ be the quotient group of $(\BM_{\BC}^{\times})^N$ by its subgroup formed of $(c_1,c_2,\cdots,c_N) \in (\BC^{\times})^N$ such that $c_1c_2\cdots c_N = 1$. We show that $K_0(\CFm)$ has a $\BZ$-basis indexed by $\CM$.


For $\Bf = (f_1(z),f_2(z), \cdots, f_N(z)) \in (\BM_{\BC}^{\times})^N$, the vector space $\BM$  with the following difference operators $D_l(z)$ is an object in category $\CFm$ denoted by $\BM_{\Bf}$:
$$g(\lambda) \mapsto g(\lambda-\hbar \varpi_{N-l})  f_N(z)f_{N-1}(z+\hbar) f_{N-2}(z+2\hbar) \cdots f_{N-l+1}(z+(l-1)\hbar). $$
We have $K_l(z) g(\lambda) = g(\lambda+\hbar \epsilon_l) f_l(z)$. As a consequence of (M3) in Definition \ref{defi: merom eigen}, all irreducible objects of category $\CFm$ are of this form. 

\begin{lem}
Let $\Be, \Bf \in  (\BM_{\BC}^{\times})^N$. The objects $\BM_{\Be}$ and $\BM_{\Bf}$ are isomorphic in category $\CFm$ if and only if $\Be, \Bf$ have the same image under the quotient $(\BM_{\BC}^{\times})^N \twoheadrightarrow \CM$.
\end{lem}
\begin{proof}
Write $\Be = (e_1(z), e_2(z), \cdots, e_N(z))$ and $\Bf = (f_1(z),f_2(z),\cdots, f_N(z))$.

Sufficiency: assume $e_l(z) = f_l(z) c_l$ with $c_l \in \BC^{\times}$ and $c_1 c_2 \cdots c_N = 1$. For $1\leq l < N$, choose $b_l$ such that $c_l = e^{b_l \hbar} $. Set $b_N := -b_1 - b_2-\cdots - b_{N-1}$. Then $e^{b_N \hbar} = c_1^{-1}c_2^{-1}\cdots c_{N-1}^{-1} = c_N$ and the following is a well-defined element of  $\BM^{\times}$:
$$ \varphi(x_1\epsilon_1 + x_2 \epsilon_2 + \cdots + x_N\epsilon_N ) = e^{b_1x_1 + b_2 x_2 + \cdots + b_N x_N} \quad \mathrm{for}\quad x_1,x_2,\cdots, x_N \in \BC. $$
Indeed $\varphi(\alpha+\beta) = \varphi(\alpha) \varphi(\beta)$ and $\varphi(x \epsilon_1+x\epsilon_2+\cdots+x\epsilon_N) = 1$ for $x \in \BC$. Notably,
$$ \varphi(\lambda+\hbar \epsilon_l) = \varphi(\lambda)  \varphi(\hbar \epsilon_l) = e^{\hbar b_l} \varphi(\lambda) = c_l \varphi(\lambda).   $$
So $\BM_{\Be} \longrightarrow \BM_{\Bf},\ g(\lambda) \mapsto g(\lambda) \varphi(\lambda)$ is an isomorphism in category $\CFm$. 

Necessity: let $\Phi: \BM_{\Be} \longrightarrow \BM_{\Bf}$ be an isomorphism in category $\CFm$. Set $\varphi(\lambda) := \Phi(1)$. Then $\varphi(\lambda) \in \BM^{\times}$. Applying  $\Phi K_l(z) = K_l(z) \Phi$ to $1$ we get
$$ \varphi(\lambda + \hbar \epsilon_l) f_l(z) = \varphi(\lambda) e_l(z).  $$ 
So $\frac{e_l(z)}{f_l(z)} = \frac{\varphi(\lambda+\hbar \epsilon_l)}{\varphi(\lambda)}$, being independent of $z$, is a constant function $c_l \in \BC^{\times}$. We have $e_l(z) = f_l(z) c_l$ and $\varphi(\lambda + \hbar \epsilon_l) = c_l \varphi(\lambda)$. It follows that
$$ \varphi(\lambda) = \varphi(\lambda + \hbar \epsilon_1 + \hbar \epsilon_2 + \cdots + \hbar \epsilon_N) = c_1c_2\cdots c_N \varphi(\lambda), $$
which implies $c_1c_2\cdots c_N = 1$. So $\Be$ and $\Bf$ have the same image in $\CM$.
\end{proof}

For each $\Bf \in \CM$, let us fix a pre-image $\Bf'$ in $(\BM_{\BC}^{\times})^N$ and set $\BM(\Bf) := \BM_{\Bf'}$. Then the isomorphism classes $[\BM(\Bf)]$ for $\Bf \in \CM$ form a $\BZ$-basis of $K_0(\CFm)$. When no confusion arises, we identify an element of $(\BM_{\BC}^{\times})^N$ with its image in $\CM$.

\begin{lem}  \label{lem: category F}
Let $V$ be in category $\CFm$. Assume $B$ is an ordered basis of $V$ with respect to which the matrices of the difference operators $K_l(z)$ are upper triangular. Then for $b \in B$ and $1\leq l \leq N$ there exist $\varphi_b(\lambda) \in \BM^{\times}$ and $f_{b,l}(z) \in \BM_{\BC}^{\times}$ such that 
$$ [K_l]_{bb}(z;\lambda) = f_{b,l}(z)  \frac{\varphi_b(\lambda)}{\varphi_b(\lambda+\hbar \epsilon_l)}. $$
\end{lem}
Recall that $[K_l]_{bb}(z;\lambda)$ is the coefficient of $b$ in $K_l(z) b$. This lemma says that if the matrices of the $K_l(z)$ are upper triangular, then their diagonal entries must be of the form $f(z) h(\lambda)$, and the $h(\lambda)$ can be gauged away uniformly.

More precisely,  the new basis $\{\varphi_b(\lambda) b \ |\ b \in B \}$ with the ordering induced from $B$ satisfies (M3) in Definition \ref{defi: merom eigen}; the diagonal entry of $K_l(z)$ associated to $\varphi_b(\lambda)b$ is $f_{b,l}(z)$. This yields the following identity in the Grothendieck group $K_0(\CFm)$:
$$ [V] = \sum_{b \in B} [\BM(f_{b,1}(z), f_{b,2}(z), \cdots, f_{b,N}(z))]. $$
\begin{proof}
Write $B = \{b_1 < b_2 < \cdots < b_m\}$. We proceed by induction on the dimension $m = \dim (V)$. If $m = 1$, then there exist $\Bf = (f_1(z),f_2(z),\cdots,f_N(z)) \in (\BM_{\BC}^{\times})^N$ and an isomorphism $\Phi: \BM_f \longrightarrow V$ in category $\CFm$. Let $\Phi(1) = \varphi(\lambda) b_1$. Then applying $\Phi K_l(z) = K_l(z) \Phi$ to $1$ we obtain the desired identity
$$ f_l(z) \varphi(\lambda) = [K_l]_{b_1b_1}(z;\lambda) \varphi(\lambda + \hbar \epsilon_l). $$
If $m > 1$, then the subspace $V'$ of $V$ spanned by $(b_1, b_2, \cdots,  b_{m-1})$ is stable by the $K_l(z)$ and $D_l(z)$ by the triangularity assumption. So $V'$ is an object of category $\CFm$ and we obtain a short exact sequence $0 \longrightarrow V' \longrightarrow V \longrightarrow V/V' \longrightarrow 0$. The rest is clear by applying the induction hypothesis to $V', V/V'$, which have ordered bases $\{b_1 < b_2 < \cdots < b_{m-1}\}$ and $\{b_m + V'\}$ respectively. 
\end{proof}

\begin{defi}  \label{def: O}
$\BGG$ is the full subcategory of $\tBGG$ consisting of $\CE$-modules $X$ such that $X[\mu]$ endowed with the action of the $\eD_l(z)$ belongs to $\CFm$ for all $\mu \in \wt(X)$.
\end{defi}

The definition of $\tBGG$ is standard as in the cases of Kac--Moody algebras \cite{Ka} and quantum affinizations \cite{H2}. Definition \ref{def: O} is a special feature of elliptic quantum groups. It is meant to loose the dependence on the dynamical parameter $\lambda$. \footnote{For the elliptic quantum group associated to an arbitrary finite-dimensional simple Lie algebra, Gautam--Toledano Laredo \cite[\S 2.3]{GTL2} defined a category of {\it integrable} modules on which the action of the elliptic Cartan currents, analogs of $\eD_k(z)$, is independent of $\lambda$. The asymptotic modules that we will construct in Section \ref{sec: asym} are not integrable.  }

 $\BGG$ is an abelian subcategory of $\tBGG$.  For $X$ in category $\BGG$, Eq.\eqref{def: elliptic diagonal} defines difference operators $K_l(z): X \longrightarrow X $ of bi-degree $(\epsilon_l,\epsilon_l)$ for $1\leq l \leq N$. \footnote{The $K_l(z)$ do not come from the elliptic quantum group, yet formally they are elliptic Cartan currents $K_l^+(z)$ in \cite[Corollary E.24]{K1}, arising from a Gauss decomposition of an $\hL$-matrix \cite{Ding-Frenkel}.  }

 Following \cite[Definition 2.1]{C}, a non-zero weight vector of a module $X$ in category $\tBGG$ is called {\it singular} if it is annihilated by the $L_{ij}(z)$ for $1\leq j < i \leq N$.   
  
 \begin{lem} \label{lem: K vs L}
Let $X$ be in category $\BGG$. If $v \in X$ is singular, then $K_i(z) v = \hL_{i}(z) v$ for all $1\leq i \leq N$.
 \end{lem}
 \begin{proof}
Descending induction on $i$: for $i = N$ we have $K_N(z) = L_{NN}(z) = \hL_N(z)$. Assume the statement for $i > N-t$ where $1\leq t < N$. We need to prove the case $i = N-t$. Let $\alpha$ be the weight of $v$ and let $Y$ be the submodule of $X$ generated by $v$. By \cite[Lemma 2.3]{C}, $Y$ is linearly spanned by vectors of the form
 $$ L_{p_1q_1}(z_1) L_{p_2q_2}(z_2) \cdots L_{p_nq_n}(z_n) v $$
where $1\leq p_l \leq q_l \leq N$ and $z_l \in \BC$ for $1\leq l \leq n$. So $\alpha + \epsilon_p- \epsilon_q \notin \wt(Y)$ for $1\leq p < q \leq N$, and any non-zero vector $\omega \in  Y[\alpha]$ is singular. Apply $\eD_k(z)$ to $\omega$. At the right-hand side of Eq.\eqref{def: determinant}  only the term $\sigma = \mathrm{Id}$ is non-zero and equal to
$\hL_{N}(z) \hL_{N-1}(z+\hbar) \cdots \hL_{N-k+1}(z+(k-1)\hbar) \omega$ by Eq.\eqref{defi: auxiliary L}.
It follows that 
\begin{align*}
\eD_{t+1}(z)  v &= \underline{\hL_N(z)} \hL_{N-1}(z+\hbar) \cdots \hL_{N-t+1}(z+(t-1)\hbar) \hL_{N-t}(z+t\hbar) v \\ 
&= K_N(z) \underline{\hL_{N-1}(z+\hbar)}  \cdots \hL_{N-t+1}(z+(t-1)\hbar) \hL_{N-t}(z+t\hbar) v \\
&= \cdots = K_N(z) K_{N-1}(z+\hbar) \cdots K_{N-t+1}(z+(t-1)\hbar) \hL_{N-t}(z+t\hbar) v.
\end{align*}
Here we applied the induction hypothesis to $N,N-1,\cdots,N-t+1$ successively to singular vectors to the right of the underlines. Since the $K_l(z)$ are invertible, in view of Eq.\eqref{def: elliptic diagonal} we must have $\hL_{N-t}(z+t\hbar) v = K_{N-t}(z+t\hbar) v$.
 \end{proof}

We extend the q-character theory of H. Knight and Frenkel--Reshetikhin to category $\BGG$, as in \cite[\S 3]{FZ}. Take the product group 
$ \CMw := \CM \times \Hlie$, by viewing $\Hlie$ as an additive group.
Let $\varpi: \CMw \twoheadrightarrow \Hlie$ be the projection to the second component.

As in \cite[\S 3.2]{HL}, let $\CMt$ be the set of formal sums  $\sum_{\Bf \in \CMw} c_{\Bf} \Bf$ with integer coefficients $c_{\Bf} \in \BZ$ such that: for $\mu \in \Hlie$, all but finitely many $c_{\Bf}$ with $\varpi(\Bf) = \mu$ is zero; the set $\{\varpi(\Bf): c_{\Bf} \neq 0 \}$ is contained in a finite union of cones $\nu + \BQ_-$ with $\nu \in \Hlie$.  
Make $\CMt$ into a ring: addition is the usual one of formal sums; multiplication is induced from that of $\CMw$. 

\begin{defi} \label{defi: q-char}
Let $X$ be in category $ \BGG$. For $\mu \in \wt(X)$, since $X[\mu]$ equipped with the difference operators $\eD_k(z)$ is in category $\CFm$, in the Grothendieck group of which we have $[X[\mu]] = \sum_{i=1}^{\dim X[\mu]} [\BM(\Bf^{\mu,i})]$
where $\Bf^{\mu,i} \in \CM$ for $1\leq i \leq \dim X[\mu]$. Each of the $(\Bf^{\mu,i}; \mu) \in \CMw$ is called an {\it e-weight} of $X$. Let $\ewt(X)$ be the set of e-weights of $X$. The {\it q-character} of $X$ is defined to be
$$ \qc(X) := \sum_{\mu \in \wt(X)} \sum_{i=1}^{\dim X[\mu]} (\Bf^{\mu,i}; \mu) \in \CMt.  $$
\end{defi}

\begin{prop} \label{prop: monoidal O}
Let $X,Y$ be in category $\BGG$. The $\CE$-module $X \wtimes Y$ is also in category $\BGG$ and $\qc(X\wtimes Y) = \qc(X)\qc(Y)$.
\end{prop}
\begin{proof}
Clearly $X \wtimes Y$ is in category $\tBGG$. Let us verify Property (M3) of Definition \ref{defi: merom eigen}.
The idea is almost the same as that of \cite[Prop.3.9]{FZ}, which in turn followed \cite[\S 2.4]{FR}. For $\alpha, \beta \in \Hlie$, let us choose ordered bases $(v_i^{\alpha})_{1\leq i \leq p_{\alpha}}$ and $(w_{j}^{\beta})_{1\leq j \leq q_{\beta}}$ for $X[\alpha]$ and $Y[\beta]$ respectively satisfying (M3). Note that $(v_i^{\alpha} \wtimes w_j^{\beta})_{\alpha,\beta,i,j}$ forms a basis $\mathcal{B}$ of $X \wtimes Y$. Choose a partial order $\unlhd$ on $\mathcal{B}$ with the property:
\begin{itemize}
\item[(a)] $v_{i}^{\alpha} \wtimes w_{j}^{\beta} \unlhd v_r^{\alpha} \wtimes w_s^{\beta}$ if $i \leq r$ and $j \leq s$;
\item[(b)] $v_i^{\alpha} \wtimes w_j^{\beta} \lhd v_r^{\gamma} \wtimes w_s^{\delta}$ if $\gamma \prec \alpha$ and $\beta \prec \delta$.
\end{itemize}
For $1 \leq k \leq N$, by Corollary \ref{cor: Jucys-Murphy}, $\eD_k^{X\wtimes Y}(z) (v_r^{\gamma} \wtimes w_s^{\delta})= \eD_k^X(z) v_r^{\gamma}\ \wtimes\ \eD_k^Y(z) w_s^{\delta} + Z$ where $Z$ is a finite sum of vectors in $X[\gamma +\varpi_{N-k} + \eta] \wtimes Y[\delta - \varpi_{N-k}- \eta]$ for $\eta \in \Hlie$ such that $-\varpi_{N-k} \prec \eta$. So the ordered basis $\mathcal{B}$ induces an upper triangular matrix for $\eD_k^{X \wtimes Y}(z)$ whose diagonal entry associated to $v_r^{\gamma} \wtimes w_s^{\delta}$ is the product of those associated to $v_r^{\gamma}$ and $w_s^{\delta}$. This implies (M3) for the weight spaces $(X\wtimes Y)[\alpha]$ with bases $\mathcal{B} \cap (X\wtimes Y)[\alpha]$ and the multiplicative formula of $q$-characters as well.
\end{proof}
For $f(z) \in \BM_{\BC}^{\times}$ and $\alpha \in \Hlie$ we make the simplifications 
\begin{equation*} 
f(z):=(f(z),\cdots,f(z);0),\quad e^{\alpha} := (1,\cdots,1;\alpha) \in \CMw.
\end{equation*}
\begin{defi} \label{defi: fund weights}
Let $1\leq i,k \leq N$ such that $i \neq N$. Set $\ell_k := \frac{N-k-1}{2}$.
For $a \in \BC$, define the following elements of $\CMw$:
\begin{gather*}
A_{i,a} := (\underbrace{1,\cdots,1}_{i-1}, \frac{\theta(z+(a-\ell_i)\hbar)}{\theta(z+(a-\ell_i-1)\hbar)}, \frac{\theta(z+(a-\ell_i)\hbar)}{\theta(z+(a-\ell_i+1)\hbar)},\ \underbrace{1,\cdots,1}_{N-i-1};\ \alpha_i); \\
\Psi_{k,a} := (\underbrace{\theta(z+(a-\ell_k)\hbar),\cdots,\theta(z+(a-\ell_k)\hbar)}_{k},\ \underbrace{1,\cdots,1}_{N-k};\ a\varpi_k); \\
Y_{k,a} := (\underbrace{\frac{\theta(z + (a-\ell_k+\frac{1}{2})\hbar)}{\theta(z + (a-\ell_k-\frac{1}{2})\hbar)}, \cdots, \frac{\theta(z + (a-\ell_k+\frac{1}{2})\hbar)}{\theta(z + (a-\ell_k-\frac{1}{2})\hbar)}}_k,\ \underbrace{1,\cdots, 1}_{N-k};\ \varpi_k); \\
\boxed{k}_a := (\underbrace{\frac{\theta(u+\hbar)\theta(u-\hbar)}{\theta(u)^2}, \cdots, \frac{\theta(u+\hbar)\theta(u-\hbar)}{\theta(u)^2}}_{k-1},\frac{\theta(u+\hbar)}{\theta(u)},\ \underbrace{1, \cdots, 1}_{N-k};\ \epsilon_k)|_{u = z + a\hbar}.   
\end{gather*}
\end{defi}
$A_{i,a}, Y_{k,a}$ and $\Psi_{k,a}$ are elliptic analogs of generalized simple roots, fundamental $\ell$-weight \cite{FR} and prefundamental weight \cite{HJ}. 
Set $c_{ij} := 2 \delta_{ij} - \delta_{i,j\pm 1}$ and $Y_{0,a} = \Psi_{0,a} := 1$. Then (in the products $1\leq j \leq N$)
\begin{gather}   \label{equ: A Psi}
Y_{k,a} = \frac{\Psi_{k,a+\frac{1}{2}}}{\Psi_{k,a-\frac{1}{2}}}, \quad A_{i,a} =  \prod_{j} \frac{\Psi_{j,a+\frac{c_{ij}}{2}}}{\Psi_{j,a-\frac{c_{ij}}{2}}} = Y_{i,a-\frac{1}{2}}Y_{i,a+\frac{1}{2}} \prod_{j=i\pm 1} Y_{j,a}^{-1}.
\end{gather}
The interplay of $A, \Psi$ is the source of the three-term Baxter's Relation \eqref{equ: asymptotic Baxter} in category $\BGG$. 
Note that $A, Y$ can also be written in terms of $\boxed{k}$:
\begin{gather} 
     A_{i,a} = \boxed{i}_{a-\ell_i}\boxed{i+1}_{a-\ell_i}^{-1},\quad Y_{k,a} = \prod_{j=1}^k \boxed{j}_{a-\ell_k-\frac{1}{2}+j-k}, \label{equ: tableau A Y} \\
     \Psi_{N,a} = \theta(z+(a+\frac{1}{2})\hbar), \quad Y_{N,a} = \frac{\theta(z+(a+1)\hbar)}{\theta(z+a\hbar)}.
\end{gather}
\subsection{Vector representations.} \label{ss: vector representation}
Let $\BV := \oplus_{i=1}^N \BM v_i$ with $\Hlie$-grading $\BV[\epsilon_i] = \BM v_i$. Rewriting Eq.\eqref{equ: dYBE elliptic} in the form of Eq.\eqref{rel: RLL explicit}, we obtain an $\CE$-module structure on $\BV$:
$$ L_{ij}(z) v_k =  \sum_{l=1}^N \frac{\theta(z+\hbar)}{\theta(z)} R^{jk}_{il}(z;\lambda) v_l. $$
The  factor $\frac{\theta(z+\hbar)}{\theta(z)}$ is used to simplify the q-character; see Eq.\eqref{equ: vect rep JM action}. 

If $i \leq N-k+1$, since $L_{pq}(z) v_i = 0$ for all $N\geq p > q > N-k$, only the term $\sigma = \mathrm{Id}$ in Eq.\eqref{def: determinant} survives and 
\begin{gather*}
 \eD_k(z) v_i = \frac{\mor(\Theta_k(\lambda))}{\mol(\Theta_k(\lambda))}   \prod^{N-k+1}_{j=N} L_{jj}(z+(N-j)\hbar) v_i  = g_k^i(z;\lambda) v_i, \\
 g_k^i(z;\lambda) = \prod_{j>N-k} \frac{\theta(\lambda_{ij}+\hbar)}{\theta(\lambda_{ij})} \quad \mathrm{for}\ i \leq N-k, \quad  g_k^{N-k+1}(z;\lambda) = \frac{\theta(z+k\hbar)}{\theta(z+(k-1)\hbar)}.
\end{gather*}
If $i > N-k+1$, then $L_{N-k+1,i}(z) v_{N-k+1} = \frac{\theta(\hbar)\theta(z+\lambda_{N-k+1,i})}{\theta(z)\theta(\lambda_{N-k+1,i})} v_i$. By Corollary \ref{cor: Jucys-Murphy}, $\eD_k(z) v_i = g_k^{N-k+1}(z;\lambda) v_i$. Let us perform a change of basis (see \cite[Eq.(E.2)]{K1}):
\begin{align*}
\tilde{v}_i := v_i \prod_{l>i} \theta(\lambda_{il}+\hbar) \in \BV [\epsilon_i]. 
\end{align*}
After a direct computation, we obtain:
\begin{gather}  \label{equ: vect rep JM action}
\eD_k(z) \tilde{v}_i = \tilde{v}_i \times 
\begin{cases}
1 & \mathrm{for}\ i \leq N-k,\\
\frac{\theta(z+k\hbar)}{\theta(z+(k-1)\hbar)} & \mathrm{for}\ i > N-k.
\end{cases}
\end{gather}
 The basis $\{\tilde{v}_1 < \tilde{v}_2 < \cdots < \tilde{v}_N\}$ of $\BV$ satisfies Property (M3) of Definition \ref{defi: merom eigen}, so $\BV$ is in category $\BGG$. For $a \in \BC$, let $\BV(a)$ be the pullback of $\BV$ by the spectral parameter shift $\Phi_a$ in Eq.\eqref{def: spectral shift}. Naturally $\BV(a)$ is in category $\BGG$; it is called a {\it vector representation}. Combining with Eq.\eqref{def: elliptic diagonal} we have:
\begin{align*} 
\qc(\BV(a)) &= \boxed{1}_a + \boxed{2}_a + \cdots + \boxed{N}_a. 
\end{align*}

\subsection{Highest weight modules.}  \label{ss- hw}
Let $X$ be  in category $\tBGG$. A non-zero weight vector $v \in X[\alpha]$ is called a {\it highest weight vector} if it is singular and $\hL_{k}(z) v = f_k(z) v$ for $1\leq k \leq N$; here the $f_k(z) \in \BM_{\BC}^{\times}$. Call $(f_1(z),f_2(z),\cdots, f_N(z); \alpha) \in \CMw$ the highest weight of $v$; by Lemma \ref{lem: K vs L} it belongs to $\ewt(X)$ if $X$ is in category $\BGG$. 

If there is a highest weight vector $v \in X[\alpha]$ of $X$ which also generates the whole module, then $X$ is called a {\it highest weight module}; see \cite[Definition 2.1]{C}. In this case, by \cite[Lemma 2.3]{C}, $X[\alpha] = \BM v$ and $\wt(X) \subseteq \alpha + \BQ_-$, so the highest weight vector is unique up to scalar product. This implies that $X$ admits a unique irreducible quotient. The highest weight of $v$ is also called the highest weight of $X$; it is of multiplicity one in $\qc(X)$ if $X$ is in category $\BGG$.

All irreducible modules in category $\BGG$ are of highest weight. 

By \cite[Theorem 2.8]{C}: two irreducible highest weight modules in category $\tBGG$ are isomorphic if and only if their highest weights are identical in $\CMw$;  all singular vectors of an irreducible highest weight  module in category $\tBGG$ are proportional. It follows that the q-characters distinguish irreducible modules in category $\BGG$. 

Let $\CR$ be the set of $\Bd \in \CMw$ which appears as the highest weight of an irreducible module in category $\BGG$. For $\Bd \in \CR$, let us fix an irreducible module $S(\Bd)$ in category $\BGG$ of highest weight $\Bd$.  Let $\CR_0$ (resp. $\CRf$) be the set of $\Bd \in \CR$  such that $S(\Bd)$ is one-dimensional (resp. finite-dimensional). 

We shall need the {\it completed} Grothendieck group $K_0(\BGG)$. Its definition is the same as that in \cite[\S 3.2]{HL}: elements are formal sums  $\sum_{\Bd \in \CR} c_{\Bd} [S(\Bd)]$ with integer coefficients $c_{\Bd} \in \BZ$ such that $\oplus_{\Bd} S(\Bd)^{\oplus |c_{\Bd}|}$ is in category $\BGG$; addition is the usual one of formal sums. As in the case of Kac--Moody algebras \cite[\S 9.6]{Ka}, for $\Bd \in \CR$ the multiplicity $m_{\Bd,X}$ of $S(\Bd)$ in any object $X$ of category $\BGG$ is well-defined due to Definition \ref{def: O}, and $[X] := \sum_{\Bd} m_{\Bd,X}[S(\Bd)]$ belongs to $K_0(\BGG)$. In the case $X = S(\Bd)$ the right-hand side is simply $[S(\Bd)]$ as $m_{\Be,S(\Bd)} = \delta_{\Bd,\Be}$ for $\Be \in \CR$.

By Proposition \ref{prop: monoidal O}, $K_0(\BGG)$ is endowed with a ring structure with multiplication $[X][Y] = [X \wtimes Y]$ for $X, Y$ in category $\BGG$. Together with Definition \ref{defi: q-char}, we obtain
\begin{cor} \label{cor: injectivity}
 The assignment $[X] \mapsto \qc(X)$ defines an injective morphism of rings $\qc: K_0(\BGG) \longrightarrow \CMt$. In particular, $K_0(\BGG)$ is commutative. 
\end{cor}

Let $\BGGf$ be the full subcategory of $\BGG$ consisting of finite-dimensional modules. It is abelian and monoidal. Its Grothendieck ring $K_0(\BGGf)$ admits a $\BZ$-basis $[S(\Bd)]$ for $\Bd \in \CRf$, and is commutative as a subring of $K_0(\BGG)$.

By Proposition \ref{prop: monoidal O}, $S(\Bd)\ \wtimes\ S(\Be)$ admits an irreducible sub-quotient $S(\Bd \Be)$, so the three sets $\CR \supset \CRf \supset \CR_0$ are sub-monoids of $\CMw$.
\begin{lem}  \label{lem: highest weight condition}
Let $\Bd = ((f_k(z))_{1\leq k \leq N}; \mu) \in \CMw$. 
\begin{itemize}
\item[(i)] Suppose $\Bd \in \CR$. Then for $1\leq k < N$ we have 
$$ \frac{f_k(z)}{f_{k+1}(z)} = c \prod_{l=1}^n \frac{\theta(z+a_l\hbar)}{\theta(z+b_l\hbar)},\quad \mu_{k,k+1} = \sum_{l=1}^n (a_l-b_l) $$
for certain  $a_1, a_2, \cdots, a_n, b_1, b_2, \cdots, b_n \in \BC$ and $c \in \BC^{\times}$.
\item[(ii)] If $\Bd \in \CRf$, then (i) holds and after a rearrangement of the $a_l,b_l$ we have $a_l - b_l \in \BZ_{\geq 0}+\hbar^{-1}\Gamma$ for all $l$.
\item[(iii)] $\Bd \in \CR_0$ if and only if (ii) holds with $a_l-b_l \in \hbar^{-1}\Gamma$ for all $l$.
\end{itemize}
\end{lem}
\begin{proof}
(i) and (iii) are essentially \cite[Theorems 6 \& 9]{FV1}, which can be proved as in \cite[Theorem 4.1]{FZ} by replacing $L_{+-},L_{-+}$ therein with $L_{k,k+1},L_{k+1,k}$. (ii) comes from either \cite[Theorem 5.1]{C} or \cite[Corollary 4.6]{FZ}.
\end{proof}

As examples $Y_{N,a}, \Psi_{N,a} \in \CR_0$.
 Call an e-weight $\Be \in \CMw$ {\it dominant} (resp. {\it rational}) if $\Be = \Bd\Bm$ where $\Bd \in \CR_0$ and $\Bm$ is a product of the $Y_{i,a}$ (resp. the $\Psi_{i,a}\Psi_{i,b}^{-1}$) with $a,b \in \BC$ and $1\leq i \leq N$. Lemma \ref{lem: highest weight condition} implies that all elements of $\CRf$ (resp. $\CR$) are dominant (resp. rational). 
 
 \begin{theorem}\cite{C} \label{thm: fund module}
$\CRf$ is the set of dominant e-weights. 
\end{theorem}
\begin{proof}
It suffices to prove $Y_{n,a} \in \CRf$ for $1\leq n < N$. Note that $V(w)$ and $ \gamma$ from \cite[Eq.(1.19)]{C} correspond to our $\BV(-\frac{w}{\hbar})\ \wtimes\ S(\frac{\theta(z-w)}{\theta(z-w-\hbar)})$ and $-\hbar$. Let us rephrase \cite[Theorem 4.4]{C} in terms of the $\BV$ by replacing $z,w$ in {\it loc.cit.} with $-a\hbar, z$. 

The $\CE$-module $ \BV(a)\ \wtimes\ \BV(a+1)\ \wtimes \cdots \wtimes\ \BV(a+n-1)$
admits an irreducible quotient $S$ which  contains a singular vector $\omega$ of weight $\varpi_n$ such that $\hL_{k}(z) \omega = \Lambda_k(z) g_k(\lambda) \omega  $ where for $1\leq k \leq N$ (set $\delta_{k\leq n} = 1$ if $1\leq k \leq n$ and $\delta_{k\leq n} = 0$ if $n < k \leq N$):
$$ \Lambda_k(z) =  \frac{\theta(z+(a+1)\hbar)}{\theta(z+a\hbar)} \frac{\theta(z+(a+n)\hbar)}{\theta(z+(a+n-\delta_{k\leq n})\hbar)},\quad g_k(\lambda) \in \BM^{\times}.  $$
As a sub-quotient of tensor products of vector representations, $S$ belongs to category $\BGG$. By Lemma \ref{lem: category F}, the $g_k(\lambda)$ can be gauged away, and the highest weight of $S$ is $\Lambda_N(z) Y_{n,a-1+\frac{N+n}{2}} \in \CRf$. This implies $Y_{n,a-1+\frac{N+n}{2}} \in \CRf$.
\end{proof}
 A sharp difference from the affine case \cite[Theorem 3.11]{HJ} is that category $\BGG$ does not admit prefundamental modules, i.e. $\Psi_{r,a} \notin \CR$ if $r < N$. One might want to introduce a larger category with well-behaved q-character theory, so that modules of highest weight $\Psi_{r,a}$ exist. For this purpose, the finite-dimensionality of weight spaces should be dropped because of \cite[Theorem 9]{FV1}. The recent work \cite{B} on representations of affine quantum groups is in this direction.
 
\subsection{Young tableaux and q-character formula} \label{ss: tableau}
Let $\SP$ be the set partitions with at most $N$ parts, i.e. $N$-tuples of non negative integers  $(\mu_1 \geq \mu_2 \geq \cdots \geq \mu_N)$. To such a partition we associate a Young diagram 
$$ Y_{\mu} := \{ (i,j) \in \BZ^2 \ |\ 1 \leq i \leq N,\ 1 \leq j \leq \mu_i \}, $$
and the set $\SB_{\mu}$ of Young tableaux of shape $Y_{\mu}$. We put the Young diagram at the northwest position so that $(i,j) \in Y_{\mu}$ corresponds to the box at the $i$-th row (from bottom to top) and $j$-th column (from right to left). By a tableau we mean a function $T: Y_{\mu} \longrightarrow \{1< 2 < \cdots < N \}$ weakly increasing at each row (from left to right) and strictly increasing at each column (from top to bottom).

For $\mu = (\mu_1 \geq \mu_2 \geq \cdots \geq \mu_N) \in \SP$ and $a \in \BC$, we have the dominant e-weight
$$ \theta_{\mu,a} := \left(\frac{\theta(z+(a+\mu_1)\hbar)}{\theta(z+a\hbar)}, \frac{\theta(z+(a+\mu_2)\hbar)}{\theta(z+a\hbar)}, \cdots, \frac{\theta(z+(a+\mu_N)\hbar)}{\theta(z+a\hbar)}; \sum_{j=1}^N \mu_j \epsilon_j \right).  $$
The associated irreducible module in category $\BGGf$ is denoted by $S_{\mu,a}$. 
\begin{theorem}  \label{thm: q-char evaluation}
Let $\mu \in \SP$ and $a \in \BC$. For the $\CE_{\tau,\hbar}(\mathfrak{sl}_N)$-module $S_{\mu,a}$ we have
\begin{equation}  \label{equ: q-char evaluation}
\qc(S_{\mu,a}) = \sum_{T \in \SB_{\mu}} \prod_{(i,j)\in Y_{\mu}} \boxed{T(i,j)}_{a+j-i} \in \CMt.
\end{equation}
\end{theorem}
For $\nu = (1\geq 0 \geq 0 \geq \cdots \geq 0)$, we have $S_{\nu, a} \cong \BV(a)$, and Eq.\eqref{equ: q-char evaluation} specializes to the q-character formula in Section \ref{ss: vector representation}. As an illustration of the theorem, let $N = 3$ and $\mu = (2\geq 1 \geq 0)$. Pictorially $\SB_{\mu}$ consists of:
$$\young(:1,12),\ \young(:1,13),\ \young(:1,22),\ \young(:1,23),\ \young(:1,33), \  \young(:2,13),\ \young(:2,23),\ \young(:2,33). $$
The fourth tableau gives rise to the term $\boxed{2}_{a+1}\boxed{3}_a \boxed{1}_{a-1}$ in $\qc(S_{\mu,a})$.

\begin{rem}
Theorem \ref{thm: q-char evaluation} is an elliptic analog of the q-character formula for affine quantum groups \cite[Lemma 4.7]{FM}. In principle it can be deduced from the functor of Gautam--Toledano Laredo \cite[\S 6]{GTL2}. This is a functor from finite-dimensional representations of affine quantum groups to those of elliptic quantum groups (including our $S_{\mu,a}$), and it respects affine and elliptic q-characters.
\end{rem}

The proof of Theorem \ref{thm: q-char evaluation} will be given in Section \ref{ss: proof q-char}. It is in the spirit of \cite{FM}, based on small elliptic quantum groups of Tarasov--Varchenko \cite{TV}.
\section{Small elliptic quantum group and evaluation modules}  \label{sec: evaluation}
The aim of this section is to prove Corollary \ref{cor: Jucys-Murphy} and Theorem \ref{thm: q-char evaluation}. 

Recall that $\Hlie$ is the $\BC$-vector space generated by the $\epsilon_i$ for $1\leq i \leq N$ subject to the relation $\epsilon_1+\epsilon_2+\cdots + \epsilon_N = 0$. For $1\leq k \leq N$, define the $\BC$-vector space $\Hlie_k$ to be the quotient of $\Hlie$ by $\epsilon_1 = \epsilon_2 = \cdots = \epsilon_{N-k} = 0$. (By convention $\Hlie_N = \Hlie$.) The quotient $\Hlie \twoheadrightarrow \Hlie_k$ induces an embedding $\BM_{\Hlie_k} \hookrightarrow \BM$. 

Let $\CE^{\Hlie}_k$ (resp. $\CE_k$) be the $\Hlie$-algebra (resp. $\Hlie_k$-algebra)  generated by the $L_{ij}(z)$ for $N-k < i,j \leq N$ subject to Relation \eqref{rel: RLL explicit} with summations $N-k< p,q \leq N$. (This makes sense because the $R_{ij}^{pq}(z;\lambda)$ for $N-k < i,j,p,q \leq N$ belong to $\BM_{\Hlie_k}$.) The following defines an $\Hlie_k$-algebra morphism
$$ \Delta_k: \CE_k \longrightarrow \CE_k\ \dt\ \CE_k,\quad   L_{ij}(z) \mapsto \sum_{p=N-k+1}^N L_{ip}(z)\ \dt\ L_{pj}(z).   $$
 One has natural algebra morphisms $\CE_k \longrightarrow \CE_k^{\Hlie} \longrightarrow \CE$ sending $L_{ij}(z)$ to itself; the second is an $\Hlie$-algebra morphism.
 $\eD_1(z), \eD_2(z),\cdots, \eD_{k}(z)$ from Eq.\eqref{def: determinant} are well-defined in $\CE_{k}^{\Hlie}$ and $\CE_{k}$. Their images in $\CE$ are the first $k$  elliptic quantum minors.
\subsection{Proof of Corollary \ref{cor: Jucys-Murphy}.} \label{ss: proof of Cor}
The $\Hlie_k$-algebra with coproduct $(\CE_k, \Delta_k)$ is isomorphic to the usual elliptic quantum group $\CE_{\tau,\hbar}(\mathfrak{sl}_k)$; here we view $\Hlie_k$ as a Cartan subalgebra of $\mathfrak{sl}_k$ so that $\CE_{\tau,\hbar}(\mathfrak{sl}_k)$ is an $\Hlie_k$-algebra. Under this isomorphism, by Eq.\eqref{def: determinant}, $\eD_k(z) \in \CE_k$ corresponds to the $k$-th elliptic quantum minor of $\CE_{\tau,\hbar}(\mathfrak{sl}_k)$. So Theorem \ref{thm: determinant} can be applied to $(\CE_k, \eD_k(z),  \Delta_k)$ and then to the algebra morphism $\CE_k \longrightarrow \CE$. The first statement of the corollary is obvious, and the second is based on the fact that for $i,j > N-k$ the difference $\Delta - \Delta_k$ at $L_{ij}(z)$ is a finite sum over $\alpha \in \Hlie$ of elements in $\CE_{\epsilon_{i},\alpha}\ \dt\  \CE_{\alpha,\epsilon_{j}}$ with $\epsilon_{N-k+1} \prec \alpha$ and so $\epsilon_i,\epsilon_j \prec \alpha$.   \hfill $\Box$

We believe $0\neq \alpha+\varpi_{N-k} \in \BQ_+$ in Corollary \ref{cor: Jucys-Murphy}, as in \cite[\S 7]{Damiani} and \cite[\S 3]{Z1}.

\subsection{Small elliptic quantum group of Tarasov--Varchenko \cite{TV} } \label{ss: small}
 Let us define the linear form $\lambda_i \in \Hlie^*$ of taking $i$-th component for $1\leq i \leq N$:
$$ x_1 \epsilon_1 + x_2\epsilon_2 + \cdots + x_N \epsilon_N \mapsto x_i - \frac{1}{N}(x_1+x_2+\cdots + x_N). $$
The linear form $\lambda_{ij}$ of Section \ref{sec: elliptic} is  $\lambda_i - \lambda_j$.
For $\gamma \in \Hlie$ and $1\leq i, j \leq N$, set $\gamma_i := \lambda_i(\gamma)$ and $\gamma_{ij} := \gamma_i - \gamma_j$ as {\it complex numbers}. We hope this is not to be confused with the previously defined {\it vectors} $\lambda_i \in \Hlie^*$ and $\epsilon_i, \alpha_i, \varpi_i \in \Hlie$.

 Following \cite[\S 3]{TV}, let $\BM_{2}$ be the ring of meromorphic functions $f(\lambda^{\{1\}},\lambda^{\{2\}})$ of $(\lambda^{\{1\}},\lambda^{\{2\}}) \in \Hlie \oplus \Hlie$ whose location of singularities in $\lambda^{\{1\}}$ does not depend on $\lambda^{\{2\}}$ and vice versa. For brevity, we write $f(\lambda^{\{1\}})$ or $f(\lambda^{\{2\}})$ instead of $f(\lambda^{\{1\}},\lambda^{\{2\}})$ if the function does not depend on the other variable.

\begin{defi}\cite{TV} \label{defi: small elliptic}
The {\it small elliptic quantum group} $\mathfrak{e} := \SE_{\tau,\hbar}(\mathfrak{sl}_N)$ is the algebra with generators $\BM_2$ and $t_{ij}$ for $1\leq i,j \leq N$ and subject to relations: $\BM_2$ is a subalgebra; for $f(\lambda^{\{1\}},\lambda^{\{2\}}) \in \BM_2$ and $1\leq i,j,k,l \leq N$,
\begin{gather*}
t_{ij}f(\lambda^{\{1\}},\lambda^{\{2\}}) = f(\lambda^{\{1\}}+\hbar \epsilon_i,\lambda^{\{2\}}+ \hbar \epsilon_j) t_{ij},\quad  
t_{ij}t_{ik} = t_{ik}t_{ij},  \\
t_{ik}t_{jk} = \frac{\theta(\lambda_{ij}^{\{1\}}-\hbar)}{\theta(\lambda_{ij}^{\{1\}}+\hbar)} t_{jk}t_{ik} \quad \mathrm{for}\ i \neq j, \\
\frac{\theta(\lambda_{jl}^{\{2\}}-\hbar)}{\theta(\lambda_{jl}^{\{2\}})} t_{ij}t_{kl} - \frac{\theta(\lambda_{ik}^{\{1\}}-\hbar)}{\theta(\lambda_{ik}^{\{1\}})} t_{kl}t_{ij} = \frac{\theta(\lambda_{ik}^{\{1\}}+ \lambda_{jl}^{\{2\}})\theta(-\hbar)}{\theta(\lambda_{ik}^{\{1\}}) \theta(\lambda_{jl}^{\{2\}})} t_{il}t_{kj}
\end{gather*}
for $i \neq k$ and $j \neq l$. Here $\lambda^{\{1\}}_{ij} = \lambda_i^{\{1\}} - \lambda_j^{\{1\}}$ and $\lambda^{\{2\}}_{ij} = \lambda_i^{\{2\}} - \lambda_j^{\{2\}}$. 
\end{defi}

$\SE$ is equipped with an $\Hlie$-algebra structure: elements of $\BM_2$ are of bi-degree $(0,0)$; $t_{ij}$ is of bi-degree $(\epsilon_j,\epsilon_i)$; the moment maps are given by
$$ \mol(g(\lambda)) = g(\lambda^{\{2\}}),\quad \mor(g(\lambda)) = g(\lambda^{\{1\}}). $$
Let $X$ be an object of $\CVf$. A representation $\rho$ of $\SE$ on $X$ is a morphism of $\Hlie$-algebras $\rho: \SE \longrightarrow \BD^X$ such that for $f(\lambda^{\{1\}},\lambda^{\{2\}}) \in \BM_2$ and $v \in X[\gamma]$,
$$ \rho(f(\lambda^{\{1\}},\lambda^{\{2\}})): v \mapsto f(\lambda,\lambda+\hbar \gamma) v. $$
A morphism of two representations $(\rho,X)$ and $(\sigma,Y)$ is a morphism $\Phi: X \longrightarrow Y$ in $\CVf$ such that $\Phi \rho(t_{ij}) = \sigma(t_{ij}) \Phi$ for $1\leq i,j \leq N$. Let $\rep$ be the category of $\SE$-modules. 
The following result is \cite[Corollary 3.4]{TV}.

\begin{cor} \label{cor: evaluation module}
Let $(\rho,X)$ be a representation of $\SE$ on $X$. Then for $a \in \BC$,
$$L_{ij}(z) \mapsto \frac{\theta(z+a\hbar + \lambda_i^{\{2\}}-\lambda_j^{\{1\}})}{\theta(z+a\hbar)}\rho(t_{ji})$$
 defines a representation of $\CE$ on $X$, called the evaluation module $X(a)$.
\end{cor}
There is a flip of the subscripts $i, j$ because the bi-degrees of $L_{ij}$ and $t_{ij}$ are flips of each other. See also \cite[Eq.(3.6)]{TV} where $\mathcal{T}_{ij}(u)$ comes from $t_{ji}$.

$X \mapsto X(a)$ defines a functor $\ev_a: \rep \longrightarrow \Rep$. Let $\CF$ be the full subcategory of $\rep$ whose objects are finite-dimensional $\SE$-modules $X$ with $X(x)$ being in category $\BGG$. Then $\ev_a$ restricts to a functor of abelian categories $\CF \longrightarrow \BGGf$, and induces an injective morphism of Grothendieck groups $K_0(\CF) \hookrightarrow K_0(\BGGf)$.

For $1\leq k \leq N$, define $\hat{t}_{k} \in \SE$ in the same way as  Eq.\eqref{defi: auxiliary L}: \footnote{ The $\hat{t}_a$ are  slightly different from the $\hat{t}_{aa}$ in  \cite[Eq.(4.1)]{TV}. Yet they play the same role. }
\begin{equation*} 
\hat{t}_N(z) := t_{NN},\quad \hat{t}_k(z) = t_{kk} \prod_{j=k+1}^{N}  \frac{\mor(\theta(\lambda_{kj}))}{\mol(\theta(\lambda_{kj}))}.
\end{equation*} 
Let $\mu \in \Hlie$. There exists a unique (up to isomorphism) irreducible $\SE$-module $V_{\mu}$ with the property: $V_{\mu}$ admits a non-zero vector $v$ of weight $\mu$ such that $\hat{t}_{k} v = v,\ t_{ij} v = 0$ for $1\leq i,j,k \leq N$ and $j < k$; it is called {\it standard} in \cite[\S 4]{TV}. Let $L_{\mu}$ denote the complex irreducible module over the simple Lie algebra $\mathfrak{sl}_N$ of highest weight $\mu$. For $\nu \in \Hlie$, let $d_{\mu}[\nu] = \dim_{\BC} L_{\mu}[\nu]$ where $L_{\mu}[\nu]$ is the weight space of weight $\nu$.

\begin{theorem}\cite[Theorem 5.9]{TV}    \label{thm: small}
The $\SE$-module $V_{\mu}$ is finite-dimensional if and only if $\mu_{ij} \in \BZ_{\geq 0} + \hbar^{-1} \Gamma$ for $1\leq i < j \leq N$.  If $\tilde{\mu} \in \Hlie$ is such that $\mu_{ij} - \tilde{\mu}_{ij} \in \hbar^{-1} \Gamma$ and $\tilde{\mu}_{ij} \in \BZ_{\geq 0}$ for $i < j$, then $\dim V_{\mu}[\mu+\gamma] = d_{\tilde{\mu}}[\tilde{\mu}+\gamma]$ for $\gamma \in \BQ_-$.
\end{theorem}
In the theorem $\tilde{\mu}$ is uniquely determined by $\mu$ since $\BZ \cap h^{-1}\Gamma = \{0\}$. Such an $\SE$-module $V_{\mu}$ is in category $\CF$. Indeed, the evaluation module module $V_{\mu}(a)$ is irreducible in category $\tBGG$ of highest weight 
$$\left(\frac{\theta(z+(\mu_1+a)\hbar)}{\theta(z+a\hbar)}, \frac{\theta(z+(\mu_2+a)\hbar)}{\theta(z+a\hbar)}, \cdots, \frac{\theta(z+(\mu_N+a)\hbar)}{\theta(z+a\hbar)}; \mu \right). $$
One checks that such an e-weight is dominant. So $V_{\mu}(a)$ is in category $\BGGf$ by Theorem \ref{thm: fund module}.
 The {\it character} $\chi(V_{\mu})$ of $V_{\mu}$ is  $\sum_{\gamma} d_{\tilde{\mu}}[\tilde{\mu}+\gamma]e^{\mu+\gamma} \in \CMt$.

The isomorphism classes $[V_{\mu}]$ where $\mu \in \Hlie$ and $\mu_{ij} \in \BZ_{\geq 0}+\hbar^{-1}\Gamma$ for $i < j$ form a $\BZ$-basis of $K_0(\CF)$, and $[V_{\mu}] \mapsto \chi(V_{\mu})$ extends uniquely to a morphism of abelian groups $\chi: K_0(\CF) \longrightarrow \CMt$, which
 is injective thanks to the linear independence of characters of irreducible representations of the simple Lie algebra $\mathfrak{sl}_N$.

\subsection{Category $\BGGf'$.} \label{ss: branching}
We are going to prove Theorem \ref{thm: q-char evaluation} by induction on $N$. The idea is to view the irreducible $\CE$-module $S_{\mu,a}$ as an $\CE_{N-1}^{\Hlie}$-module and to apply the induction hypothesis. For this purpose, we need to adapt carefully the definitions of finite-dimensional module category $\BGGf$ and its q-characters in  Section \ref{ss-repr} to $\CE_{N-1}^{\Hlie}$. To distinguish with $\CE$ and to simplify notations, we shall add a prime (instead of the index $N-1$) to objects related to $\CE_{N-1}^{\Hlie}$. Notably $\Hlie' := \Hlie_{N-1}$.

We define category $\BGGf'$. An object is a {\it finite-dimensional} $\Hlie$-graded vector space $X$ (viewed as an object of category $\CVf$) endowed with difference operators $L_{ij}^X(z): X \longrightarrow X$ of bi-degree $(\epsilon_i,\epsilon_j)$ for $2\leq i,j \leq N$ depending on $z \in \BC$ such that:
\begin{itemize}
    \item[(M1')]  there exists a basis of $X$ with respect to which the matrix entries of the difference operators $L_{ij}^X(z)$ are meromorphic functions of $(z,\lambda) \in \BC \times \Hlie$;
    \item[(M2')] $L_{ij}(z) \mapsto L_{ij}^X(z)$ defines an $\Hlie$-algebra morphism $\CE_{N-1}^{\Hlie} \longrightarrow \BD^X$;
    \item[(M3')]  $X$ admits an ordered weight basis with respect to which the matrices of the difference operators $\eD_l^X(z)$ for $1\leq l < N$ are upper triangular and their diagonal entries are non-zero meromorphic functions of $z \in \BC$.
\end{itemize}
 A morphism in category $\BGGf'$ a linear map $\Phi: X \longrightarrow Y$ such that $\Phi L_{ij}^X(z) = L_{ij}^Y(z) \Phi$ for $2\leq i,j \leq N$. Category $\BGGf'$ is an abelian subcategory of $\CVf$.

 The $\Hlie$-algebra morphism $\CE_{N-1}^{\Hlie} \longrightarrow \CE$ induces restriction functor $\BGGf \longrightarrow \BGGf'$. 

Let $X$ be in category $\BGGf'$. Eq.\eqref{def: elliptic diagonal} defines difference operators $K_l^X(z): X \longrightarrow X$ of bi-degree $(\epsilon_l,\epsilon_l)$ for $2\leq l \leq N$. Condition (M3') implies that for each weight $\alpha$, the weight space $X[\alpha]$ admits an ordered basis $B_{\alpha}$ with respect to which the matrix of $K_l^X(z)$ is upper triangular and has as diagonal entries $f_{b,l}(z) \in \BM_{\BC}^{\times}$ for $b \in B_{\alpha}$. Following Definition \ref{defi: q-char}, we define the q-character of $X$ to be 
$$ \qc'(X) = \sum_{\alpha \in \wt(X)} \sum_{b \in B_{\alpha}} (1,f_{b,2}(z),f_{b,3}(z), \cdots, f_{b,N}(z); \alpha) \in \CMt. $$
It is independent of the choice of the bases $B_{\alpha}$, as  one can use category $\CFm$ to characterize the $f_{b,l}(z)$; see the comments after Lemma \ref{lem: category F}.

\begin{rem} \label{rem: from N to N-1}
 Let $X$ be in category $\BGGf$, viewed as an object of category $\BGGf'$. Then $\qc'(X)$ is obtained from $\qc(X)$ by replacing each e-weight $\Bg$ of the $\CE$-module $X$ with $\Bg'$; here for $\Bg = (g_1(z),g_2(z),\cdots,g_N(z);\alpha) \in \CMt$ we define
$$ \Bg' := (1,g_2(z),g_3(z),\cdots,g_N(z);\alpha) \in \CMt.  $$
Reciprocally, 
if $X$ is an irreducible $\CE$-module in category $\BGGf$ of highest weight $(e_1(z),e_2(z),\cdots,e_N(z);\alpha) \in \CRf$, then $\qc(X)$ can be recovered from $\qc'(X)$. Indeed,  since the $N$-th elliptic quantum minor is central, by Schur Lemma, it acts on $X$ as a scalar. Each e-weight $(f_1(z),f_2(z),\cdots,f_N(z);\beta)$ of the $\CE$-module $X$ is  determined by the its last $N$ components in $\qc'(X)$ as follows:
$$ e_1(z+(N-1)\hbar) e_2(z+(N-2)\hbar) \cdots e_N(z) = f_1(z+(N-1)\hbar) f_2(z+(N-2)\hbar) \cdots f_N(z). $$
\end{rem}
The highest weight theory in Section \ref{ss- hw} carries over to category $\BGGf'$ since $\hL_k(z) \in \CE_{N-1}^{\Hlie}$ for $2\leq k \leq N$. Irreducible objects in $\BGGf'$ are classified by their highest weight, and the q-character map is an injective morphism from the Grothendieck group $K_0(\BGGf')$ to the additive group $\CMt$.
Let $\SP'$ be the set of partitions with at most $N-1$ parts $(\nu_2 \geq \nu_3 \geq \cdots \geq \nu_N)$. For such a partition and for $c, a \in \BC$, 
$$ \left(1,\frac{\theta(z+(a+\nu_2)\hbar)}{\theta(z+a\hbar)}, \frac{\theta(z+(a+\nu_3)\hbar)}{\theta(z+a\hbar)}, \cdots, \frac{\theta(z+(a+\nu_N)\hbar)}{\theta(z+a\hbar)}; c\epsilon_1 + \sum_{j=2}^N \nu_j \epsilon_j \right)  $$
is the highest weight of an irreducible $\CE_{N-1}^{\Hlie}$-module in category $\BGGf'$, which is denoted by $S'_{\nu,c,a}$.
As in Section \ref{ss: tableau}, $\nu$ is identified with its Young diagram $Y_{\nu}$. Let $\SB_{\nu}'$ be the set of Young tableaux $Y_{\nu} \longrightarrow \{2<3<\cdots < N\}$ of shape $\nu$.

\begin{lem}  \label{lem: (k) vs k}
Assume that Theorem \ref{thm: q-char evaluation} is true for $\CE_{\tau,\hbar}(\mathfrak{sl}_{N-1})$-modules. Then for $\nu \in \SP'$ and $c, a \in \BC$, the q-character of the $\CE_{N-1}^{\Hlie}$-module $S_{\nu,c,a}'$ is 
$$ \qc'(S_{\nu,c,a}') = e^{c \epsilon_1}\sum_{T \in \SB_{\nu}'} \prod_{(i,j)\in Y_{\nu}} \boxed{T(i,j)}_{a+j-i}'  \in \CMt. $$
\end{lem}
\begin{proof}
We shall need $\CE_{N-1}$-modules which are $\Hlie'$-graded $\BM_{\Hlie'}$-vector spaces; similar category of finite-dimensional modules and q-characters are defined, based on the $\Hlie'$-algebra isomorphism $\CE_{\tau,\hbar}(\mathfrak{sl}_{N-1}) \cong \CE_{N-1}$ in Section \ref{ss: proof of Cor}. 

For $\nu := (\nu_2 \geq \nu_3 \geq \cdots \geq \nu_N) \in \SP'$ and $a \in \BC$, there exists a unique (up to isomorphism) irreducible $\CE_{N-1}$-module, denoted by $S_{\nu,a}'$, which contains a non-zero vector $\omega$ of $\Hlie'$-weight $\nu_2 \epsilon_2 + \nu_3 \epsilon_3 + \cdots + \nu_N \epsilon_N$ such that 
$$ L_{ij}(z) \omega = 0,\quad \hL_k(z) \omega = \frac{\theta(z+(a+\nu_k)\hbar)}{\theta(z+a\hbar)} \omega $$
for $2\leq i,j,k \leq N$ with $j < i$. We endow the $\BM$-vector space $X := \BM \otimes_{\BM_{\Hlie'}} S_{\nu,a}'$ with an $\CE_{N-1}^{\Hlie}$-module structure in category $\BGGf'$. 

Let $w$ be a non zero weight vector in $S_{\nu,a}'$. Its $\Hlie'$-weight is written uniquely in the form $(\nu_2 \epsilon_2 + \nu_3 \epsilon_3 + \cdots + \nu_N \epsilon_N) + (x_2 \alpha_2 + x_3\alpha_3 + \cdots + x_{N-1}\alpha_{N-1}) \in \Hlie'$, where $x_j \in \BZ_{\leq 0}$. Define the $\Hlie$-weight of $g(\lambda) \otimes_{\BM_{\Hlie'}} w$, for $g(\lambda) \in  \BM^{\times}$, to be 
$$(c\epsilon_1 + \nu_2 \epsilon_2 + \nu_3 \epsilon_3 + \cdots + \nu_N \epsilon_N) + (x_2 \alpha_2 + x_3\alpha_3 + \cdots + x_{N-1}\alpha_{N-1}) \in \Hlie,  $$
and define the action of $L_{ij}(z)$ for $2\leq i,j \leq N$ by the formula
$$ L_{ij}(z) (g(\lambda)\otimes_{\BM_{\Hlie'}} w) = g(\lambda+\hbar \epsilon_j) \otimes_{\BM_{\Hlie'}} L_{ij}(z) w. $$
(M1')--(M2') are clear from the $\CE_{N-1}$-module structure on $S_{\nu,a}'$. Choose an ordered weight basis $B$ of $S_{\nu,a}'$ over $\BM_{\Hlie'}$ such that the matrices of $\eD_k(z)$ for $1\leq k < N$ are upper triangular and their diagonal entries belong to $\BM_{\BC}^{\times}$. Then the ordered basis $\{1 \otimes_{\BM_{\Hlie'}} b\ |\ b \in B\} =: B'$  of $X$ satisfies (M3'). So $X$ is in category $\BGGf'$. 

The matrices of $\eD_k(z)$ with respect to the basis $B'$ of $X$ and the basis $B$ of $S_{\nu,a}'$ are the same. So $\qc'(X)$ up to a normalization factor $e^{c\epsilon_1}$, is equal to the q-character of the $\CE_{\tau,\hbar}(\mathfrak{sl}_{N-1})$-module $S_{\nu,a}$. The latter  is given by Eq.\eqref{equ: q-char evaluation}.

$X$ has a unique (up to scalar) singular vector and is of highest weight, so it is irreducible. A comparison of highest weights shows that $X \cong S_{\nu,c,a}'$.
\end{proof}
Fix $\mu \in \SP$ a partition with at most $N$ parts. Given a tableau 
 $T \in \SB_{\mu}$, by deleting the boxes $\boxed{1}$ in $T$, we obtain a Young diagram $T^{-1}(\{2,3,\cdots,N\})$ with at most $N-1$ rows, which corresponds to a partition in $\SP'$, denoted by $\nu_T$. 
Let $\SW_{\mu}$ be the set of all such $\nu_T$ with $T \in \SB_{\mu}$. For $\nu \in \SW_{\mu}$, define $c_{\nu}$ to be the cardinal of the finite subset $Y_{\mu} \setminus Y_{\nu}$ of $\BZ^2$.

Again take the example $N = 3$ and $\mu = (2 \geq 1)$ after Theorem \ref{thm: q-char evaluation}. The eight tableaux in $\SB_{\mu}$ with $\boxed{1}$ deleted give four Young diagrams and partitions
$$\young(~) = (1),\quad \young(~~) = (2), \quad  \young(~,~) = (1\geq 1),\quad \young(:~,~~) = (2\geq 1). $$
The corresponding integers $c_{\nu}$ are $2, 1, 1, 0$.

\begin{lem} \label{lem: branching}
Let $\mu \in \SP$ and $a \in \BC$. In the Grothendieck group $K_0(\BGGf')$:
$$ [S_{\mu,a}] = \sum_{\nu \in \SW_{\mu}} [S_{\nu,c_{\nu},a}']. $$
\end{lem}
\begin{proof}
Let $\SE'$ be the subalgebra of $\SE$ generated by $\BM_2$ and the $t_{ij}$ for $2\leq i,j \leq N$. One can define similar abelian category $\CF'$ of $\SE'$-modules (which are $\Hlie$-graded $\BM$-vector spaces) equipped with:
\begin{itemize}
    \item[(a)] the evaluation functor $\ev_a': \CF' \longrightarrow \BGGf'$ from $\SE'$-modules to $\CE_{N-1}^{\Hlie}$-modules;
    \item[(b)] the injective character map $\chi: K_0(\CF') \longrightarrow \CMt$ from the $\Hlie$-grading.
\end{itemize}
Theorem \ref{thm: small} applied to the $\Hlie'$-algebra $\SE_{\tau,\hbar}(\mathfrak{sl}_{N-1})$, from the scalar extension in the proof of Lemma \ref{lem: (k) vs k}, one obtains an irreducible object $V_{\nu,c}'$ in category $\CF'$ for $\nu = (\nu_2 \geq \nu_3 \geq \cdots \geq \nu_N) \in \SP'$ and $c \in \BC$ with the following properties: 
\begin{itemize}
    \item[(c)] $V_{\nu,c}'$ admits a non-zero vector $v$ of weight $c \epsilon_1 + \nu_2 \epsilon_2 + \nu_3 \epsilon_3 + \cdots + \nu_N \epsilon_N$ and $\hat{t}_{k} v = v,\ t_{ij} v = 0$ for $2\leq i,j,k \leq N$ and $i < j$;
    \item[(d)] $\chi(V_{\nu,c}')$ is equal to the character of the irreducible $\mathfrak{sl}_{N-1}'$-module of highest weight $c \epsilon_1 + \nu_2 \epsilon_2 + \nu_3 \epsilon_3 + \cdots + \nu_N \epsilon_N$; here $\mathfrak{sl}_{N-1}'$ is the parabolic Lie subalgebra of $\mathfrak{sl}_N$ (with the same Cartan algebra $\Hlie$) associated to the simple roots $\alpha_2, \alpha_3, \cdots, \alpha_{N-1}$.
\end{itemize}
By comparing highest weight we observe that 
$ \ev_a'(V_{\nu,c}')  \cong S_{\nu,c,a}'$ in category $\BGGf'$. 

Let $\mu = (\mu_1 \geq \mu_2 \geq \cdots \geq \mu_N) \in \SP$. Set $\overline{\mu} := \mu_1 \epsilon_1 + \mu_2\epsilon_2 + \cdots + \mu_N \epsilon_N$. Then $S_{\mu,a} \cong \ev_a(V_{\overline{\mu}})$ in category $\BGGf$.
By diagram chasing 
\[ \xymatrixcolsep{5pc}
\xymatrix{
K_0(\BGGf) \ar[d]_{\mathrm{Res}} &   K_0(\CF) \ar[l]_{\ev_a} \ar[d]_{\mathrm{Res}}  \ar[r]^{\chi} & \CMt  \\
K_0(\BGGf') &  K_0(\CF') \ar[l]_{\ev_a'} \ar[ur]_{\chi} &
}
\]
Lemma \ref{lem: branching} is equivalent to the character identity $\chi(V_{\overline{\mu}}) = \sum_{\nu \in \CW_{\mu}} \chi(V_{\nu,c_{\nu}}')$. Since the left-hand side (resp. the right-hand side) is the character of a representation of $\mathfrak{sl}_N$ by Theorem \ref{thm: small} (resp. of $\mathfrak{sl}_{N-1}'$ by (d)), this identity is a consequence of the branching rule for representations of the reductive Lie algebras $\mathfrak{sl}_N \supset \mathfrak{sl}_{N-1}'$. 
\end{proof}

\subsection{Proof of Theorem \ref{thm: q-char evaluation}.} \label{ss: proof q-char}
We proceed by induction on $N$. For $N = 1$ and $\mu = (n)$, since $S_{\mu,a}$ is one-dimensional, its q-character is equal to its highest weight 
$$  (\frac{\theta(z+(a+n)\hbar)}{\theta(z+a\hbar)};n \epsilon_1) = \prod_{j=1}^n \boxed{1}_{a+j-1}. $$
Suppose $N > 1$. By Lemma \ref{lem: (k) vs k}, the induction hypothesis in the case of $N-1$ gives the q-character formula for all the $\CE^{\Hlie}_{N-1}$-modules $S_{\nu,c,a}'$ where $\nu \in \SP'$ and $c \in \BC$. So the q-character $\qc'(S_{\mu,a})$ of the $\CE^{\Hlie}_{N-1}$-module $S_{\mu,a}$ is known by Lemma \ref{lem: branching}.

Since $S_{\mu,a}$ is an irreducible $\CE$-module in category $\BGGf$, by Remark \ref{rem: from N to N-1}, $\qc(S_{\mu,a})$ can be recovered from $\qc'(S_{\mu,a})$. Since $\SB_{\mu}$ is the disjoint union of the the $\SB_{\nu}'$ for $\nu \in \CW_{\mu}$, it suffices to check that for each e-weight $ (m_1^T(z),m_2^T(z),\cdots,m_N^T(z);\alpha)$ at the right-hand side of Eq.\eqref{equ: q-char evaluation}, where $T \in \SB_{\mu}$, the following product 
$$ m^T(z) := m_1^T(z+(N-1)\hbar)m_2^T(z+(N-2)\hbar)\cdots m_N^T(z)$$
is the eigenvalue of scalar action of $\eD_N(z)$ on $S_{\mu,a}$. Notice first that
$$\prod_{p=1}^N \frac{\theta(z+(a+N-p+\mu_p)\hbar)}{\theta(z+(a+N-p)\hbar)}= \prod_{(i,j) \in Y_{\mu}} \frac{\theta(z+(a+j-i+N)\hbar)}{\theta(z+(a+j-i+N-1)\hbar)}.$$
By Eq.\eqref{equ: vect rep JM action}, each box $\boxed{i}_x$ contributes to $\frac{\theta(z+(x+N)\hbar)}{\theta(z+(x+N-1)\hbar)}$, so the right-hand side of the identity is exactly $m^T(z)$. By Remark \ref{rem: from N to N-1}, the left-hand side is the scalar of $\eD_N(z)$ acting on $S_{\mu,a}$.
 This completes the proof of Theorem \ref{thm: q-char evaluation}. \hfill $\Box$ 

\section{Kirillov--Reshetikhin modules} \label{sec: KR}
We study certain irreducible $\CE$-modules via q-characters.

Fix $a \in \BC$. For $k \in \BC$ and  $1\leq r \leq N$, define the {\it asymptotic e-weight}
\begin{equation}  \label{equ: asymp e-weight}
\Bw_{k,a}^{(r)} := \Psi_{r,a+k}\Psi_{r,a}^{-1} \in \CMw.
\end{equation}
Assume $k \in \BZ_{\geq 0}$. We identify $k\varpi_r$ with the partition $(k\geq k \geq \cdots \geq k)$ where $k$ appears $r$ times. Then $\Bw_{k,a}^{(r)} = Y_{r,a+\frac{1}{2}} Y_{r,a+\frac{3}{2}} \cdots Y_{r,a+k-\frac{1}{2}} = \theta_{k\varpi_r,a}$  by Eqs.\eqref{equ: A Psi}--\eqref{equ: tableau A Y}, and the finite-dimensional irreducible $\CE$-module $S(\Bw_{k,a}^{(r)})$ in category $\BGGf$ is denoted by $W_{k,a}^{(r)}$ and called {\it Kirillov--Reshetikhin module} (KR module). 

The $Y_{i,a+m}$ (resp. the $A_{i,a+m}$) for $1 \leq i \leq N$ (resp. $1\leq i < N$) and $m \in \frac{1}{2}\BZ$ are linearly independent  in the abelian group $\CMw$, and generate the subgroup $\CP_a$ (resp. $\CQ_a$) and the submonoid $\CP_a^+$ (resp. $\CQ_a^+$). The inverses of these submonoids are denoted by $\CP_a^-$ and $\CQ_a^-$ respectively. By Eq.\eqref{equ: q-char evaluation} and Eq.\eqref{equ: tableau A Y},
$$\ewt(S_{\mu,a}) \subset \theta_{\mu,a} \CQ_a^- \subset \CP_a \quad \mathrm{for}\ \mu \in \SP. $$
Indeed, let $T_{\mu} \in \SB_{\mu}$ be such that the associated monomial in Eq.\eqref{equ: q-char evaluation} is $\theta_{\mu,a}$. Then for $S \in \SB_{\mu}$, we must have $S(i,j) \geq T_{\mu}(i,j)$ for all $(i,j) \in Y_{\mu}$.  

Following \cite[\S 6]{FM1}, we call $\Bf \in \CP_a$ {\it right negative} if the factors $Y_{i,a+m}$ with $1\leq i < N$ appearing in $\Bf$, for which $m \in \frac{1}{2}\BZ$ is minimal, have negative powers. 
\begin{lem} \cite{FM1} \label{lem: right negative}
Let $\Be, \Bf \in \CP_a$. If $\Be,\Bf$ are right negative, then so is $\Be\Bf$.
\end{lem}
All elements in $\CQ_a^-$ different from 1 are right-negative by Eq.\eqref{equ: A Psi}.

\begin{lem}  \label{lem: KR q-char}
Let $k \in \BZ_{>0}$ and $1\leq r < N$. 
\begin{itemize}
\item[(1)] For $1\leq l \leq k$, $\Bw_{k,a}^{(r)}A_{r,a}^{-1}A_{r,a+1}^{-1} \cdots A_{r,a+l-1}^{-1}$ is an e-weight of $W_{k,a}^{(r)}$ of multiplicity one in $\qc(W_{k,a}^{(r)})$.
\item[(2)] An e-weight of $W_{k,a}^{(r)}$ different from those in (1) and from $\Bw_{k,a}^{(r)}$ must belong to $\Bw_{k,a}^{(r)} A_{r,a}^{-1} A_{s,a-\frac{1}{2}}^{-1}\CQ_a^-$ for certain $1\leq s < N$ with $s = r \pm 1$.
\item[(3)] Any e-weight of $W_{k,a}^{(r)}$ is either $\Bw_{k,a}^{(r)}$ or right negative.
\end{itemize}
\end{lem}
\begin{proof}
The Young diagram $Y_{k\varpi_r}$ is a rectangle of $r$ rows and $k$ columns. For (1)--(2) the proof of \cite[Lemma 3.4]{Z2} works by applying Theorem \ref{thm: q-char evaluation} to $W_{k,a}^{(r)} \cong S_{k\varpi_r,a-\ell_r}$. For (3), $\Bw_{k,a}^{(r)} A_{r,a}^{-1}$ is right negative, and so is any element of $\Bw_{k,a}^{(r)} A_{r,a}^{-1} \CQ_a^-$. 
\end{proof}

For $1\leq r < N$ and $k,t, a \in \BC$, define as in \cite[\S 4.3]{FH2} and \cite[Remark 3.2]{Z2}:
\begin{equation}  \label{equ: Demazure e-weight}
\Bd_{k,a}^{(r,t)} := \frac{\Psi_{r,a+t}}{\Psi_{r,a}} \prod_{s = r\pm 1} \frac{\Psi_{s,a-\frac{1}{2}}}{\Psi_{s,a-\frac{1}{2}-k}} = \Bw_{t,a}^{(r)} \prod_{s= r\pm 1} \Bw_{k,a-\frac{1}{2}-k}^{(s)} \in \CMw.
\end{equation}
If $k,t \in \BZ_{\geq 0}$, then $\Bd_{k,a}^{(r,t)}  \in \CRf$ and set $D_{k,a}^{(r,t)} := S(\Bd_{k,a}^{(r,t)})$. 

\begin{lem}   \label{lem: auxiliary KR}
Let $1\leq r < N$ and $ m,\  k \in \BZ_{> 0}$. 
\begin{itemize}
    \item[(1)] The  dominant e-weights of  $W_{k+m-1,1}^{(r)}\ \wtimes\ W_{k,0}^{(r)}$ and $W_{k-1,1}^{(r)}\ \wtimes\ W_{k+m,0}^{(r)}$ are
    \begin{gather*}
    \Bw_{k+m-1,1}^{(r)}\Bw_{k,0}^{(r)} \quad \mathrm{and} \quad \Bw_{k+m-1,1}^{(r)}\Bw_{k,0}^{(r)} A_{r,1}^{-1} A_{r,2}^{-1}\cdots A_{r,l}^{-1} \quad \mathrm{for}\ 1 \leq l \leq k, \\
    \Bw_{k-1,1}^{(r)}\Bw_{k+m,0}^{(r)} \quad \mathrm{and} \quad \Bw_{k-1,1}^{(r)}\Bw_{k+m,0}^{(r)}  A_{r,1}^{-1} A_{r,2}^{-1} \cdots A_{r,l}^{-1} \quad \mathrm{for}\ 1 \leq l < k,
    \end{gather*}
    respectively. All such e-weights are of multiplicity one.
    \item[(2)] The module $W_{k-1,1}^{(r)}\ \wtimes\ W_{k+m,0}^{(r)}$ is irreducible.
\end{itemize}
\end{lem}
\begin{proof}
For (1), one can copy the last two paragraphs of the proof of \cite[Theorem 4.1]{FoH}, since the right-negativity property of KR modules in the elliptic case (Lemma \ref{lem: KR q-char}) is the same as in the affine case.  Let $T$ be the tensor product module of (2). Suppose $T$ is not irreducible. Then there exists $1\leq l \leq k-1$ such that $T$ admits an irreducible sub-quotient $S \cong S(\Bd_l)$ where  by Eq.\eqref{equ: A Psi}:
$$ \Bd_l := \Bw_{k-1,1}^{(r)}\Bw_{k+m,0}^{(r)} \prod_{j=1}^l A_{r,j}^{-1} = \frac{\Psi_{r,k}\Psi_{r,k+m}}{\Psi_{r,l+1}\Psi_{r,l}} \prod_{s= r\pm 1} \frac{\Psi_{s,l+\frac{1}{2}}}{\Psi_{s,\frac{1}{2}}}. $$
Set $\mu := \varpi(\Bd_l)$. The weight space $S[\mu-\alpha_r]$ is non-zero since the $\Psi_r$ do not cancel in $\Bd_l$, and its possible e-weights are $\Bd_l A_{r,l}^{-1}, \Bd_l A_{r,l+1}^{-1}$ since $S$ is a sub-quotient of $W_{k-l-1,l+1}^{(r)}\ \wtimes\ W_{k-l+m,l}^{(r)}\ \wtimes (\wtimes_{s= r\pm 1} W_{l,\frac{1}{2}}^{(s)})$. 
If $\Bd_l A_{r,l}^{-1}$ is an e-weight of $S$, then
$$ \Bw_{k-1,1}^{(r)}\Bw_{k+m,0}^{(r)} A_{r,1}^{-1}A_{r,2}^{-1}\cdots A_{r,l-1}^{-1} A_{r,l}^{-2} \in \ewt(T) = \ewt(W_{k-1,1}^{(r)}) \ewt(W_{k+m,0}^{(r)})  $$
which contradicts with the q-characters of KR modules in Lemma \ref{lem: KR q-char}. So $k > l+1$ and  $S[\mu-\alpha_r] = \BM v \neq 0$. Let $\omega$ be a highest weight vector of $S$. Then 
$$p := \mu_{r,r+1} = 2k-2l+m-1,\quad L_{r,r+1}(z) \omega = A(z;\lambda) v,\quad L_{r+1,r}(z) v = B(z;\lambda) \omega $$
 for some meromorphic functions $A,B$ of $(z,\lambda) \in \BC \times \Hlie$. For $1\leq i \leq N$, let $g_i(z) \in \BM_{\BC}^{\times}$ be the $i$-th component of $\Bd_l \in \CMw$. Then $ L_{ii}(z) \omega = g_i(z) \varphi_i(\lambda) \omega$ 
 for certain $\varphi_i(\lambda) \in \BM^{\times}$ by Eq.\eqref{defi: auxiliary L}. Set $h(z) :=  \frac{g_r(z)}{g_{r+1}(z)}$. We have
 $$h(z) = \frac{\theta(z+(k-\ell_r)\hbar)\theta(z+(k+m-\ell_r)\hbar)}{\theta(z+(l+1-\ell_r)\hbar)\theta(z+(l-\ell_r)\hbar)} = \frac{\theta(z-w_1)\theta(z-w_2)}{\theta(z-w_3)\theta(z-w_4)}  $$
 where $w_1 := (\ell_r-k)\hbar,\ w_2 := (\ell_r-k-m)\hbar$ and so on.  
  Applying Eq.\eqref{rel: RLL explicit} with $(i,j) = (r+1,r) = (n,m)$ to $\omega$, as in the proof of \cite[Theorem 4.1]{FZ}, we obtain 
 \begin{gather*}
  \left(\frac{\theta(z-w+\lambda_{r,r+1}+p\hbar)\theta(\hbar)}{\theta(z-w+\hbar)\theta(\lambda_{r,r+1}+p\hbar)} g_{r+1}(z)g_r(w) - \frac{\theta(z-w+\lambda_{r,r+1})\theta(\hbar)}{\theta(z-w+\hbar)\theta(\lambda_{r,r+1})} g_{r+1}(w)g_r(z) \right)  \\
 \times \varphi_r(\lambda + \hbar \epsilon_{r+1})\varphi_{r+1}(\lambda) = \frac{\theta(z-w)\theta(\lambda_{r,r+1}+\hbar)}{\theta(z-w+\hbar)\theta(\lambda_{r,r+1})} B(w;\lambda)A(z;\lambda+\hbar \epsilon_r).
 \end{gather*}
 Multiplying both sides by $\frac{\theta(z-w+\hbar)}{g_{r+1}(z)g_{r+1}(w)}$ and noticing $g_{r+1}(z) = \frac{\theta(z-w_3)}{\theta(z-w_3-l\hbar)}$, one can evaluate $w$ at $w_1$ and $w_2$ to obtain identities of meromorphic functions of $(z,\lambda)$:
 $$ \tilde{A}(z;\lambda) x_i(\lambda)  = \frac{\theta(z-w_i+\lambda_{r,r+1})}{\theta(z-w_i)} f(\lambda) h(z) \quad \mathrm{for}\ i = 1,2, $$
Here we set $\varphi(\lambda) := \varphi_r(\lambda + \hbar \epsilon_{r+1})\varphi_{r+1}(\lambda)$ and
$$ \tilde{A}(z;\lambda) := \frac{A(z;\lambda+\hbar \epsilon_r)}{g_{r+1}(z)},\quad x_i(\lambda) := \frac{B(w_i;\lambda)}{g_{r+1}(w_i)}, \quad f(\lambda) := -\frac{\theta(\hbar)\varphi(\lambda)}{\theta(\lambda_{r,r+1}+\hbar)}. $$
Since $f(\lambda) h(z) \neq 0$, we have $x_i(\lambda) \neq 0$ and so 
$$ \frac{\theta(z-w_1+\lambda_{r,r+1})\theta(z-w_2)}{x_1(\lambda)\theta(z-w_3)\theta(z-w_4)} = \frac{\theta(z-w_2+\lambda_{r,r+1})\theta(z-w_1)}{x_2(\lambda)\theta(z-w_3)\theta(z-w_4)}$$
as non-zero meromorphic functions of $(z,\lambda)$. This forces $w_1-w_2 = m\hbar \in \BZ +\BZ\tau$, which certainly does not hold. This proves (3).
\end{proof}
\begin{theorem}  \label{thm: Demazure KR}
For $1\leq r < N,\ t \in \BZ_{\geq 0}$ and $k> 0$, we have the following identities in the Grothendieck ring of category $\BGGf$:
\begin{gather}
    [D_{k,k+1}^{(r,t)}] + [W_{k-1,1}^{(r)}] [W_{k+t+1,0}^{(r)}] = [W_{k+t,1}^{(r)}] [W_{k,0}^{(r)}], \label{equ: Demazure KR} \\
    [D_{k,k}^{(r,t+1)}][W_{k+t,0}^{(r)}] = [D_{k+t+1,k+t+1}^{(r,0)}][W_{k-1,0}^{(r)}] +  [D_{k,k}^{(r,t)}] [W_{k+t+1,0}^{(r)}]. \label{equ: Demazure T-system}
\end{gather}
\end{theorem}
\begin{proof}
 Set $T := W_{k+t,1}^{(r)}\ \wtimes\  W_{k,0}^{(r)}$ and $\Bd := \Bw_{k+t,1}^{(r)}\Bw_{k,0}^{(r)}$. Then $S := S(\Bd)$ is an irreducible sub-quotient of $T$ and by Eqs.\eqref{equ: asymp e-weight}--\eqref{equ: Demazure e-weight}:
$$ \Bd = \Bw_{k-1,1}^{(r)}\Bw_{k+t+1,0}^{(r)},\quad \Bd_{k,k+1}^{(r,t)} =  A_{r,1}^{-1}A_{r,2}^{-1}\cdots A_{r,k}^{-1}\Bd. $$
Set $m = t+1$ in Lemma \ref{lem: auxiliary KR}. Then $S  \cong W_{k-1,1}^{(r)}\ \wtimes\  W_{k+t+1,0}^{(r)}$, and there is exactly one dominant e-weight (counted with multiplicity) in $\ewt(T)\setminus \ewt(S)$, namely $\Bd_{k,k+1}^{(r,t)}$. This proves Eq.\eqref{equ: Demazure KR}, which implies after taking spectral parameter shifts
\begin{align*}
     [D_{k,k}^{(r,t+1)}] &= [W_{k+t+1,0}^{(r)}] [W_{k,-1}^{(r)}] - [W_{k-1,0}^{(r)}] [W_{k+t+2,-1}^{(r)}],  \\
     [D_{k+t+1,k+t+1}^{(r,0)}] &= [W_{k+t+1,0}^{(r)}] [W_{k+t+1,-1}^{(r)}] - [W_{k+t,0}^{(r)}] [W_{k+t+2,-1}^{(r)}], \\
    [D_{k,k}^{(r,t)}] &= [W_{k+t,0}^{(r)}] [W_{k,-1}^{(r)}] - [W_{k-1,0}^{(r)}] [W_{k+t+1,-1}^{(r)}].
\end{align*}
Eq.\eqref{equ: Demazure T-system} becomes a trivial identity involving only KR modules.
\end{proof}

$D_{k,k+1}^{(r,t)}$ is {\it special} in the sense of \cite{Nakajima} as it contains only one dominant e-weight. For $t = 0$, we have $D_{k,k+1}^{(r,0)}\cong W_{k,\frac{1}{2}}^{(r-1)}\ \wtimes\ W_{k,\frac{1}{2}}^{(r+1)}$ by showing that the tensor product is special as in \cite{Nakajima}, and  Eq.\eqref{equ: Demazure KR} is the T-system of KR modules.

\begin{cor}  \label{cor: Demaure q-char}
Let $1\leq r < N,\ a \in \BC$ and $k,t \in \BZ_{>0}$. 
\begin{itemize}
      \item[(1)] $\Bd_{k,a}^{(r,t)} A_{r,a}^{-1}A_{r,a+1}^{-1}\cdots A_{r,a+l-1}^{-1}\in \ewt(D_{k,a}^{(r,t)})$ for $1 \leq l \leq t$.
    \item[(2)] Any e-weight of $D_{k,a}^{(r,t)}$ different from those in (1) and from $\Bd_{k,a}^{(r,d)}$ belongs to
    $ \Bd_{k,a}^{(r,t)}  \{A_{r,a-k-1}^{-1},\ A_{s,a-k-\frac{1}{2}}^{-1}\}  \CQ_a^-$ for certain $1\leq s < N$ with $s= r\pm 1$.
\end{itemize}
\end{cor}
\begin{proof}
This comes from Lemma \ref{lem: KR q-char} and Theorem \ref{thm: Demazure KR}.
\end{proof}

\begin{lem}  \label{lem: fusion KR}
Let $1\leq r < N$ and $t \in \BZ_{\geq 0}$. There is a short exact sequence 
$$0 \longrightarrow  D_{1,a}^{(r,t)} \longrightarrow W_{t+1,a-1}^{(r)}\ \wtimes\ W_{1,a-2}^{(r)}   \longrightarrow W_{t+2,a-2}^{(r)} \longrightarrow 0   $$
of $\CE$-modules in category $\BGGf$.
\end{lem}
\begin{proof}
Let $T$ and $S$ be the second and third terms above (zero excluded). Let $\omega_1,\omega_2$ be highest weight vectors of $W_{t+1,a-1}^{(r)}$ and $W_{1,a-2}^{(r)}$ respectively. Then $\omega_1 \wtimes \omega_2$ is a highest weight vector of $T$ and generates a sub-module $T'$. Suppose $T' = T$. Then $T$ is a highest weight module whose highest weight is equal to that of the irreducible module $S$. There is a surjective morphism of modules $T \longrightarrow S$, the kernel of which is $D_{1,a}^{(r,t)}$ by Eq.\eqref{equ: Demazure KR} (one  applies a spectral parameter shift $\Phi_{a-2}$ to the equation with $k = 1$). This is the desired short exact sequence.

Suppose $T \neq T'$. Then $[T'] = [S]$ or $[T'] = [D_{k,a}^{(r,t)}]$. By comparing highest weights, we have $[T'] = [S]$. So the weight space $T'[(t+2)\varpi_r-\alpha_r]$ is one-dimensional. Corollary \ref{cor: evaluation module} applied to $W_{t+1,a-1}^{(r)} \cong S_{(t+1)\varpi_r, a-\ell_r-1}$, one finds $g(\lambda) \in \BM^{\times}$ such that $L_{r+1,r+1}(z) \omega_1 =\omega_1$ and (set $b := a - \ell_r - 1$)
$$ L_{r,r+1}(z) \omega_1 = \frac{\theta(z+(b+t)\hbar+\lambda_{r,r+1})}{\theta(z+b\hbar)} \omega_1',\ L_{rr}(z) \omega_1 = \frac{\theta(z+(b+t+1)\hbar)}{\theta(z+b\hbar)} g(\lambda)  \omega_1, $$
where $0 \neq \omega_1'$ is of weight $(t+1)\varpi_r-\alpha_r$. Similarly $L_{r+1,r+1}(z) \omega_2 = \omega_2$ and
$$ L_{r,r+1}(z) \omega_2 = \frac{\theta(z+(b-1)\hbar+\lambda_{r,r+1})}{\theta(z+(b-1)\hbar)} \omega_2' $$
with $\omega_2' \neq 0$ of weight $\varpi_r-\alpha_r$. 
Since $\omega_1,\omega_2$ are highest weight vectors, we have
\begin{align*}
L_{r,r+1}(z) (\omega_1\wtimes \omega_2) &= L_{r,r+1}(z) \omega_1\ \wtimes\ L_{r+1,r+1}(z) \omega_2 + L_{rr}(z) \omega_1 \ \wtimes\ L_{r,r+1}(z) \omega_2 \\
&= \left( \frac{\theta(z+(b+t)\hbar+\lambda_{r,r+1})}{\theta(z+b\hbar)} \omega_1'\right) \wtimes\ \omega_2   \\
& + \frac{\theta(z+(b+t+1)\hbar)}{\theta(z+b\hbar)} g_1(\lambda)  \omega_1 \ \wtimes \left(\frac{\theta(z+(b-1)\hbar+\lambda_{r,r+1})}{\theta(z+(b-1)\hbar)} \omega_2' \right).
\end{align*} 
Setting $z = -(b+t+1)\hbar$ we obtain $\omega_1' \wtimes \omega_2 \in T'$, and so $\omega_1 \otimes \omega_2' \in T'$. The weight space $T'[(t+2)\varpi_r-\alpha_r]$ is at least two-dimensional, a contradiction.
\end{proof}
Lemma \ref{lem: fusion KR} is inspired by \cite[\S 5.3]{Moura}: to transform identities in the Grothendieck group into exact sequences by restriction to $\mathfrak{sl}_2$ \cite{Chari}. More generally, we have the short exact sequences in category $\BGGf$ by \cite[Proposition 4.3, Corollary 4.5]{FZ}: \footnote{The elliptic quantum group of \cite{FZ} is slightly different as it is defined by another R-matrix, which is gauge equivalent to the present $\BR$ by \cite{EF}. \label{ft: different R} }
\begin{gather*}
 0 \rightarrow D_{k,k+1}^{(r,t)} \longrightarrow W_{k+t,1}^{(r)} \wtimes W_{k,0}^{(r)}  \longrightarrow W_{k-1,1}^{(r)} \wtimes W_{k+t+1,0}^{(r)} \rightarrow 0, \\
 0 \rightarrow D_{k+t+1,k+t+1}^{(r,0)}\ \wtimes\ W_{k-1,0}^{(r)} \longrightarrow  D_{k,k}^{(r,t+1)}\ \wtimes\ W_{k+t,0}^{(r)} \longrightarrow D_{k,k}^{(r,t)}\ \wtimes\ W_{k+t+1,0}^{(r)}  \rightarrow 0.
\end{gather*}
These exact sequences hold for affine quantum (super)groups \cite{FoH,Z2}. In the super case the proof is more delicate since Lemma \ref{lem: KR q-char} (3) fails.

\section{Asymptotic representations}  \label{sec: asym}

 We construct infinite-dimensional modules in category $\BGG$ as inductive limits ($k \rightarrow \infty$) of the KR modules $W_{k,a}^{(r)}$ for fixed $1\leq r < N$ and $a := \ell_r$. 

The general strategy follows that of Hernandez--Jimbo \cite{HJ}:
\begin{itemize}
    \item[(i)] produce an inductive system of vector spaces $W_{0,a}^{(r)} \subseteq W_{1,a}^{(r)} \subseteq W_{2,a}^{(r)} \subseteq \cdots$; 
    \item[(ii)] prove that the matrix entries of the $L_{ij}(z)$ are {\it good} functions of $k \in \BZ_{\geq 0}$;
    \item[(iii)] define the module structure on the inductive limit of (i).
\end{itemize}
Step (i) is done in Lemma \ref{lem: injectivity}, Step (ii) in Lemma \ref{lem: asymptotic property}, and Step (iii) in Proposition \ref{prop: asymptotic representations}. We shall see that the proofs in each step are different from \cite{HJ}.

In what follows, by $k > l$ we implicitly assume that $k, l \in \BZ_{\geq 0}$ are positive integers. 
For $k > l$, set $Z_{kl} := W_{k-l,a+l}^{(r)} \cong S_{(k-l)\varpi_r,l}$ and fix a highest weight vector $\omega_{kl} \in Z_{kl}$.  
By Eq.\eqref{defi: auxiliary L}, we have for $1\leq i \leq r < j \leq N$:
$$  L_{ii}(z) \omega_{kl} = \omega_{kl}\frac{\theta(z+k \hbar)}{\theta(z+l\hbar)} \prod_{q=r+1}^N \frac{\theta(\lambda_{iq}+(k-l+1)\hbar)}{\theta(\lambda_{iq}+\hbar)} ,\quad L_{jj}(z) \omega_{kl}  = \omega_{kl}. $$
Note that $Z_{k0} = W_{k,a}^{(r)}$, and we simply write $\omega_{k0} =: \omega_k$.

\begin{lem} \label{lem: fusion k l vs k+l}
Let $t > k > l>m$. There exists a unique morphism of $\CE$-modules 
$$ \SG_{k,m}^l: Z_{kl}\  \wtimes\ Z_{lm} \longrightarrow Z_{km}$$
such that $\SG_{k,m}^l(\omega_{kl} \wtimes \omega_{lm}) = \omega_{km}$. Moreover the following diagram commutes:
\begin{equation}  \label{diagram}
\xymatrixcolsep{5pc}
\xymatrix{
\makebox[2em][r]{ $Z_{tk}\ \wtimes\ Z_{kl}\ \wtimes\ Z_{lm}$} \ar[r]^{\SG_{t,l}^k \wtimes \mathrm{Id}} \ar[d]_{\mathrm{Id}\ \wtimes\ \SG_{k,m}^l}  & \makebox[2em][l]{$ Z_{tl}\ \wtimes\ Z_{lm} \ar[d]^{\SG_{t,m}^l}$}  \\
\makebox[2em][r]{$Z_{tk}\ \wtimes\ Z_{km}$} \ar[r]^{\SG_{t,m}^k}  & Z_{tm}.
 }
\end{equation}
\end{lem}
\begin{proof}
(Uniqueness) Let $F, G$ be two such morphisms and let $X$ be the image of $F-G$. Then $\omega_{km} \notin X$. If $X \neq 0$, then $X$ has a highest weight vector $v \neq 0$, which is proportional to $\omega_{km}$ by the irreducibility of $Z_{km}$, a contradiction. So $X = 0$ and $F = G$. The  commutativity of \eqref{diagram} is proved in the same way.

(Existence) Let $b \in \BC$ and $n \in \BZ_{>0}$. By Lemma \ref{lem: fusion KR}, there exists a surjective $\CE$-linear map $W_{n-1,b+1}^{(r)}\ \wtimes\ W_{1,b}^{(r)} \longrightarrow W_{n,b}^{(r)}$. An induction on $n$ shows that the $\CE$-module $W_{1,b+n-1}^{(r)}\ \wtimes\ W_{1,b+n-2}^{(r)}\ \wtimes \cdots \wtimes\ W_{1,b+1}^{(r)}\ \wtimes\ W_{1,b}^{(r)}$ can be projected onto $W_{n,b}^{(r)}$. Setting $(n,b)= (k-m, a+m)$  we obtain a surjective $\CE$-linear map  
$$ g:  Z_{k,k-1}\ \wtimes\ Z_{k-1,k-2}\ \wtimes \cdots \wtimes\ Z_{m+2,m+1}\ \wtimes\ Z_{m+1,m} =: T \twoheadrightarrow Z_{km}. $$
Taking $(n,b)$ to be $(k-l,a+l)$ and $(l-m,a+m)$, we project the first $k-l$ and the last $l-m$ tensor factors of $T$ onto $Z_{kl}$ and $Z_{lm} $ respectively. The tensor product of these projections gives $f: T \twoheadrightarrow Z_{kl}\ \wtimes\ Z_{lm}$. Since $\omega_{kl} \wtimes \omega_{lm}, \omega_{km}$ and $\omega := \omega_{k,k-1} \wtimes \omega_{k-1,k-2} \wtimes \cdots \wtimes \omega_{m+2,m+1} \wtimes \omega_{m+1,m} \in T $ are highest weight vectors of the same e-weight, by surjectivity one can assume $f(\omega) = \omega_{kl} \wtimes \omega_{lm}$ and $g(\omega) = \omega_{km}$.
It suffices to prove that $g$ factorizes through $f$, and so $g = \SG_{k,m}^l f$. Set $Y := \ker(f)$ and $Z := \ker(g)$. The image of $g$ being irreducible, $Z$ is a maximal submodule of $T$. Since $\omega \notin Y+Z$, we have $Y + Z = Z$ and $Y \subseteq Z$. 
\end{proof}

We need two special cases of the $\SG$: for $k > l$ and $t - 1 > l$, 
$$ \SF_{k,l} = \SG_{k,0}^l: Z_{kl}\ \wtimes\ W_{l,a}^{(r)} \longrightarrow W_{k,a}^{(r)},\quad \SG_{t,l} = \SG_{t,l}^{l+1}: Z_{t,l+1}\ \wtimes\ Z_{l+1,l} \longrightarrow Z_{tl}.  $$ 
As in \cite[\S 4.2]{HJ}, for $k > l$ define the restriction map 
$$F_{k,l}: W_{l,a}^{(r)} \longrightarrow W_{k,a}^{(r)}, \quad v \mapsto \SF_{k,l}(\omega_{kl}\wtimes v). $$
It is a difference map of bi-degree $((l-k)\varpi_r,0)$.

Applying \eqref{diagram} with $t>k>l>0$ to $\omega_{tk} \wtimes \omega_{kl} \wtimes  W_{l,a}^{(r)}$ gives $F_{t,k}F_{k,l} = F_{t,l}$. So $(W_{l,a}^{(r)}, F_{k,l})$ is an inductive system of vector spaces. \footnote{In the affine case \cite[Eq.(4.26)]{HJ} the structure map comes from the stronger fact that $Z_{kl} \wtimes Z_{lm}$ is of highest weight with $Z_{km}$ being the irreducible quotient. }

Applying \eqref{diagram} with $k>l+1>l>0$ to $\omega_{k,l+1} \wtimes Z_{l+1,l}\ \wtimes\ W_{l,a}^{(r)}$, we obtain 
\begin{equation}  \label{equ: F G}
\SF_{k,l}(\SG_{k,l}(\omega_{k,l+1} \wtimes v) \wtimes w) = F_{k,l}\SF_{l+1,l}(v \wtimes w) \quad \mathrm{for}\ v \wtimes w \in Z_{l+1,l}\ \wtimes\ W_{l,a}^{(r)}.
\end{equation}

\begin{lem}  \label{lem: injectivity}
The linear maps $F_{k,l}$ are injective. 
\end{lem}
\begin{proof}
Assume $K := \ker(F_{k,l}) \neq 0$; it is a graded subspace of $W_{l,a}^{(r)}$. Choose $\mu \in \wt(K)$ such that $\mu + \alpha_i \notin \wt(K)$ for all $1\leq i < N$ and fix $0 \neq w \in K[\mu]$. We show that $w$ is a singular vector, so $w \in \BM \omega_{l}$ and $\omega_{l} \in K$, a contradiction. It suffices to prove that $L_{ji}(z) w \in K$ for all $1\leq i < j \leq N$; this implies $L_{ji}(z) w = 0$ because by assumption on $\mu$ the weight space $K[\mu +\epsilon_i-\epsilon_j]$ vanishes.

Suppose $j > r$. If $1\leq p \leq N$ and $p \neq j$, then $(k-l)\varpi_r + \epsilon_p - \epsilon_j \notin \wt(Z_{kl})$ by Theorem \ref{thm: q-char evaluation}. It follows that for $v \in W_{l,a}^{(r)}$  we have in $Z_{kl}\ \wtimes\ W_{l,a}^{(r)}$, 
$$ L_{ji}(z)(\omega_{kl} \wtimes v) = L_{jj}(z) \omega_{kl}\ \wtimes\ L_{ji}(z) v = \omega_{kl}\ \wtimes\ L_{ji}(z) v. $$
It follows that $L_{ji}(z) K \subseteq K$ because of the commutativity:
\begin{equation}  \label{inductive: L > r}
L_{ji}(z) F_{k,l} =  F_{k,l} L_{ji}(z) \quad \mathrm{for}\ 1\leq i,  j \leq N \ \mathrm{with}\ j > r. 
\end{equation}

Suppose $j \leq r$. For $p > r$ since $r \geq j > r$ we have $L_{pi}(z)w \in K$ and so $L_{pi}(z) w = 0$. For $p \leq r$, by Theorem \ref{thm: q-char evaluation}, $L_{jp}(z) \omega_{kl} = 0$ if $p \neq j$. This implies
$$ L_{ji}(z)(\omega_{kl} \wtimes w) = L_{jj}(z)\omega_{kl}\ \wtimes\ L_{ji}(z) w =  \frac{\theta(z+k \hbar)}{\theta(z+l\hbar)} g(\lambda)( \omega_{kl} \wtimes L_{ji}(z)w) $$
for certain $g(\lambda) \in \BM^{\times}$. Applying $\SF_{k,l}$ we obtain $F_{k,l}L_{ji}(z)w = 0$, as desired.
\end{proof}
In what follows $k,l$ denote positive integers, while $i,j,m,n,p,q,s,t,u,v$ the integers between $1$ and $N$ related to the Lie algebra $\mathfrak{sl}_N$.

\begin{lem} \label{lem: inductive vs K}
For $k > l$ and $1\leq i \leq N$ we have
\begin{equation} \label{inductive: K}
K_i(z) F_{k,l} = \left( \frac{\theta(z+k \hbar)}{\theta(z+l\hbar)} \right)^{\delta_{i\leq r}} F_{k,l} K_i(z).
\end{equation}
\end{lem}
\begin{proof}
We compute $\eD_i(z)(\omega_{kl} \wtimes v)$ for $v \in W_{l,a}^{(r)}$ based on the coproduct of Corollary \ref{cor: Jucys-Murphy}.
 If $ -\varpi_{N-k} \prec \alpha$ then $\alpha + \varpi_{N-k} \notin \BQ_-$ and $(k-l)\omega_{kl} + \alpha + \varpi_{N-k} \notin \wt(Z_{kl})$. The extra terms $x_{\alpha} \dt y_{\alpha}$ in the coproduct do not contribute, and so $\eD_i(z) (\omega_{kl} \wtimes v) = \eD_i(z) \omega_{kl}\  \wtimes\  \eD_i(z) v$. By Eq.\eqref{def: elliptic diagonal} similar identity holds when $\eD_i(z)$ is replaced by $K_i(z)$, because $K_i(z) \omega_{kl} = \left( \frac{\theta(z+k \hbar)}{\theta(z+l\hbar)} \right)^{\delta_{i\leq r}} \omega_{kl}$ is independent of $\lambda$. Applying $\SF_{k,l}$ to the new identity involving $K_i(z)$, we obtain Eq.\eqref{inductive: K}.
\end{proof}

{\it From now on up to Corollary \ref{cor: inductive: L < r}, we shall always fix integers $j,p$ with condition $1 \leq j \leq r < p \leq N$.} For $k > l$, introduce $\omega_{kl}^{jp} \in Z_{kl}$ by Corollary \ref{cor: evaluation module}:
$$ L_{jp}(z) \omega_{kl} = \frac{\theta(z+(k-1)\hbar+\lambda_{jp})}{\theta(z+l\hbar)} \omega_{kl}^{jp}. $$
Indeed $\omega_{kl}^{jp} = t_{pj} \omega_{kl}$ in the evaluation module  $Z_{kl} \cong V_{(k-l)\varpi_r}(l)$. Since $Y_{(k-l)\varpi_r}$ is a rectangle,   $\BM \omega_{kl}^{jp}$ is the weight space of weight $(k-l)\varpi_r + \epsilon_p - \epsilon_j$.
\begin{lem}  \label{comp: sl2}
 In the $\CE$-module $Z_{kl}$ we have $\omega_{kl}^{jp} \neq 0$ and 
\begin{equation*}  
L_{pj}(z) \omega_{kl}^{jp} = - \omega_{kl} \frac{\theta(z+l\hbar-\lambda_{jp})\theta((k-l)\hbar)\theta(\hbar)}{\theta(z+l\hbar)\theta(\lambda_{jp})\theta(\lambda_{jp}+\hbar)}\prod_{r<q\neq p} \frac{\theta(\lambda_{jq}+(k-l+1)\hbar)}{\theta(\lambda_{jq}+\hbar)}.
\end{equation*}
The product is taken over integers $q$ such that $r+1\leq q \leq N$ and $q \neq p$.
\end{lem}
\begin{proof}
The weight grading on $Z_{kl} = S_{(k-l)\varpi_r,l}$ indicates $t_{jp} \omega_{kl}^{jp} = g(\lambda) \omega_{kl}$ for certain $g(\lambda) \in \BM$. The last relation of Definition \ref{defi: small elliptic} with $a = d = j$ and $c = b = p$ applied to the highest weight vector $\omega_{kl}$, the second term vanishes and
$$\frac{\theta(\lambda_{jp}+(k-l+1)\hbar)}{\theta(\lambda_{jp}+(k-l)\hbar)} g(\lambda) = \frac{\theta((k-l)\hbar)\theta(-\hbar)}{\theta(\lambda_{jp}) \theta(\lambda_{jp}+(k-l)\hbar)} \prod_{q>r} \frac{\theta(\lambda_{jq}+(k-l+1)\hbar)}{\theta(\lambda_{jq}+\hbar)}.  $$
This implies $\omega_{kl}^{jp} \neq 0$. Conclude from $L_{pj}(z) \omega_{kl}^{jp} = \frac{\theta(z+l\hbar-\lambda_{jp})}{\theta(z+l\hbar)} g(\lambda) \omega_{kl}$.
\end{proof}

\begin{lem} \label{comp: sl2 tensor product}
Let $k - 1 > l$. In the $\CE$-module $Z_{k,l+1} \wtimes Z_{l+1,l}$ we have 
\begin{gather*}
L_{pj}(z) \left(\Ba_{jp}^{(l)}(k;\lambda)(\omega_{k,l+1} \wtimes \omega_{l+1,l}^{jp}) - \omega_{k,l+1}^{jp} \wtimes \omega_{l+1,l}  \right) = 0 \quad \mathrm{where} \\
\Ba_{jp}^{(l)}(k;\lambda) :=\frac{\theta((k-l-1)\hbar)\theta(\lambda_{jp}-\hbar)}{\theta(\hbar)\theta(\lambda_{jp})} \prod_{r<q\neq p}\frac{\theta(\lambda_{jq}+(k-l)\hbar)}{\theta(\lambda_{jq}+\hbar)}. 
\end{gather*}
Furthermore $\SG_{k,l}(\Ba_{jp}^{(l)}(k;\lambda)(\omega_{k,l+1} \wtimes \omega_{l+1,l}^{jp}) - \omega_{k,l+1}^{jp} \wtimes \omega_{l+1,l}) = 0$.
\end{lem}
\begin{proof}
We compute $L_{pj}(z) ( \omega_{k,l+1}^{jp} \wtimes \omega_{l+1,l}) = \sum_{q=1}^N L_{pq}(z)  \omega_{k,l+1}^{jp}\ \wtimes\ L_{qj}(z) \omega_{l+1,l}$. Since $\omega_{l+1,l}$ is a highest weight vector, the terms with $q > j$ vanish. The weight of $L_{qj}(z)  \omega_{k,l+1}^{jp} $ is  $(k-l-1)\varpi_r + \epsilon_q-\epsilon_j$, which does not belong to $\wt(Z_{k,l+1})$ for $q < j$. So only the term $q = j$ survives. By Lemma \ref{comp: sl2},
\begin{align*}
& L_{pj}(z) ( \omega_{k,l+1}^{jp} \wtimes \omega_{l+1,l})  = L_{pj}(z) \omega_{k,l+1}^{jp}\ \wtimes\ L_{jj}(z) \omega_{l+1,l} \\
= &\ -\  \frac{\theta(z+(l+1)\hbar-\lambda_{jp})\theta((k-l-1)\hbar)\theta(\hbar)}{\theta(z+(l+1)\hbar)\theta(\lambda_{jp})\theta(\lambda_{jp}+\hbar)} \prod_{r<q\neq p} \frac{\theta(\lambda_{jq}+(k-l)\hbar)}{\theta(\lambda_{jq}+\hbar)}  \omega_{k,l+1} \\ 
&\quad\quad  \wtimes\ \frac{\theta(z+(l+1)\hbar)}{\theta(z+l\hbar)} \prod_{q>r} \frac{\theta(\lambda_{jq}+2\hbar)}{\theta(\lambda_{jq}+\hbar)}   \omega_{l+1,l} \\
=& -\frac{\theta(z+l\hbar-\lambda_{jp})\theta((k-l-1)\hbar)\theta(\hbar)}{\theta(z+l\hbar)\theta(\lambda_{jp}+\hbar)^2} \prod_{r<q\neq p}\frac{\theta(\lambda_{jq}+(k-l+1)\hbar)}{\theta(\lambda_{jq}+\hbar)}  (\omega_{k,l+1} \wtimes \omega_{l+1,l}).
\end{align*}
Similar arguments lead to:
\begin{align*}
& L_{pj}(z) ( \omega_{k,l+1} \wtimes \omega_{l+1,l}^{jp})  = L_{pp}(z) \omega_{k,l+1}\ \wtimes\ L_{pj}(z) \omega_{l+1,l}^{jp} \\
= &\ -\frac{\theta(z+l\hbar-\lambda_{jp})\theta(\hbar)^2}{\theta(z+l\hbar)\theta(\lambda_{jp})\theta(\lambda_{jp}+\hbar)} \prod_{r<q\neq p} \frac{\theta(\lambda_{jq}+2\hbar)}{\theta(\lambda_{jq}+\hbar)}(\omega_{k,l+1} \wtimes \omega_{l+1,l}).
\end{align*}
$\Ba_{jp}^{(l)}(k;\lambda+\hbar\epsilon_j)$ is the ratio of the two coefficients of $\omega_{kl} \wtimes \omega_{l+1,l}$ above, which is easily seen to be independent of $z$. For the last identity, let $x$ be the vector in the argument of $\SG_{k,l}$. Then both $\SG_{k,l}(x)$ and $\omega_{kl}^{jp}$ belong to the one-dimensional weight space of weight $(k-l)\varpi_{r} + \epsilon_j -\epsilon_p$. These two vectors are proportional, the first is annihilated by $L_{pj}(z)$, while the second is not. So $\SG_{k,l}(x) = 0$.
\end{proof}

\begin{cor} \label{cor: sl2 tensor product}
Let $k - 1 > l$. In the $\CE$-module $Z_{kl}$ we have  
\begin{gather*}
L_{jp}(z) \omega_{kl} = \SG_{k,l}(\omega_{k,l+1} \wtimes \omega_{l+1,l}^{jp}) \times \Bb_{jp}^{(l)}(k,z;\lambda) \quad \mathrm{where}\\  
\Bb_{jp}^{(l)}(k,z;\lambda) :=\frac{\theta(z+(k-1) \hbar+\lambda_{jp})\theta((k-l)\hbar) }{\theta(z+l\hbar)\theta(\hbar)}  \prod_{r<q\neq p} \frac{\theta(\lambda_{jq}+(k-l)\hbar)}{\theta(\lambda_{jq}+\hbar)} . 
\end{gather*}
\end{cor}
\begin{proof}
The idea is similar to \cite[Lemma 7.6]{Z2}. We compute $L_{jp}(z) (\omega_{k,l+1} \wtimes \omega_{l+1,l})$.
 As in the proof of Lemma \ref{comp: sl2 tensor product}, only two terms survive:
\begin{align*}
& L_{jp}(z)(\omega_{k,l+1} \wtimes \omega_{l+1,l}) = L_{jj}(z) \omega_{k,l+1} \wtimes L_{jp}(z) \omega_{l+1,l} + L_{jp}(z) \omega_{k,l+1} \wtimes L_{pp}(z) \omega_{l+1,l} \\
&=  \frac{\theta(z+k \hbar)}{\theta(z+(l+1)\hbar)} \prod_{q>r} \frac{\theta(\lambda_{jq}+(k-l)\hbar)}{\theta(\lambda_{jq}+\hbar)}  \omega_{k,l+1}\ \wtimes\ \frac{\theta(z+l\hbar+\lambda_{jp})}{\theta(z+l\hbar)} \omega_{l+1,l}^{jp} \\
& \quad\ +\  \frac{\theta(z+(k-1)\hbar+\lambda_{jp})}{\theta(z+(l+1)\hbar)} \omega_{k,l+1}^{jp}\ \wtimes\  \omega_{l+1,l} \\
&=\ \Be_{jp}^{(l)}(k,z;\lambda) (\omega_{k,l+1} \wtimes \omega_{l+1,l}^{jp} ) \ +\ \frac{\theta(z+k\hbar+\lambda_{jp})}{\theta(z+(l+1)\hbar)} ( \omega_{k,l+1}^{jp}\wtimes \omega_{l+1,l}).
\end{align*}
Here $\Be_{jp}^{(l)}(k,z;\lambda)$ is the following meromorphic function of $(k,z,\lambda) \in \BC \times \BC \times \Hlie$:
$$\frac{\theta(z+k \hbar)\theta(z+l\hbar+\lambda_{jp})}{\theta(z+(l+1)\hbar)\theta(z+l\hbar)} \frac{\theta(\lambda_{jp}+(k-l-1)\hbar)}{\theta(\lambda_{jp})}  \prod_{r<q\neq p} \frac{\theta(\lambda_{jq}+(k-l)\hbar)}{\theta(\lambda_{jq}+\hbar)}.$$
Set $x:= \Ba_{jp}^{(l)}(k;\lambda)(\omega_{k,l+1} \wtimes \omega_{l+1,l}^{jp}) - \omega_{k,l+1}^{jp} \wtimes \omega_{l+1,l} $, which is in the kernel of $\SG_{k,l}$ by Lemma \ref{comp: sl2 tensor product}. It follows that for any $g(\lambda) \in \BM$ we have
$$ L_{jp}(z)\omega_{kl} = L_{jp}(z) \SG_{k,l}(\omega_{k,l+1}\wtimes \omega_{l+1,l}) =  \SG_{k,l}(L_{jp}(z)(\omega_{k,l+1} \wtimes \omega_{l+1,l}) + g(\lambda) x).  $$
Let us fix $g(z;\lambda) := \frac{\theta(z+k\hbar+\lambda_{jp})}{\theta(z+(l+1)\hbar)}$. Then $L_{jp}(z)(\omega_{k,l+1} \wtimes \omega_{l+1,l}) + g(z;\lambda) x$ is proportional to $\omega_{k,l+1} \wtimes \omega_{l+1,l}^{jp}$ and $L_{jp}(z) \omega_{kl} = \SG_{k,l}(\omega_{k,l+1} \wtimes \omega_{l+1,l}^{jp}) \times \Bb_{jp}^{(l)}(k,z;\lambda)$ where
\begin{align*}
\Bb_{jp}^{(l)}(k,z;\lambda) &= \Be_{jp}^{(l)}(k,z;\lambda) + g(z;\lambda) \Ba_{jp}^{(l)}(k;\lambda) \\
&=\Bb(k,z;\lambda) \times   \prod_{r<q\neq p} \frac{\theta(\lambda_{jq}+(k-l)\hbar)}{\theta(\lambda_{jq}+\hbar)}, \\
\Bb(k,z;\lambda) & := \frac{\theta(z+k \hbar)\theta(z+l\hbar+\lambda_{jp})}{\theta(z+(l+1)\hbar)\theta(z+l\hbar)} \frac{\theta(\lambda_{jp}+(k-l-1)\hbar)}{\theta(\lambda_{jp})} \\
&\quad \ + \ \frac{\theta(z+k\hbar+\lambda_{jp})\theta((k-l-1)\hbar)\theta(\lambda_{jp}-\hbar)}{\theta(z+(l+1)\hbar)\theta(\lambda_{jp})\theta(\hbar)}.
\end{align*}
$\Bb(k,z;\lambda)$ viewed as an entire function of $k$, satisfies the same double periodicity as $\theta(k\hbar)\theta(k\hbar + z +\lambda_{jp}-(l+1)\hbar)$.  One checks that $\Bb(l,z;\lambda) = 0$. This implies 
$$ \Bb(k,z;\lambda) = \theta(k \hbar + z- \hbar+\lambda_{jp})\theta(k\hbar-l\hbar) f(z;\lambda) $$
where $f(z;\lambda)$ is a meromorphic function of $(z;\lambda) \in \BC \times \Hlie$ independent of $k$. Now setting $k\hbar = - z$, we obtain $f(z;\lambda) = \frac{1}{\theta(z+l\hbar)\theta(\hbar)}$. 
\end{proof}

\begin{cor}  \label{cor: inductive: L < r}
Let $1\leq i,j \leq N$ with $j \leq r$. For $k - 1 > l$ and $x \in W_{l,a}^{(r)}$:
\begin{multline}   \label{inductive: L < r}
\ L_{ji}(z) F_{k,l}(x) = F_{k,l} \frac{\theta(z+k\hbar)}{\theta(z+l\hbar)} \mol\left(\prod_{q=r+1}^N \frac{\theta(\lambda_{jq}+(k-l+1)\hbar)}{\theta(\lambda_{jq}+\hbar)}\right)    L_{ji}(z)x \\
\ +\  F_{k,l+1}\SF_{l+1,l}\left(\sum_{p=r+1}^N   \omega_{l+1,l}^{jp}\ \wtimes\ \mol\left(\Bb_{jp}^{(l)}(k,z;\lambda)\right) L_{pi}(z) x\right).
\end{multline}
\end{cor}
\begin{proof}
Consider
$ L_{ji}(z) F_{k,l}(x) = \SF_{k,l}(\sum_{p=1}^N L_{jp}(z) \omega_{kl} \wtimes L_{pi}(z)x)$. As in the proof of Lemma \ref{comp: sl2 tensor product},
$L_{jp}(z) \omega_{kl} = 0$ if $p \notin \{j,r+1,r+2,\cdots,N\}$. For $p = j$, we obtain the first row of Eq.\eqref{inductive: L < r}, while for $r < p \leq N$, Corollary \ref{cor: sl2 tensor product} and Eq.\eqref{equ: F G} with $v = \omega_{l+1,l}^{jp}$ give the second row.
\end{proof}

Fix weight bases $\mathcal{B}_l$ of $W_{l,a}^{(r)}$ for $l > 0$ uniformly so that $F_{k,l}(\mathcal{B}_l) \subseteq \mathcal{B}_k$.

We view $\Bb_{jp}^{(l)}(c,z;\lambda)$  in Corollary \ref{cor: sl2 tensor product} as  a meromorphic function of $(c,z,\lambda) \in \BC^2 \times \Hlie$. For $1\leq i,j \leq N,\ l > 0$ and $c, z \in \BC$, define $\SL_{ji}^{(l)}(c,z): W_{l,a}^{(r)} \longrightarrow W_{l+1,a}^{(r)}$:
\begin{align*}
&\SL_{ji}^{(l)}(c,z)x = F_{l+1,l} L_{ji}(z)x = \frac{\theta(z+(\gamma_j+\delta_{ij}-1)\hbar + \lambda_{ji})}{\theta(z)} F_{l+1,l} t_{ij} x \quad \mathrm{for} \ j > r, \\
&\SL_{ji}^{(l)}(c,z)x =  \frac{\theta(z+c\hbar)}{\theta(z+l\hbar)} \prod_{q=r+1}^N \frac{\theta(\lambda_{jq}+(c+\gamma_{jq}+\delta_{ij}-\delta_{iq})\hbar)}{\theta(\lambda_{jq}+(l+\gamma_{jq}+\delta_{ij}-\delta_{iq})\hbar)}  F_{l+1,l}  L_{ji}(z)x   \\
&\quad \quad + \sum_{p=r+1}^N \Bb_{jp}^{(l)}(c,z;\lambda+(\gamma+l\varpi_r+\epsilon_i-\epsilon_p)\hbar) \SF_{l+1,l}(\omega_{l+1,l}^{jp} \wtimes\  L_{pi}(z)x) \quad \mathrm{for}\ j \leq r.
\end{align*}
Here $x \in W_{l,a}^{(r)}[\gamma+l\varpi_r]$ and $\delta_{ij}$ is the usual Kronecker symbol.   Corollary \ref{cor: evaluation module} applied to the evaluation module $W_{l,a}^{(r)} \cong V_{l\varpi_r}(0)$ indicates that for $b' \in \mathcal{B}_{l+1}$ and $b \in \mathcal{B}_l$:
\begin{itemize}
\item[] $\SL_{ji}^{(l)}(c,z)$ is a difference map of bi-degree $( \epsilon_j-\varpi_r,\epsilon_i)$. Its matrix entry $[\SL_{ji}^{(l)}]_{b'b}(c,z;\lambda)$ is a meromorphic function of $(c,z,\lambda) \in \BC^2 \times \Hlie$. Moreover,  $\theta(z)\theta(z+l\hbar)[\SL_{ji}^{(l)}]_{b'b}(c,z;\lambda)$ is entire on $(c,z)$ for generic $\lambda$.
\end{itemize}
As a unification of Eqs.\eqref{inductive: L > r} and \eqref{inductive: L < r}, we have 
\begin{equation}  \label{inductive}
L_{ji}(z) F_{k,l} = F_{k,l+1} \SL_{ji}^{(l)}(k,z) \quad \mathrm{for}\ k > l+1.
\end{equation}

For $k \in \BZ_{>0}$ and $z \in \BC$ let $\Xi(c;k,z)$ be the set of entire functions $F(c)$ of $c \in \BC$ 
with the following double periodicity:
$$ F(c + \hbar^{-1}) = (-1)^k F(c),\quad F(c + \tau\hbar^{-1}) = (-1)^k e^{-k\BI \pi \tau - 2k \BI \pi c\hbar - 2\BI \pi z} F(c). $$
A typical example is $\theta(c\hbar)^{k-1}\theta(c\hbar+z)$. Such a function is called {\it homogeneous}.  If $f(c),g(c) \in \Xi(c;k,z)$, then we write $f(c) \approx g(c)$.

Note that $\Xi(c;k,z) \Xi(c;k',z') \subseteq \Xi(c;k+k',z+z')$.

\begin{lem}  \label{lem: asymptotic property}
Let $b \in \mathcal{B}_l$ be of weight $\gamma+l\varpi_r$ and $b' \in \mathcal{B}_{l+1}$. For $j > r$ the matrix entry  $[\SL_{ji}^{(l)}]_{b'b}(c,z;\lambda)$ is independent of $c$. For $j \leq r$ as  entire functions of $c$
$$[\SL_{ji}^{(l)}]_{b'b}(c,z;\lambda)  \approx  \theta(z+c\hbar) \prod_{q=r+1}^N \theta(\lambda_{jq} + (c+\gamma_{jq}+\delta_{ij}-\delta_{iq})\hbar). $$
Moreover, $\theta(z)[\SL_{ji}^{(l)}]_{b'b}(c,z;\lambda)$ is an entire function of $(c,z)$ for generic $\lambda$.
\end{lem}
\begin{proof}
In the case $j > r$, Corollary \ref{cor: evaluation module} applied to $W_{l,a}^{(r)} \cong S_{l\varpi_r,0}$, the matrix entry is of the form $\frac{\theta(z+(\gamma_j+\delta_{ij}-1)\hbar + \lambda_{ji})}{\theta(z)} g_1(\lambda)$ for $g_1(\lambda) \in \BM$.
Assume $j \leq r$. By Corollary \ref{cor: sl2 tensor product} the matrix entry is of the form  $\frac{E(c,z;\lambda)g_2(\lambda)}{\theta(z)\theta(z+l\hbar)}$, where $g_2(\lambda) \in \BM$ and $E(c,z;\lambda)$ is an entire function of $(c,z,\lambda) \in \BC \times \BC \times \Hlie$. As functions of $z,c$ resp., we have
\begin{align*}
    E(c,z\hbar;\lambda) &\in \Xi(z;2,(c+l+\gamma_j+\delta_{ij}-1)\hbar + \lambda_{ji} ), \\
    E(c,z;\lambda) &\in \Xi(c;N-r+1,\sum_{q=r+1}^N(\lambda_{jq} + (\gamma_{jq}+\delta_{ij}-\delta_{iq})\hbar) + z).
\end{align*}
On the other hand, for $k > l+1$ we have by Corollary \ref{cor: evaluation module} and Eq.\eqref{inductive},
$$ F_{k,l+1} \SL_{ji}^{(l)}(k,z) b = L_{ji}(z) F_{k,l}v = \frac{\theta(z+\lambda_{ji}+(k+\gamma_j+\delta_{ij}-1)\hbar)}{\theta(z)} t_{ij} F_{k,l}b. $$
The right-hand side as a function of $z$ is regular at $z = -l\hbar$, so is any of the coefficients of the left-hand side $\frac{E(k,z;\lambda)}{\theta(z)\theta(z+l\hbar)} g_2(\lambda)$. This forces $E(k,-l\hbar;\lambda) = 0$ and 
$$ E(c,z;\lambda) = \theta(z+l\hbar)\theta(z+(c+\gamma_j+\delta_{ij}-1)\hbar+\lambda_{ji}) D(c;\lambda)g_3(\lambda)  $$
where $g_3(\lambda) \in \BM$ and $D(c;\lambda)$ is an entire function of $(c,\lambda)$. Applying the double periodicity with respect to $c$ once more, we obtain the desired result.
\end{proof}
\begin{lem}  \label{lem: elliptic homogeneous functions}
Let $f(c)$ be a homogeneous entire function. If $f(k) = 0$ for infinitely many integers $k$, then  $f(c)$ is identically zero.
\end{lem}
\begin{proof}
By definition the homogeneous entire function $f(c)$, if non-zero, can be written as a product of theta functions $\theta(c\hbar + z)$. Since $\hbar \notin \mathbb{Q} + \mathbb{Q} \tau$, each of these theta functions of $c$ can not have zeroes at infinitely many integers. 
\end{proof}

Let $W_{\infty}$ be the inductive limit of the inductive system $(W_{l,a}^{(r)}, F_{k,l})$ of vector spaces (over $\BM$), with the $F_l: W_{l,a}^{(r)} \longrightarrow W_{\infty}$ for $l > 0$ being the structural maps.

From now on fix $d \in \BC$. A vector $0 \neq w \in W_{\infty}$ is of weight $d \varpi_r + \gamma$ if there exist $l > 0$ and $w' \in W_{l,a}^{(r)}[l\varpi_r + \gamma]$ such that $w = F_l(w')$.  The weight grading is independent of the choice of $l$ because $F_{k,l}$ sends $W_{l,a}^{(r)}[l\varpi_r + \gamma]$ to $W_{k,a}^{(r)}[k\varpi_r + \gamma]$. Let $W_{\infty}^d$ denote the resulting object of $\CV$. By construction $\wt(W_{\infty}^d) \subseteq d\varpi_r + \BQ_-$, and $F_l: W_{l,a}^{(r)} \longrightarrow W_{\infty}^d$ is a difference map of bi-degree $((l-d)\varpi_r ,0)$.

Let $\gamma \in \BQ_-$. The injective maps $F_{k,l}$ together with Theorems \ref{thm: q-char evaluation} and \ref{thm: small} imply that $\dim (W_{k,a}^{(r)}[k\varpi_r+\gamma]) = d_{k\varpi_r}[k\varpi_r+\gamma]$, as $k \rightarrow \infty$, converges to an integer which is exactly $\dim (W_{\infty}^d[d\varpi_r+ \gamma])$. So $W_{\infty}^d$ is an object of $\CVf$. Our goal is to make $W_{\infty}^d$ into an $\CE$-module in category $\BGG$ with favorable q-character.  \footnote{In the affine case, the matrix entries of analogs of $\SL_{ji}^{(l)}(k,z)$ are Laurent  polynomials of $e^{k\hbar}$. Hernandez--Jimbo \cite{HJ} proved this by using elimination theorems of q-characters and then took the limit $e^{k\hbar} \rightarrow 0$ as $k \rightarrow \infty$ to obtain modules over Borel subalgebras of affine quantum groups. Later in \cite{Z2,Z3} an elementary proof of polynomiality was given based on $\mathfrak{sl}_2$-representation theory, which by taking limit $e^{k\hbar} \rightarrow e^{d \hbar}$ as $k \rightarrow \infty$ (with $d \in \BC$ a new parameter) resulted in modules over affine quantum groups. Here we adapt the second approach to the elliptic case. \label{ft: affine} }

For $1\leq i,j \leq N$ and $z \in \BC$ with $\theta(z) \neq 0$, the $\SL_{ji}^{(l)}(d,z)$ constitute a morphism of inductive system of $\BC$-vector spaces:
\[ \xymatrixcolsep{5pc}
\xymatrix{
\makebox[2em][r]{$W_{l,a}^{(r)} \ar[r]^{\SL_{ji}^{(l)}(d,z)} \ar[d]_{F_{l',l}} $} & \makebox[2em][l]{$ W_{l+1,a}^{(r)} \ar[d]^{F_{l'+1,l+1}} $} \\
\makebox[2em][r]{$ W_{l',a}^{(r)} \ar[r]^{\SL_{ji}^{(l')}(d,z)} $}  & \makebox[2em][l]{$ W_{l'+1,a}^{(r)}$}
} \quad \mathrm{for}\ l' > l.
\]
Indeed, the matrix entries of $F_{l'+1,l+1} \SL_{ji}^{(l)}(c,z)$ and $\SL_{ji}^{(l')}(c,z) F_{l',l}$, as difference maps $W_{l,a}^{(r)}\longrightarrow W_{l'+1,a}^{(r)}$, are homogeneous entire functions of $c$ with the same double periodicity by Lemma \ref{lem: asymptotic property} and are equal at all integers $c$ larger than $l' + 1$ by Eq.\eqref{inductive}.  By Lemma \ref{lem: elliptic homogeneous functions} these two maps coincide for all $c \in \BC$. Define
$$ \SL_{ji}^d(z) := \lim\limits_{\rightarrow} \SL_{ji}^{(l)}(d,z) \in \mathrm{Hom}_{\BC}(W_{\infty}^d,W_{\infty}^d). $$
For $x \in W_{\infty}^d[d\varpi_r+\gamma]$ with $x = F_l(x')$ and $x' \in W_{l,a}^{(r)}[l\varpi_r+\gamma]$, we have
\begin{equation}  \label{asymptotic def}
    \SL_{ji}^d(z)x = F_{l+1} \SL_{ji}^{(l)}(d,z)  x'.
\end{equation}
The difference maps $\SL_{ji}^{(l)}(d,z)$ and $F_{l+1}$ are of bi-degree  $(\epsilon_j-\varpi_r,\epsilon_i)$ and $((l+1-d)\varpi_r ,0)$ respectively. So $\SL_{ji}^d(z)$ is a difference operator of bi-degree $(\epsilon_j,\epsilon_i)$. 
\begin{prop}   \label{prop: asymptotic representations}
$(W_{\infty}^d, \SL_{ji}^{d}(z))$ is an $\CE$-module in category $\BGG$. Moreover, 
\begin{equation}  \label{equ: q-char asym}
\qc(W_{\infty}^d, \SL_{ji}^{d}(z)) = \Bw_{d,a}^{(r)} \times \lim_{k \rightarrow \infty} (\Bw_{k,a}^{(r)})^{-1} \qc(W_{k,a}^{(r)}).
\end{equation}
\end{prop}
\begin{proof}
We need to prove Conditions (M1)--(M3) of Section \ref{ss-repr}. First (M1) follows from Eq.\eqref{asymptotic def} and from the comments before Eq.\eqref{inductive}. To prove (M2),
let $x \in W_{\infty}^d[d\varpi_r+\gamma]$ and $x' \in W_{l,a}^{(r)}[l\varpi_r+\gamma]$ such that $x = F_l(x')$. We assume $l$ so large that $W_{\infty}^d[d\varpi_r+\gamma]$ and $W_{l,a}^{(r)}[l\varpi_r+\gamma]$ have the same dimension.

\noindent {\bf Step I: Proof of (M2).}  We need to show that for $1\leq i,j,m,n \leq N$ 
\begin{multline*}   
\quad \quad \sum_{p,q} R^{pq}_{mn}(z-w;\lambda+(\epsilon_i+\epsilon_j + d\varpi_r + \gamma)\hbar) \SL_{pi}^d(z) \SL_{qj}^d(w)x \\
= \sum_{s,t} R_{st}^{ij}(z-w;\lambda) \SL_{nt}^d(w)\SL_{ms}^d(z)x  \in W_{\infty}^d.
\end{multline*}
Here at the right-hand side we have used $R_{mn}^{pq}(z;\lambda) = R_{mn}^{pq}(z;\lambda+\hbar\epsilon_p+\hbar\epsilon_q)$ to move $R$ to the left. By Eq.\eqref{asymptotic def} it is enough to prove the equation:
\begin{multline}\label{star}
 \quad \quad \sum_{p,q} R^{pq}_{mn}(z-w;\lambda+(\epsilon_i+\epsilon_j + c\varpi_r + \gamma)\hbar) \SL_{pi}^{(l+1)}(c,z) \SL_{qj}^{(l)}(c,w)x' \\
= \sum_{s,t} R_{st}^{ij}(z-w;\lambda) \SL_{nt}^{(l+1)}(c,w)\SL_{ms}^{(l)}(c,z)x' \in W_{l+2,a}^{(r)}.
\end{multline}
Let $A_1(c,z,w)$ and $A_2(c,z,w)$ denote the left-hand side and the right-hand side of this equation without $x'$.  These are difference maps $W_{l,a}^{(r)} \longrightarrow W_{l+2,a}^{(r)}$ of bi-degree $(\epsilon_m+\epsilon_n - 2\varpi_r, \epsilon_i + \epsilon_j)$, as $R_{mn}^{pq} \neq 0$ implies $\epsilon_m + \epsilon_n =  \epsilon_p + \epsilon_q$.

\medskip

\noindent {\it Claim 1.} 
For $b \in \mathcal{B}_l$ of weight $l\varpi_r+\gamma$ and $b' \in \mathcal{B}_{l+2}$, as entire functions of $c$, 
$$[A_1]_{b'b}(c,z,w;\lambda) \approx  [A_2]_{b'b}(c,z,w;\lambda).$$

This is divided into four cases. For simplicity let us drop $b',b,z,w,\lambda$ from $A_1, A_2$. 

\noindent {Case 1.1: $m,n > r$.} $A_1(c)$ and $A_2(c)$ are independent of $c$ by Lemma \ref{lem: asymptotic property}.

\noindent {Case 1.2: $m,n \leq r$.} At the left-hand side of Eq.\eqref{star} we have $\{p,q\} = \{m,n\}$ and so $R_{mn}^{pq}$ is independent of $c$. At the right-hand side $\{s,t\} = \{i,j\}$. Therefore
\begin{align*}
A_1(c) & \approx  \theta(c\hbar + z) \prod_{u>r} \theta(\lambda_{pu} + (c+\gamma_{pu}+\delta_{ip}-\delta_{iu}+\delta_{jp}-\delta_{ju}-\delta_{qp}+\delta_{qu})\hbar) \\
& \quad \times  \theta(c\hbar + w) \prod_{v>r} \theta(\lambda_{qv}+(c+\gamma_{qv}+\delta_{jq}-\delta_{jv}+\delta_{iq}-\delta_{iv})\hbar), \\
A_2(c) & \approx  \theta(c\hbar + w) \prod_{u>r} \theta(\lambda_{nu}+(c+\gamma_{nu}+\delta_{tn}-\delta_{tu}+\delta_{sn}-\delta_{su}-\delta_{mn}+\delta_{mu})\hbar) \\
& \quad \times  \theta(c\hbar + z) \prod_{v>r} \theta( \lambda_{mv}+(c+\gamma_{mv}+\delta_{sm}-\delta_{sv} + \delta_{tm}-\delta_{tv})\hbar).
\end{align*}
These formulas are deduced from Lemma \ref{lem: asymptotic property}. One needs to take into account the shifts of $\gamma, \lambda$. For example at the left-hand side of Eq.\eqref{star}, the term $\SL_{qj}$ (resp. $\SL_{pi}$) shifts $\gamma$ (resp. $\lambda$) by $\epsilon_j-\epsilon_q$ (resp. $\hbar\epsilon_i$). The right-hand sides of these two formulas lie in $\Xi(c;2+2N-2r,e)$ with $e\in \BC$ independent of the choices of $p,q,s,t$.

\noindent {Case 1.3: $m \leq r < n$.} At the right-hand side $\{s,t\} = \{i,j\}$ and
\begin{align*}
A_2(c) & \approx   \theta(c\hbar + z) \prod_{v>r} \theta( \lambda_{mv}+(c+\gamma_{mv}+\delta_{sm}-\delta_{sv} + \delta_{tm}-\delta_{tv})\hbar) \\
& \approx \theta(c\hbar + z) \prod_{v>r} \theta( \lambda_{mv}+(c+\gamma_{mv}+\delta_{im}-\delta_{iv} + \delta_{jm}-\delta_{jv})\hbar). 
\end{align*}
The last term is independent of $s,t$. On the other hand $A_1(c) = E(c) + F(c)$ where $E, F$ correspond to $(p,q) = (m,n)$ and $(p,q) = (n,m)$ respectively and so:
\begin{align*}
E(c) & \approx \frac{\theta(f-\hbar)}{\theta(f)}   \theta(c\hbar + z) \prod_{u>r} \theta( \lambda_{mu}+(c+\gamma_{mu}+\delta_{im}-\delta_{iu}+\delta_{jm}-\delta_{ju}+\delta_{nu})\hbar), \\
F(c) & \approx \frac{\theta(f+z-w)}{\theta(f)}\theta(c\hbar+w) \prod_{v>r} \theta( \lambda_{mv}+(c+\gamma_{mv}+\delta_{jm}-\delta_{jv}+\delta_{im}-\delta_{iv})\hbar). 
\end{align*}
Here $f := c\hbar + \lambda_{mn} + (\gamma_{mn} + \delta_{im}-\delta_{in}+\delta_{jm}-\delta_{jn})\hbar$. We observe easily that $A_2(c) \approx E(c) \approx F(c)$ and so $A_1(c) \approx A_2(c)$ are homogeneous. 

\noindent {Case 1.4: $n \leq r < m$.} This is parallel to the third case. 

\medskip

\noindent {\it Claim 2.} In Claim 1 equality holds for $c = k \in \BZ_{>l+2}$.

 Let us apply $F_{k,l+2}$ to Eq.\eqref{star} with $c = k$ and $x' = b$. By Eq.\eqref{inductive}:
\begin{align*}
    F_{k,l+2} \SL_{pi}^{(l+1)}(k,z) \SL_{qj}^{(l)}(k,w) = L_{pi}(z)F_{k,l+1} \SL_{qj}^{(l)}(k,w)x' = L_{pi}(z)L_{qj}(w) F_{k,l}, 
\end{align*}
and similarly $F_{k,l+2}\SL_{nt}^{(l+1)}(d,w)\SL_{ms}^{(l)}(d,z)  b = L_{nt}(w)L_{ms}(z) F_{k,l}b$. We obtain the defining relation
$RLL = LLR$  of the $\CE$-module $W_{k,a}^{(r)}$ applied to the vector $F_{k,l}(b)$. Since $F_{k,l+2}$ is injective, Eq.\eqref{star} holds for $c = k$ and $x' = b$. This proves Claim 2. 

Together with Lemma \ref{lem: elliptic homogeneous functions}, we obtain equality in Claim 1 for all $c \in \BC$. This proves Eq.\eqref{star}.

\medskip

\noindent {\bf Step II.} Let $1\leq i \leq N$. We have by Eqs.\eqref{def: determinant} and \eqref{star}:
\begin{equation}  \label{double star}
    \eD_i(z) x = \frac{\Theta_i(\lambda)}{\Theta_i(\lambda + (d\varpi_r + \gamma)\hbar)}  F_{l+i} \SD_i^{(l)}(d,z) x'.
\end{equation}
Here $\SD_i^{(l)}(c,z) = \sum_{\sigma \in \mathfrak{S}^i} T_{\sigma}(c,z)$ and  $T_{\sigma}(c,z): W_{l,a}^{(i)} \longrightarrow W_{l+i,a}^{(i)}$ for $\sigma \in \mathfrak{S}^i$ is 
$$ \mathrm{sign}(\sigma) \SL_{\sigma(N),N}^{(l+i-1)}(c,z)\SL_{\sigma(N-1),N-1}^{(l+i-2)}(c,z+\hbar) \cdots \SL_{\sigma(N-i+1),N-i+1}^{(l)}(c,z+(i-1)\hbar). $$
 Each $T_{\sigma}(c,z)$ is a difference map of bi-degree $(-\varpi_{N-i}-i\varpi_r,-\varpi_{N-i})$.
Define the meromorphic function of $(c,z) \in \BC^2$ (note that $l$ is fixed): 
$$ g(c,z) = 1\ \mathrm{if}\ i < N+1-r,\quad g(c,z) =  \prod_{p=N-i+1}^r \frac{\theta(z+(N-p+c)\hbar)}{\theta(z+(N-p+l)\hbar)} \ \mathrm{otherwise}. $$

\medskip

\noindent {\it Claim 3.} For $b \in \mathcal{B}_l$ of weight $l\varpi_r+\gamma$ and $b' \in \mathcal{B}_{l+i}$, as entire functions of $c \in \BC$,  $$ [T_{\sigma}]_{b'b}(c,z;\lambda) \approx g_i(c,z) \Theta_i(\lambda + (c\varpi_r + \gamma)\hbar).$$

The idea is the same as Claim 1, based on Lemma \ref{lem: asymptotic property}. 
If $N-i+1 > r$, then  $T_{\sigma}^b(c,z;\lambda), \Theta_i(\lambda + (c\varpi_r + \gamma)\hbar)$ are independent of $c$, and we are done. 

Assume $N-i+1 \leq r$. By Eq.\eqref{def: Vandermond} and Lemma \ref{lem: asymptotic property},
\begin{align*}
    &\Theta_i(\lambda + (c\varpi_r + \gamma)\hbar) \approx \prod_{p=N-i+1}^r \prod_{u=r+1}^N \theta(\lambda_{pu} + (c+\gamma_{pu})\hbar), \\
    & [T_{\mathrm{Id}}]_{b,b'}(c,z;\lambda) \approx \prod_{p=N-i+1}^r \theta(z+(c+N-p)\hbar) \prod_{u=r+1}^N \theta(\lambda_{pu}^{(p)} + (c+\gamma_{pu}+1)\hbar).
\end{align*}
Here $\lambda^{(p)} = \lambda + \hbar \sum_{v=p+1}^N \epsilon_v$ and so $\lambda_{pu}^{(p)} = \lambda_{pu} - \hbar$ for $p \leq r < u$. The case $\sigma = \Id$ in Claim 3 is now obvious. It remains to show $[T_{\sigma}]_{b'b}(c,z;\lambda) \approx [T_{\sigma'}]_{b'b}(c,z;\lambda)$ for all $\sigma, \sigma' \in \mathfrak{S}^i$. One can assume $\sigma' = \sigma s_j$ where $s_j = (j,j+1)$ is a simple transposition with $N-i+1 \leq j < N-1$.  Let us define
$$p := \sigma(j+1),\ q := \sigma(j),\quad l' := l+i+j-1-N,\quad w := z + (N-j)\hbar$$ 
Then we have the decomposition of difference maps
\begin{align*}
    T_{\sigma}(c,z) &= \mathrm{sign}(\sigma)A(c,z)  U_{pq}(c,w)  B(c,z),\\  
    T_{\sigma'}(c,z)&= \mathrm{sign}(\sigma)A(c,z)  U_{qp}(c,w)  B(c,z).
\end{align*}
The difference maps $A, B, U$ are defined by (descending order in the products)
\begin{align*}
    A(c,z) &=  \prod_{u=N}^{j+2} \SL_{\sigma(u),u}^{(l+i-1-N+u)}(c,z+(N-u)\hbar): W_{l'+2,a}^{(r)} \longrightarrow W_{l+i-1,a}^{(r)},\\  
    B(c,z)&= \prod_{u=j-1}^{N-i+1} \SL_{\sigma(u),u}^{(l+i-1-N+u)}(c,z+(N-u)\hbar): W_{l,a}^{(r)} \longrightarrow W_{j-1,a}^{(r)}, \\
 U_{pq}(c,w) &=  \SL_{p,j+1}^{(l'+1)}(c,w-\hbar) \SL_{qj}^{(l')}(c,w): W_{l',a}^{(r)} \longrightarrow W_{l'+2,a}^{(r)}. 
 \end{align*}
 Flipping $p,q$ one gets $U_{qp}$. Now $[T_{\sigma}]_{b'b}(c,z;\lambda) \approx [T_{\sigma'}]_{b'b}(c,z;\lambda)$ is a consequence of the following claim.

\medskip

\noindent {\it Claim 4.} For $y \in \mathcal{B}_{l'}$ of weight $l'\varpi_r + \eta$ and $y' \in \mathcal{B}_{l'+2}$, as entire functions of $c$,  $$[U_{pq}]_{y'y}(c,w;\lambda) \approx [U_{qp}]_{y'y}(c,w;\lambda).$$

If $p,q \leq r$, then by Lemma \ref{lem: asymptotic property} (setting $\eta' = \eta + \epsilon_j - \epsilon_q$ and $\lambda' = \lambda + \hbar \epsilon_{j+1}$)
\begin{align*}
    [U_{pq}]_{y'y}(c,w;\lambda) \approx &\ \theta(w+(c-1)\hbar)  \prod_{u=r+1}^N \theta(\lambda_{pu} + (c+ \eta_{pu}'+\delta_{p,j+1}-\delta_{u,j+1})\hbar)  \\
    &  \times \theta(w+c\hbar) \prod_{v=r+1}^N \theta(\lambda_{qv}' + (c+\eta_{qv}+\delta_{jq}-\delta_{jv})\hbar).
\end{align*}
We have $U_{pq}^{b'}(c,w;\lambda) \in \Xi(c;2N-2r+2,e)$ with $e = e(p,q)$ symmetric on $p,q$. So $[U_{pq}]_{y'y}(c,w;\lambda) \approx [U_{qp}]_{y'y}(c,w;\lambda)$. 

The other cases of $p,q$ are proved in the same way as in Claim 1.  
\medskip

\noindent {\bf Step III: Proof of (M3).} Let $k > l+i$. Notice that $\eD_i(z) \omega_{kl} = g_i(k,z) \omega_{kl}$. From the proof of Lemma \ref{lem: inductive vs K} and from Eqs.\eqref{inductive} and \eqref{double star} we get 
$$  g_i(k,z)  F_{k,l} \eD_i(z) x' = \eD_i(z) F_{k,l}(x') = F_{k,l+i} \frac{\Theta_i(\lambda)}{\Theta_i(\lambda + (k\varpi_r + \gamma)\hbar)} \SD_i^{(l)}(k,z) x'.  $$
Applying $F_{k}$ to this identity and multiplying $\Theta_i(\lambda + (k\varpi_r + \gamma)\hbar)$ we have 
$$ \Theta_i(\lambda + (k\varpi_r + \gamma)\hbar)g_i(k,z)  F_{l} \eD_i(z) x' = \Theta_i(\lambda) F_{l+i}\SD_i^{(l)}(k,z) x'\quad \mathrm{for}\ k > l+i.$$
Both sides after taking coefficients with respect to a basis of $W_{\infty}^d[d\varpi_r+\gamma]$ can be viewed as entire functions of $k \in \BC$, and they satisfy the same double periodicity by Claim 3. By Lemma \ref{lem: elliptic homogeneous functions}, the above identity holds for all $k \in \BC$. Taking $k = d$, by Eq.\eqref{double star}, we obtain $\eD_i(z) x =  g_i(d,z) F_l \eD_i(z)x' $.  

Let $B$ be a basis of $W_{l,a}^{(r)}[l\varpi_r+\gamma]$ satisfying the upper triangular property of (M3). Then so does the basis $F_l(B)$ of $W_{\infty}^d[d\varpi_r+\gamma]$. The $\CE$-module $W_{\infty}^d$ is in category $\BGG$.
The diagonal entry of $\eD_i(z)$ associated to $F_l(x') \in W_{\infty}^d$ for $x' \in B$ is equal to that of $\eD_i(z)$ associated to $x' \in W_{l,a}^{(r)}$ multiplied by $g_i(d,z)$. The q-character formula in Eq.\eqref{equ: q-char asym} follows from the explicit formula of $g_i(d,z)$.
\end{proof}

\begin{que} \label{question}
Let $F(c)$ be a finite sum of homogeneous entire functions. If $F(k) = 0$ for infinitely many integers $k$, then is  $F(c)$ identically zero?
\end{que}

If the answer to this question is affirmative, then the proof of Proposition \ref{prop: asymptotic representations} can be largely simplified: Claims 1, 3 and 4 are not necessary. \footnote{In the affine case, by Footnote \ref{ft: affine} the situation is much easier: a Laurent polynomial vanishing at infinitely many integers must be zero; see \cite[\S 2]{Z3}.}

\begin{rem}  \label{rem: evaluation asymptotic}
By Lemma \ref{lem: asymptotic property},  $W_{\infty}^d\cong \CW(0)$ with $\CW$ an $\SE$-module of character
$\lim\limits_{k \rightarrow \infty} e^{(d-k)\varpi_r} \chi(W_{k,a}^{(r)})$, so it is in the image of the functor \cite[Proposition 4.1]{EM}. By Lemma \ref{lem: injectivity} and its proof, $\CW$ contains a unique highest weight vector up to scalar. Let $Q$ be the quotient of standard Verma module $M_{d\varpi_r,1}$ in \cite[Proposition 4.7]{TV} by $t_{a+1,a} v_{d\varpi_r,1}$ for $a \neq r$. Then $\CW$ is the contragradient module to $Q$ in \cite[\S 6]{TV}.  It is interesting to have a direct proof of $\CW(0)$ being in category $\BGG$.
\end{rem}

For $x \in \BC$ let $\CW_{d,x}^{(r)}$ be the pullback of $W_{\infty}^d$ by $\Phi_{x-a}$ in Eq.\eqref{def: spectral shift}; it is called {\it asymptotic module}. Set $\CW_{d,x}^{(N)} := S(\Bw_{d,x}^{(N)})$ and $\CW_{s,x} := \CW_{x,0}^{(s)}$ for $1\leq s \leq N$.

\begin{cor}  \label{cor: highest weight in O}
\begin{itemize}
\item[(i)] $\CR$ is the set of rational e-weights. 
\item[(ii)] For any $\CE$-module $M$ in category $\BGG$, we have $\ewt(M) \subset \CR$.
\item[(iii)] For $d,x \in \BC$ and $1\leq r \leq N$ we have in $\CMt$ and $K_0(\BGG)$ respectively
\begin{equation}  \label{equ: sov}
\qc(\CW_{d,x}^{(r)}) = \Bw_{d,x}^{(r)} \times \qc(\CW_{0,x}^{(r)}),\quad [\CW_{d,0}^{(r)}][\CW_{0,x}^{(r)}] = [\CW_{d-x,x}^{(r)}][\CW_{x,0}^{(r)}].
\end{equation}
\end{itemize}
\end{cor}
\begin{proof}
(iii) comes from Eq.\eqref{equ: q-char asym}, as in the proof of \cite[Theorem 3.11]{FZ}. 
 $\Bw_{d,x}^{(r)}$ as a highest weight of $\CW_{d,x}^{(r)}$ belongs to $\CR$. Together with Lemma \ref{lem: highest weight condition} we obtain (i). In (ii)  one may assume $M$ irreducible. Then $M$ is a sub-quotient of a tensor product of asymptotic modules. Since e-weights of an asymptotic module are rational, we conclude from the multiplicative structure of q-characters in Proposition \ref{prop: monoidal O}. 
\end{proof}

In Section \ref{ss: small} the evaluation module $V_{\mu}(x)$ is an irreducible highest weight module in category $\tBGG$. Its highest weight is easily shown to be rational. 

\begin{cor} \label{cor: evaluation O}
 $V_{\mu}(x)$ is in category $\BGG$ for $\mu \in \Hlie$ and $x \in \BC$.
\end{cor}

 Finite-dimensional modules in category $\BGG$ are related to the asymptotic modules by {\it generalized Baxter relations} in the sense of Frenkel--Hernandez \cite[Theorem 4.8]{FH}; see \cite[Corollary 4.7]{FZ} and \cite[Theorem 5.11]{Z3} for a closer situation.

\begin{theorem}  \label{thm: generalized TQ}
Let $V$ be a finite-dimensional $\CE$-module in category $\BGG$. Then 
\begin{equation}  \label{equ: generalized TQ}
[V] = \sum_{j=1}^{\dim V} [S(\Bd_j)] \times \Bm_j 
\end{equation}
in a fraction ring of the Grothendieck ring of $\BGG$. Here $\Bd_j \in \CR_0$ and $\Bm_j$ is a product of the $\frac{[\CW_{r,x}]}{[\CW_{r,y}]}$ with $x,y\in \BC$ and $1\leq r < N$.
\end{theorem}
\begin{proof}
The idea is the same as \cite{FH}. Since the q-character map is injective, one can replace isomorphism classes with q-characters. $\qc(V)$ is the sum of its e-weights, the number of which is $\dim V$.  By Corollary \ref{cor: highest weight in O}, any e-weight $\Be$ is of the form $\Bd \prod \frac{\Psi_{r,x}}{\Psi_{r,y}} =  \qc(S(\Bd))  \prod \frac{\qc(\CW_{r,x})}{\qc(\CW_{r,y})}$, where $\Bd \in \CR_0$ and the product is over $1\leq r < N$ and $x,y \in \BC$. This proves Eq.\eqref{equ: generalized TQ} in terms of q-characters.
\end{proof}
To compare with \cite[Theorem 4.8]{FH}, one imagines that for $1\leq r < N$ and $x \in \BC$ there existed a {\it positive prefundamental module} $L_{r,x}^+$ in category $\BGG$ with q-character $\qc(L_{r,x}^+) = \Psi_{r,x} \times \chi(L_{r,0}^+)$ as in \cite[Theorem 4.1]{FH}. Then $\frac{[\CW_{r,x}]}{[\CW_{r,y}]} = \frac{[L_{r,x}^+]}{[L_{r,y}^+]}$. Note that the q-character of $\CW_{0,x}^{(r)}$ in Eq.\eqref{equ: sov} is different from its character. 

\begin{example} \label{exam: sl3}
Let $N = 3$. Consider the vector representation $\BV$ of Section \ref{ss: vector representation}:
\begin{gather*}
\boxed{1}_0  = \frac{\Psi_{1,\frac{3}{2}}}{\Psi_{1,\frac{1}{2}}}, \quad \boxed{2}_0 = \frac{\Psi_{1,-\frac{1}{2}}}{\Psi_{1,\frac{1}{2}}} \frac{\Psi_{2,1}}{\Psi_{2,0}}, \quad \boxed{3}_0 = \frac{\theta(z+\hbar)}{\theta(z)} \frac{\Psi_{2,-1}}{\Psi_{2,0}}, \\
[\BV] = \frac{[\CW_{1,\frac{3}{2}}]}{[\CW_{1,\frac{1}{2}}]} +  \frac{[\CW_{1,-\frac{1}{2}}]}{[\CW_{1,\frac{1}{2}}]}  \frac{[\CW_{2,1}]}{[\CW_{2,0}]} +\frac{[\CW_{3,\frac{1}{2}}]}{[\CW_{3,-\frac{1}{2}}]} \frac{[\CW_{2,-1}]}{[\CW_{2,0}]}.
\end{gather*}
\end{example}
\begin{example}
Let us construct the $\CE_{\tau,\hbar}(\mathfrak{sl}_2)$-module $\CW_{1,\Lambda}$ from  \cite[Theorem 3]{FV1}. In {\it loc.cit.}, set $\eta = -\frac{1}{2} \hbar,\ \lambda = \lambda_{12}$ and $(a,b,c,d) = (L_{11},L_{12},L_{21},L_{22})$. For $\Lambda \in \BZ_{>0}$, consider the evaluation module $L_{\Lambda}((\Lambda-1)\eta)$ with basis $(e_k)_{0\leq k \leq \Lambda}$. Note that $k$ indicates the basis vectors, while $\Lambda$  the integer parameter of a KR module. Let us make a change of basis (the second product is empty if $k=0$)
$$ v_k := e_k \prod_{i=1}^{\Lambda}\frac{\theta(\lambda+(i-k)\hbar)}{\theta(\hbar)} \times \prod_{j=1}^k \frac{\theta(\lambda-j\hbar)}{\theta((\Lambda-k+j)\hbar)} \quad \mathrm{for}\ 0 \leq k \leq \Lambda.  $$
Tensoring $L_{\Lambda}((\Lambda-1)\eta)$ with the one-dimensional module of highest weight $\frac{\theta(w+\Lambda\hbar)}{\theta(w)}$, we obtain another irreducible module $V_{\Lambda}$ with basis $(v_k  = v_k \wtimes 1)_{0\leq k \leq \Lambda}$; here  to follow {\it loc.cit.} $w$ denotes $z$. We have $\wt(v_k) = (\Lambda-k)\epsilon_1 + k \epsilon_2$ and
\begin{gather*}
a(w) v_k = \frac{\theta(w+(\Lambda-k)\hbar)}{\theta(w)} \frac{\theta(\lambda+(\Lambda-k+1)\hbar)}{\theta(\lambda+(1-k)\hbar)} v_k, \\
b(w) v_k = \frac{\theta(w+\lambda+(\Lambda-k-1)\hbar)}{\theta(w)} \frac{\theta((\Lambda-k)\hbar)\theta(\hbar)}{\theta(\lambda-\hbar)\theta(\lambda)} v_{k+1}, \\
c(w) v_k = - \frac{\theta(w-\lambda+(k-1)\hbar)}{\theta(w)} \frac{\theta(k\hbar)}{\theta(\hbar)} v_{k-1}, \\
d(w) v_k = \frac{\theta(w+ k\hbar)}{\theta(w)} \frac{\theta(\lambda-(k+1)\hbar)\theta(\lambda-k\hbar)}{\theta(\lambda-\hbar)\theta(\lambda)} v_k.
\end{gather*}
We have $t_{12} v_k = -\frac{\theta(k\hbar)}{\theta(\hbar)}  v_{k-1}$ and $v_0$ is of highest weight $\Bw_{\Lambda,0}^{(1)}$,
so $V_{\Lambda}\cong W_{\Lambda,0}^{(1)}$. The bases $(v_k)$ trivialize the inductive system $(V_{\Lambda})$ because the inductive maps commute with $t_{12}$ by Eq.\eqref{inductive: L > r}.  For $\Lambda \in \BC$, the above formulas define an $\CE_{\tau,\hbar}(\mathfrak{sl}_2)$-module structure on $\oplus_{k=0}^{\infty} \BM v_k$, with $\wt(v_k) = (\Lambda-k)\epsilon_1+k\epsilon_2$. This is the desired $\CW_{1,\Lambda}$. 

General formulas for the $\CE_{\tau,\hbar}(\mathfrak{sl}_N)$-module $\CW_{1,\Lambda}$ can be found in \cite[\S 3.4]{C}.
\end{example}

\section{Baxter TQ relations} \label{sec: TQ}
We derive three-term relations in the Grothendieck ring $K_0(\BGG)$ for the asymptotic modules. For $1\leq r < N$ and $k,x,t \in \BC$, by Corollary \ref{cor: highest weight in O}, $\Bd_{k,x}^{(r,t)} \in \CR$ and is the highest weight of an irreducible module $D_{k,x}^{(r,t)}$ in category $\BGG$.

Call a complex number $c \in \BC$ {\it generic} if $c \notin \frac{1}{2} \BZ + \frac{1}{\hbar} (\BZ + \BZ \tau)$.
This condition is equivalent to $\CQ_a \cap \CQ_{a+c} = \{1\}$ for all $a \in \BC$.

\begin{theorem} \label{thm: Baxter TQ}
Let $1\leq r < N,\ t \in \BZ_{>0}$ and $k,a,b \in \BC$ with $k$ generic. Then 
\begin{equation}  \label{for: asym Demazure}
\qc(D_{k,a}^{(r,t)}) = \Bd_{k,a}^{(r,t)} (1+ \sum_{l=1}^t A_{r,a}^{-1} A_{r,a+1}^{-1} \cdots A_{r,a+l-1}^{-1}) \prod_{s= r\pm 1} \qc(\CW_{0,a-k-\frac{1}{2}}^{(s)})
\end{equation}
and $D_{k,a}^{(r,0)} \cong \wtimes_{s= r\pm 1} \CW_{k,a-k-\frac{1}{2}}^{(s)}$. 
\end{theorem}
\begin{proof}
Set $x := a-k-\frac{1}{2}$. Define $\Bd := \Bd_{k,a}^{(r,t)}$ and  for $1\leq l \leq t$: 
$$ \Bm_l := \Bm_0 A_{r,a}^{-1} A_{r,a+1}^{-1} \cdots A_{r+a+l-1}^{-1}, \quad \Bm_0 := \Bd \prod_{j=1}^t \Bw_{k+j-\frac{1}{2},x}^{(r)}.  $$
By Eqs.\eqref{equ: A Psi} and \eqref{equ: asymp e-weight}--\eqref{equ: Demazure e-weight}, we have for $0 \leq l \leq t$:
$$ \Bm_l = \prod_{j=l+1}^t \frac{\Psi_{r,a+j}}{\Psi_{r,x}} \times \prod_{j=1}^l \frac{\Psi_{r,a+j-2}}{\Psi_{r,x}} \times \prod_{s= r\pm 1} \frac{\Psi_{s,a+l-\frac{1}{2}}}{\Psi_{s,x}}. $$
Let us introduce the tensor products for $0 \leq l \leq t$,
\begin{align*}
     S^l &:= (\wtimes_{j=l+1}^t \CW_{k+j+\frac{1}{2},x}^{(r)})\ \wtimes\ (\wtimes_{j=1}^l \CW_{k+j-\frac{3}{2},x}^{(r)})\ \wtimes\ (\wtimes_{s= r\pm 1} \CW_{k+l,x}^{(s)}), \\
     T & := D_{k,a}^{(r,t)}\ \wtimes\ (\wtimes_{j=1}^t \CW_{k+j-\frac{1}{2},x}^{(r)}).
\end{align*}
Eq.\eqref{for: asym Demazure} is equivalent to  $\qc(T) = \sum_{l=0}^t\qc(S^l)$ in view of Eq.\eqref{equ: sov}. 

Given two elements $\chi = \sum_{\Bf} c_{\Bf} \Bf$ and $\chi' := \sum_{\Bf} c_{\Bf}' \Bf$ of $\CMt$, we say that $\chi$ is bounded above by $\chi'$ if $c_{\Bf} \leq c_{\Bf}'$ for all $\Bf \in \CMw$. When this is the case, $\chi'$ is bounded below by $\chi$. If $\chi$ is bounded below and above by $\chi'$, then $\chi = \chi'$.

\medskip

\noindent {\it Claim 1.} The $S^l$ are irreducible. In particular, $D_{k,a}^{(r,0)} \cong \wtimes_{s= r\pm 1} \CW_{k,x}^{(s)}$.

Fix $0 \leq l \leq t$. Let $S' := S(\Bm_l)$. For $n \in \BZ_{>0}$, set
$$ S_n' := (W_{n,x}^{(r)})^{\wtimes t}\ \wtimes\ (\wtimes_{s= r\pm 1} W_{n,x}^{(s)}),\quad \Bs_n' := (\Bw_{n,x}^{(r)})^t \prod_{s= r\pm 1} \Bw_{n,x}^{(s)}. $$
By Lemma \ref{lem: KR q-char}, any e-weight $\Bs_n'\Be \in \CP_x$ of $S_n'$ different from $\Bs_n'$ is right negative. So $S_n'$ is irreducible. Viewing $S_n'$ as an irreducible sub-quotient of
$$ S'\ \wtimes\ (\wtimes_{j=l+1}^t \CW_{n-k-j-\frac{1}{2},a+j}^{(r)})\ \wtimes\ (\wtimes_{j=1}^l \CW_{n-k-j+\frac{3}{2},a+j-2}^{(r)})\ \wtimes\ (\wtimes_{s= r\pm 1} \CW_{n-k-l,a+l-\frac{1}{2}}^{(s)}), $$
we have $\Be = \Be' \prod_{j=1}^t  \Be_j \prod_{s = r\pm 1} \Be^{(s)}$ where
$\Bm_l \Be', \ \Bw_{n-k-j-\frac{1}{2},a+ j}^{(r)} \Be_j$  for $l<j\leq t$,  $\Bw_{n-k-j+\frac{3}{2},a+j-2}^{(r)} \Be_j$ for $1\leq j \leq l$, and $\Bw_{n-k-l,a+l-\frac{1}{2}}^{(s)}\Be^{(s)} $ are e-weights of the corresponding tensor factors. By Lemma \ref{lem: KR q-char} and Proposition \ref{prop: asymptotic representations},
$$\Be,  \Be' \in  \CQ_x^-,\quad \Be_j, \Be^{(s)} \in   \CQ_a^-.  $$
Since $a-x= k+\frac{1}{2}$ is generic,  $\CQ_{a}^-\cap \CQ_x^- = \{1\}$ and so $ \Be = \Be'$. The normalized q-character of $S'$ is bounded below by that of $S_n'$ for all $n \in \BZ_{>0}$. On the other hand, viewing $S'$ as an irreducible sub-quotient of $S^l$ and applying Eq.\eqref{equ: q-char asym} to $S^l$, we see that the normalized q-character of $S'$ is bounded above by the limit of that of $S_n'$ as $n \rightarrow \infty$. Therefore $S^l \cong S'$ is irreducible.

\medskip

\noindent {\it Claim 2.} For $1\leq l \leq t$, we have $ \Bd A_{r,a}^{-1} A_{r,a+1}^{-1} \cdots A_{r,a+l-1}^{-1}  \in \ewt(D_{k,a}^{(r,t)})$. It follows that $\Bm_l \in \ewt(T)$.

Let us view the KR module $W_{t,a}^{(r)}$ as an irreducible sub-quotient of
$$ D_{k,a}^{(r,t)}\ \wtimes\ (\wtimes_{s= r\pm 1} \CW_{-k,a-\frac{1}{2}}^{(s)}). $$
By Lemma \ref{lem: KR q-char}, $\Bw_{t,a}^{(r)} A_{r,a}^{-1} A_{r,a+1}^{-1} \cdots A_{r,a+l-1}^{-1} \in \ewt(W_{t,a}^{(r)})$. The $A_{r,a+j}^{-1}$ must arise from $\ewt(D_{k,a}^{(r,t)})$ instead of any of the $\ewt(W^{(s)}_{-k,a-\frac{1}{2}})$ with $s \neq r$.

\medskip

For $0 \leq j,l \leq t$, since $\ewt(S^l) \subset \Bm_l \CQ_x^-$ and $\Bm_j \in \Bm_l \CQ_a$, we have $\Bm_j \in \ewt(S^l)$ if and only if $l = j$. Therefore, all the $S^l$ appear as irreducible sub-quotients of $T$, and they are mutually non-isomorphic. So $\qc(T) $ is bounded below by $\sum_{l=0}^t\qc(S^l)$. 

\medskip 

\noindent {\it Claim 3.} $\qc(D_{k,a}^{(r,t)})$ is bounded above by 
$$\Bd (1+\sum_{l=1}^t A_{r,a}^{-1} A_{r,a+1}^{-1} \cdots A_{r+a+l-1}^{-1}) \prod_{s= r\pm 1} \qc(\CW_{0,x}^{(s)}). $$

 Fix $\Bd \Bf \in \ewt(D_{k,a}^{(r,t)})$. For $n \in \BZ_{>0}$,  viewing $D_{k,a}^{(r,t)}$ as a sub-quotient of
 $$ D_{n,a}^{(r,t)}\ \wtimes\ (\wtimes_{s= r\pm 1} \CW_{k-n,x}^{(s)} ) $$
 gives $\Bf = \Bf_n \prod_{s= r\pm 1} \Bf^{(s)}$ where by Lemma \ref{lem: KR q-char} and Corollary \ref{cor: Demaure q-char}: 
 $$\Bf_n \Bd_{n,a}^{(r,t)} \in \ewt(D_{n,a}^{(r,t)}),\quad \Bf^{(s)} \Bw_{k-n,x}^{(s)} \in \ewt(\CW_{k-n,x}^{(s)}) = \Bw_{k-n,x}^{(s)} \ewt(\CW_{0,x}^{(s)}). $$
 It follows that $\Bf_n \in \CQ_a^-, \ \Bf^s \in \CQ_x^-$ and $\Bf \in \CQ_a^- \CQ_x^-$. 

 \medskip
 
Let $n \in \BZ_{>0}$ be large enough so that $\Bf \in \CQ_{a;n}^- \CQ_x^-$ where $\CQ_{a;n}^-$ is the submonoid of $\CQ_a^-$ generated by the  $A_{i,a+m}^{-1}$ for $1\leq i < N$ and $m \in \frac{1}{2}\BZ$ with $m > -n$. Since $a-x = k+\frac{1}{2}$ is generic, Corollary \ref{cor: Demaure q-char} implies that 
$$\Bf_n \in \{1,A_{r,a}^{-1}, A_{r,a}^{-1}A_{r,a+1}^{-1}, \cdots, A_{r,a}^{-1} A_{r,a+1}^{-1} \cdots A_{r+a+t-1}^{-1}\}$$
is uniquely determined by $\Bf$. The coefficient of $\Bd \Bf$ in $\qc(M_{k,a}^{(r)})$ is bounded above by that of $\prod_{s= r\pm 1} \Bf^{(s)} $ in $\prod_{s= r\pm 1} \qc(\CW_{0,x}^{(s)})$. This proves the claim.

\medskip 

It follows from Claim 3 that $\qc(T)$ is bounded above by 
$$\Bm_0 (1+\sum_{l=1}^t A_{r,a}^{-1} A_{r,a+1}^{-1} \cdots A_{r+a+l-1}^{-1}) \prod_{s= r\pm 1} \qc(\CW_{0,x}^{(s)}) \times \prod_{j=1}^t \qc(\CW_{0,x}^{(r)}) = \sum_{l=0}^t \qc(S^l). $$ 
Since ``bounded below" also holds, we obtain the exact formula for $\qc(T)$, which implies Eq.\eqref{for: asym Demazure}. This completes the proof of the theorem.
\end{proof}

Claim 1 is in the spirit of \cite[Theorem 4.11]{FH}, and Claim 3 \cite[Eq.(6.14)]{HL}, \cite[\S 4.3]{FH2} and \cite[Theorem 3.3]{Z2}, the main difference being the non-existence of prefundamental modules. If both $k,t$ are generic, then $\qc(D_{k,a}^{(r,t)})$ is obtained from the right-hand side of Eq.\eqref{for: asym Demazure} by replacing $\sum_{l=1}^t$ therein with $\sum_{l=1}^{\infty}$.

\begin{cor}  \label{cor: TQ}
Let $k \in \BC$ be generic and $1 \leq r < N$. In $K_0(\BGG)$ holds
\begin{equation} \label{equ: TQ three term}
 [D_{k,k+\frac{1}{2}}^{(r,1)}][\CW_{r, k+\frac{1}{2}}] = [\CW_{r,k-\frac{1}{2}}] \prod_{s= r\pm 1} [\CW_{s,k+1}] + [\CW_{r,k+\frac{3}{2}}] \prod_{s= r\pm 1} [\CW_{s,k}].  
\end{equation}
\end{cor}
\begin{proof}
From Eq.\eqref{for: asym Demazure} and the injectivity of the q-character map we obtain
\begin{equation}  \label{equ: asymptotic Baxter}
 [D_{k,a}^{(r,t)}][\CW_{a-b+t-1,b}^{(r)}] = [D_{k+t,a+t}^{(r,0)}][\CW_{a-b-1,b}^{(r)}] + [D_{k,a}^{(r,t-1)}][\CW_{a-b+t,b}^{(r)}]
\end{equation}
for $a,b \in \BC$ and $t \in \BZ_{>0}$. Eq.\eqref{equ: TQ three term} is the special case $(t,a,b) = (1,k+\frac{1}{2},0)$ of this identity in view of the tensor product decomposition of $D_{k,a}^{(r,0)}$ in Theorem \ref{thm: Baxter TQ}.
\end{proof}
 Eq.\eqref{equ: asymptotic Baxter} can be viewed as a generic version of Eq.\eqref{equ: Demazure T-system}.

\section{Transfer matrices and Baxter operators} \label{sec: Q}
We have obtained three types of identities Eq.\eqref{equ: sov}, \eqref{equ: generalized TQ}, and \eqref{equ: TQ three term} in the Grothendieck ring $K_0(\BGG)$. These are viewed as universal functional relations \cite{BLZ2,BLZ3,BT} in the sense that when specialized to quantum integrable systems they imply functional relations of transfer matrices. In this section, we study one such example, with the quantum space being a tensor product of vector representations \cite{HSY}. 

Fix $\ell := N\kappa$ with $\kappa \in \BZ_{>0}$ and  $a_1,a_2,\cdots, a_{\ell} \in \BC \setminus \Gamma$. Set $I := \{1,2,\cdots,N\}$.  Let $I_0^{\ell}$ be the subset of $I^{\ell}$ formed of $\underline{i}$ such that $\epsilon_{i_1} + \epsilon_{i_2} + \cdots + \epsilon_{i_{\ell}} = 0 \in \Hlie$. Upon identification $\underline{i} := v_{i_1} \wtimes v_{i_2} \wtimes \cdots \wtimes v_{i_{\ell}}$, the weight space $\BV^{\wtimes \ell}[0]$ has basis $I_0^{\ell}$. 

Let $\BD_p$ be the set of formal sums $\sum_{\alpha \in \Hlie} p^{\alpha} T_{\alpha}  f_{\alpha}(z;\lambda)$ such that: the $f_{\alpha}(z;\lambda)$ are meromorphic functions of $(z,\lambda) \in \BC \times \Hlie$;  the set $\{\alpha: f_{\alpha} \neq 0 \}$ is contained in a finite union of cones $\nu + \BQ_-$ with $\nu \in \Hlie$. Make $\BD_p$ into a ring: addition is the usual one of formal sums; multiplication is induced from 
\begin{equation} \label{equ: difference ring}
p^{\alpha} T_{\alpha} f(z;\lambda)  \times p^{\beta} T_{\beta} g(z;\lambda) = p^{\alpha+\beta} T_{\alpha+\beta} f(z;\lambda+\hbar \beta) g(z;\lambda). 
\end{equation}
As in \cite{FV2,FZ}, we construct a ring morphism $[X] \mapsto t_X(z)$ from $K_0(\BGG)$ to the ring $\mathrm{M}(I_0^{\ell};\BD_p)$ of $I_0^{\ell} \times I_0^{\ell}$ matrices with coefficients in $\BD_p$. (We think of $\mathrm{M}(I_0^{\ell};\BD_p)$ as a ring of formal difference operators on $\BV^{\wtimes \ell}[0]$.)

Let $X$ be an object of category $\BGG$. To $\underline{i}, \underline{j} \in I_0^{\ell}$ we associate
$$ L_{\underline{i}\underline{j}}^X(z) := L_{i_1j_1}^X(z+a_1) L_{i_2j_2}^X(z+a_2) \cdots L_{i_{\ell}j_{\ell}}^X(z+a_{\ell}) \in (D_X)_{0,0}. $$ 
Since $(D_X)_{0,0} \subseteq \End_{\BM}(X)$, one can take trace of $L_{\underline{i}\underline{j}}^X(z)$ over weight spaces of $X$.

\begin{defi} \label{def: transfer matrix}
The {\it transfer matrix} associated to an object $X$ in category $\BGG$ is the matrix $t_X(z) \in \mathrm{M}(I_0^{\ell};\BD_p)$ whose $(\underline{i},\underline{j})$-th entry for $\underline{i}, \underline{j} \in I_0^{\ell}$ is 
$$  \sum_{\alpha\in \wt(X)} p^{\alpha} T_{\alpha} \times  \mathrm{Tr}_{X[\alpha]}\left(L_{\underline{i}\underline{j}}^X(z)|_{X[\alpha]}\right) \in \BD_p.  $$
\end{defi}

Almost all of the results and comments in \cite[\S 5]{FZ} hold true after slight modification in our present situation. In the following, we focus on the modification of these results, referring to \cite{FZ} for their proofs.

We remark that $t_{\BV}(z)|_{p=1}$ can be identified with the transfer matrix $T(z)$ in \cite[Eq.(2.22)]{HSY} where the $E_{\tau,\eta}(sl_n)$-module $W$ is $V_{\Lambda_1}(a_1) \otimes V_{\Lambda_1}(a_2) \otimes \cdots \otimes V_{\Lambda_1}(a_{\ell})$.

The transfer matrix associated to the one-dimensional module of highest weight $g(z) \in \BM_{\BC}^{\times}$ is the scalar matrix $\prod_{i=1}^{\ell} g(z+a_i)$.

For $1\leq r \leq N$ and $x \in \BC$, consider the $\CE$-module $\CW_{r,x}' := \CW_{r,x} \wtimes S(\theta(z-\ell_r\hbar))$ in category $\BGG$. By Lemma \ref{lem: asymptotic property}, the matrix entries of the difference operators $L_{ij}(z)$ for $1\leq i,j \leq N$, with respect to any basis of $\CW_{r,x}'$, are entire functions of $z \in \BC$.
\begin{defi} \label{definition: Baxter Q operator}
The  $r$-th {\it Baxter Q-operator} for $1\leq r \leq N$ is defined to be 
\begin{equation}  \label{equ: Q operators}
Q_r(u) := t_{\CW_{r,u\hbar^{-1}}'}(z)|_{z=0} \quad \mathrm{for}\ u \in \BC.
\end{equation}
\end{defi} 
Since $\CW_{N,x}' = S(\theta(z+(x+\frac{1}{2})\hbar))$ is one-dimensional, $Q_N(z) = \prod_{i=1}^{\ell} \theta(z+a_i+\frac{1}{2}\hbar)$.

Let $1\leq r < N$. Then $Q_r(z) = p^{z\hbar^{-1} \varpi_r} T_{z\hbar^{-1}\varpi_r} \widetilde{Q}(z)$ and $\widetilde{Q}(z)$ is a power series in the $p^{-\alpha_i}T_{-\alpha_i}$ for $1\leq i < N$. The leading term $\widetilde{Q}_0(z)$ of $\widetilde{Q}(z)$ is invertible. Indeed $\widetilde{Q}_0(0)$ is the scalar matrix $\prod_{j=1}^{\ell}\theta(a_j) \in \mathrm{M}(I_{\ell};\BC)$, which is invertible because $\theta(a_j) \neq 0$ by assumption. (One can prove furthermore that with respect to certain order on $I_0^{\ell}$, the matrix $\widetilde{Q}_0(z)$ is upper triangular, whose entries are meromorphic functions of $(z;\lambda) \in \BC \times \Hlie$ and entire on $z$.) Therefore $Q_r(z) \in \mathrm{GL}(I_0^{\ell};\BD_p)$.

Similarly one can show that $t_{\CW_{r,x}'}(z)$ is invertible for $x \in \BC$.

\begin{prop} \label{prop: transfer matrices}
Let $X,Y $ be in category $\BGG$ and let $x,u \in \BC$. 
\begin{itemize}
\item[(i)] $t_{\Psi_u^*X}(z) = t_X(z+u\hbar)$.
\item[(ii)] $t_X(z) t_Y(z) = t_{X\wtimes Y}(z)$.
\item[(iii)] $t_{\CW_{r,x}}(z) t_{\CW_{r,0}}(z+u\hbar) = t_{\CW_{r,x-u}}(z+u\hbar) t_{\CW_{r,u}}(z)$. 
\item[(iv)] $t_X(z)t_Y(w) = t_Y(w)t_X(z)$.
\end{itemize}
\end{prop}
In (iv), we replace one of the $z$ in Eq.\eqref{equ: difference ring} with $w$ to define the multiplication. It is proved as in \cite[Theorem 5.3]{FH}:  the commutativity of transfer matrices is a consequence of the commutativity of the Grothendieck ring $K_0(\BGG)$. The standard proof by using the Yang--Baxter equation \cite{Baxter72} would require braiding in category $\BGG$, whose existence is not clear.

 (ii) and the fact that $t_X(z)$ only depends on the isomorphism class $[X]$ of $X$ imply that $[X] \mapsto t_X(z)$ is a ring homomorphism $\mathrm{tr}_p: K_0(\BGG) \longrightarrow \mathrm{M}(I_0^{\ell};\BD_p)$. Applying $\mathrm{tr}_p$ to Eq.\eqref{equ: sov} we obtain (iii). Replace $(\CW, x,u,z)$ with $(\CW', z\hbar^{-1}+x,z\hbar^{-1},0)$ in (iii) and take the inverse of $Q_r(z)$ and $t_{\CW_{r,0}'}(z)$. We have
\begin{equation}  \label{equ: generalized TQ transfer}
\frac{t_{\CW_{r,x}}(z)}{t_{\CW_{r,0}}(z)} = \frac{t_{\CW_{r,x}'}(z)}{t_{\CW_{r,0}'}(z)} = \frac{Q_r(z+x\hbar)}{Q_r(z)}
\end{equation}
as in \cite[Theorem 5.6 (i)]{FZ}. Now applying $\mathrm{tr}_p$ to Eq.\eqref{equ: generalized TQ}, we obtain
\begin{cor} \label{cor: generalized TQ transfer}
Let $V$ be a finite-dimensional $\CE$-module in category $\BGG$. Then in Eq.\eqref{equ: generalized TQ} replacing $V, \ S(\Bd_j)$ and the $\frac{[\CW_{r,a}]}{[\CW_{r,b}]}$ with $t_V(z), \ t_{S(\Bd_j)}(z)$ and $\frac{Q_r(z+a\hbar)}{Q_r(z+b\hbar)}$ respectively, we obtain an identity in $\mathrm{M}(I_0^{\ell};\BD_p)$.  
\end{cor}
This forms the generalized Baxter relations for transfer matrices. If the prefundamental modules $L_{r,a}^+$ before Example \ref{exam: sl3} existed, then we would have defined alternatively the $r$-th Baxter operator $Q_r^{\mathrm{FH}}(z) = t_{L_{r,0}^+}(z)$ as a {\it real} transfer matrix \cite[\S 5.5]{FH} and so $\frac{Q_r(z+a\hbar)}{Q_r(z+b\hbar)} = \frac{Q_r^{\mathrm{FH}}(z+a\hbar)}{Q_r^{\mathrm{FH}}(z+b\hbar)}$ based on $\frac{[\CW_{r,a}]}{[\CW_{r,b}]} = \frac{[L_{r,a}^+]}{[L_{r,b}^+]}$.

As an illustration of the corollary, let us be in the situation of Example \ref{exam: sl3}:
$$t_{\BV}(z) = \frac{Q_1(z+\frac{3}{2}\hbar)}{Q_1(z+\frac{1}{2}\hbar)} +  \frac{Q_1(z-\frac{1}{2}\hbar)}{Q_1(z+\frac{1}{2}\hbar)}  \frac{Q_2(z+\hbar)}{Q_2(z)} +  \frac{Q_2(z-\hbar)}{Q_2(z)} \prod_{j=1}^{\ell}\frac{\theta(z+a_j+\hbar)}{\theta(z+a_j)} . $$

Apply $\mathrm{tr}_p$ to Eq.\eqref{equ: TQ three term}, divide both sides by the second term, and then perform a change of variable $z + (k+\frac{1}{2})\hbar \mapsto w$. By Eq.\eqref{equ: generalized TQ transfer} and Proposition \ref{prop: transfer matrices} (i):
\begin{equation}   \label{equ: TQ transfer matrices}
X_k^{(r)}(w) \frac{Q_r(w)}{Q_r(w-\hbar)} = 1 + \frac{Q_r(w+\hbar)}{Q_r(w-\hbar)} \prod_{s=r\pm 1} \frac{Q_s(w-\frac{1}{2}\hbar)}{Q_s(w+\frac{1}{2}\hbar)}.
\end{equation}
This forms three-term Baxter TQ relations for transfer matrices, where
$$ X_k^{(r)}(w) = t_{D_{k,0}^{(r,1)}}(w) \prod_{s=r\pm 1} t_{\CW_{s,k+1}}^{-1}(w-(k+\frac{1}{2})\hbar). $$
By Eq.\eqref{equ: TQ transfer matrices}, $X_k^{(r)}(z) \in \mathrm{M}(I_0^{\ell};\BD_p)$ is independent of the choice of generic $k \in \BC$.

In the homogeneous case $a_1 = a_2 = \cdots = a_{\ell} = a$, the entries of the matrix $\widetilde{Q}_r(z)$, as entire functions of $z$, in general do not satisfy the uniform double periodicity of \cite[Theorem 5.6(ii)]{FZ}. By ``uniform" we mean the multipliers with respect to $z + 1$ and $z + \tau$ only depend on $(a, z, \ell)$. This is because the transfer matrix construction in \cite{FZ} is based on a slightly different elliptic quantum group; see Footnote \ref{ft: different R}.  

We follow \cite[\S 5]{FH2} to derive the Bethe Ansatz equations from Eq.\eqref{equ: TQ transfer matrices}. Let $u$ be a zero of $Q_r(z)$. Suppose $X_k^{(r)}(z), Q_r(z-\hbar), Q_s(z+\frac{1}{2}\hbar)$ for $s \neq r\pm 1$ have no poles at $z = u$. (This is a genericity condition.) Then as in \cite[Eq.(5.16)]{FH2}:
\begin{equation}  \label{equ: BAE}
\frac{Q_r(u+\hbar)}{Q_r(u-\hbar)} \prod_{s=r\pm 1} \frac{Q_s(u-\frac{1}{2}\hbar)}{Q_s(u+\frac{1}{2}\hbar)} = -1.
\end{equation}
To compare with \cite{FH2}, we can assume furthermore that eigenvalues of $Q_r(z)$ are of the form  $p^{z \hbar^{-1} \varpi_r} \prod_{i=1}^{d_r} \theta(z - u_{r;i})$ based on \cite[Remark 5.8]{FZ}. Then 
$$ p^{\alpha_r} \prod_{i=1}^{d_r} \frac{\theta(u_{r;k}+\hbar - u_{r;i})}{\theta(u_{r;k}-\hbar - u_{r;i})} \prod_{s=r\pm 1} \prod_{j=1}^{d_s} \frac{\theta(u_{r;k}+\frac{1}{2}\hbar - u_{s;j})}{\theta(u_{r;k}+\frac{1}{2}\hbar - u_{s;j})} = -1\quad \mathrm{for}\ 1 \leq k \leq d_r.  $$
We remark that similar Bethe Ansatz equations for $\CE$ appeared in \cite[Eq.(3.45)]{HSY}.

For affine quantum groups and toroidal $\mathfrak{gl}_1$, the genericity condition of Bethe Ansatz equations has been dropped in \cite{Jimbo1,Jimbo2}.

\medskip

{\bf Acknowledgments.} The author thanks Giovanni Felder, David Hernandez, Bernard Leclerc, Marc Rosso and Vitaly Tarasov for fruitful discussions, and  the anonymous referees for their valuable comments and suggestions.  This work was supported by the National Center of Competence in Research SwissMAP---The Mathematics of Physics of the Swiss National Science Foundation, during the author's postdoctoral stay at ETH Z\"urich.


\end{document}